\documentclass[12pt,a4paper]{book}
\usepackage[utf8]{inputenc}
\usepackage{amsmath, amsthm, amssymb, amscd, mathrsfs,amsfonts, latexsym}
\usepackage{graphicx}
\usepackage{bbm}
\usepackage{braket}
\usepackage{hyperref}
\usepackage[titletoc]{appendix}
\usepackage{fancyhdr}
\usepackage{dsfont}
\graphicspath{{Images/}}
\usepackage[backend=bibtex,url=false,style=nature,firstinits,doi=false]{biblatex}
\usepackage{csquotes}
\usepackage{lmodern}
\usepackage{indentfirst}

\newcommand{\changefont}{%
    \fontsize{11}{11}
    }

\makeatletter
\def\cleardoublepage{\clearpage\if@twoside \ifodd\c@page\else%
    \hbox{}%
    \thispagestyle{empty}
    \newpage%
    \if@twocolumn\hbox{}\newpage\fi\fi\fi}
\makeatother

\makeatletter
\let\ps@plain=\ps@empty
\makeatother

\newcommand{\calA}{{\mathcal A}}

\newcommand{\calH}{{\mathcal H}}

\newcommand{\calS}{{\mathcal S}}

\newcommand{\euA}{{\mathscr A}}
\newcommand{\euC}{{\mathscr C}}
\newcommand{\euS}{{\mathscr{S}}}

\newcommand{\Obs}{\textbf{\textsf{Obs}}}
\newcommand{\Pos}{\textbf{\textsf{P1\hspace{0.2cm}}}}
\newcommand{\Poss}{\textbf{\textsf{P2\hspace{0.2cm}}}}
\newcommand{\Posss}{\textbf{\textsf{P3\hspace{0.2cm}}}}
\newcommand{\Loc}{\textbf{\textsf{Loc}}}
\newcommand{\Calg}{\textbf{\textsf{C*}}^{\textbf{\textsf{op}}}}
\newcommand{\Top}{\textbf{\textsf{Top}}}
\newcommand{\vN}{\textbf{\textsf{vN}}^{\textbf{\textsf{op}}}}
\newcommand{\Hys}{\textbf{\textsf{Hys}}}
\newcommand{\Lms}{\textbf{\textsf{Lms}}}

\newcommand{\calL}{{\mathcal L}}
\newcommand{\calM}{{\mathcal M}}

\newcommand{\calO}{{\mathcal O}}

\newcommand{\obj}{\mbox{\rm obj}}
\newcommand{\Hom}{\mbox{\rm hom}}
\newcommand{\trace}{\mbox{tr}}

\newcommand{\id}{\mathds{1}}

\newcommand{\lang}{\langle}
\newcommand{\rang}{\rangle}
\renewcommand{\vec}[1]{\mathbf{#1}}
\newcommand{\me}{\mathrm{e}}

\theoremstyle{plain}
\newtheorem{Thm}{Theorem}[section]
\newtheorem{Pro}[Thm]{Proposition}
\newtheorem{Lem}[Thm]{Lemma}
\newtheorem{Cor}[Thm]{Corollary}
\newtheorem{Def}[Thm]{Definition}
\theoremstyle{remark}

\bibliography{thesis.bib}

\AtEveryBibitem{
 \clearlist{address}
 \clearfield{date}
 \clearfield{eprint}
 \clearfield{isbn}
 \clearfield{issn}
 \clearfield{number}
 \clearfield{issue}
 \clearlist{location}
 \clearfield{month}
 \clearfield{series}
 \clearfield{note}
 \clearfield{language}

 \ifentrytype{book}{}{
  \clearlist{publisher}
  \clearname{editor}
 }
}

\begin{document}

\thispagestyle{empty}

\vspace*{-4.2cm}
\begin{figure}[htb!]
        \hspace{-0.4cm}\begin{minipage}[t]{5cm}
        \begin{flushleft}
					\includegraphics[height=3cm]{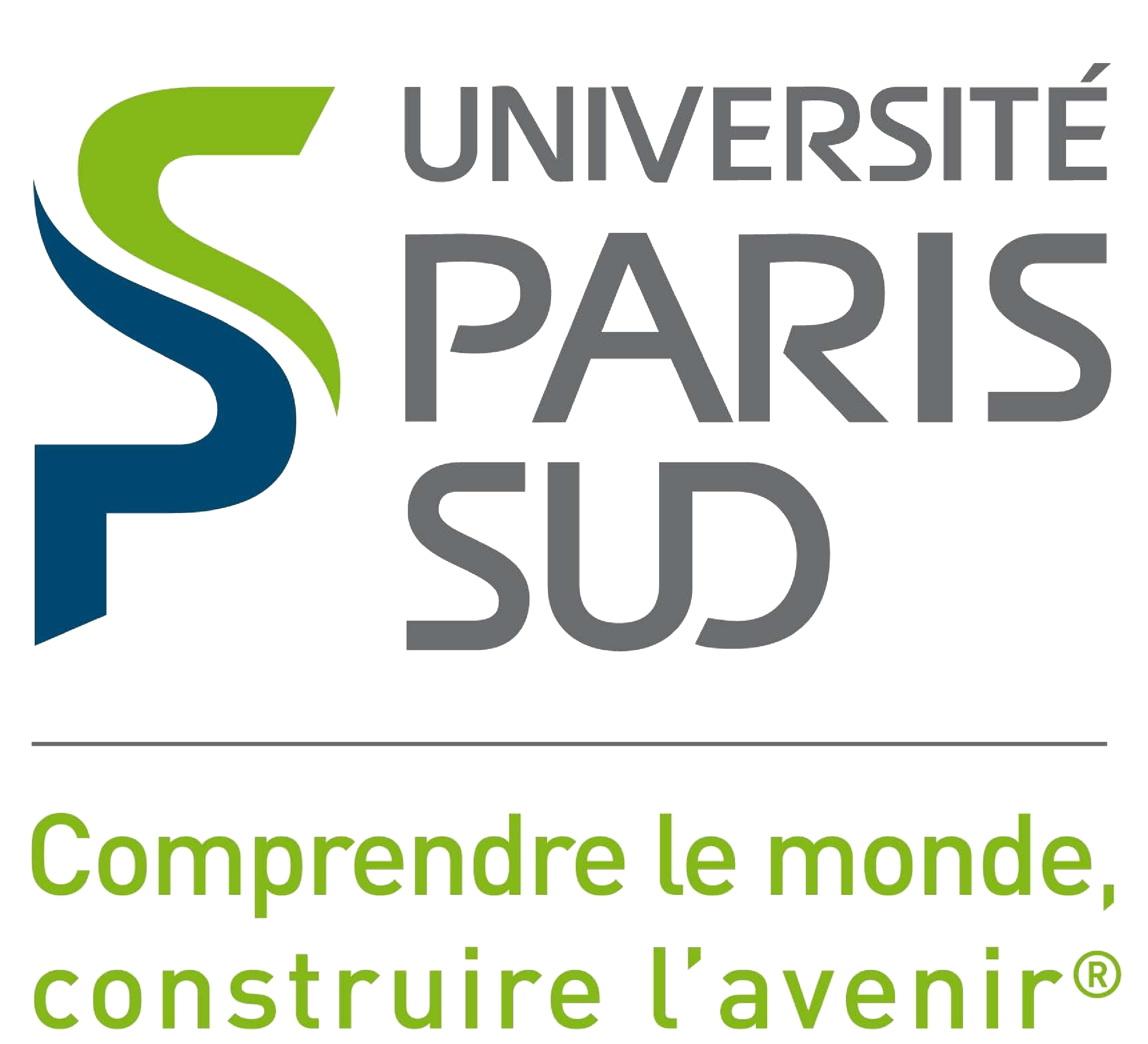}
				\end{flushleft}
				\end{minipage}
\hspace{6.6cm}
				\begin{minipage}[t]{5cm}
				\begin{flushright}
					\includegraphics[height=2cm]{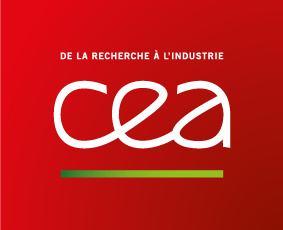}
					\includegraphics[height=2cm]{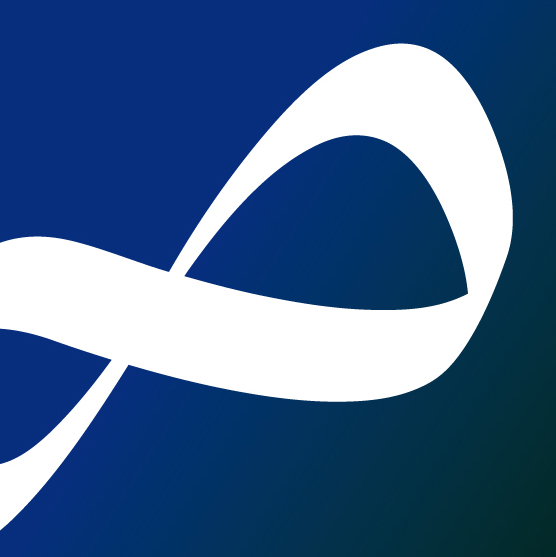}
				\end{flushright}
				\end{minipage}
\end{figure}

\begin{center}
\begin{tabular}{c}
		\\
    {\huge \textsc{Universit\'e Paris-Sud}} \\
		\\
		\\
    {\Large \textsc{Ecole Doctorale 564 :}} \\
    {\Large \textsc{Physique en Île-de-France}} \\
    \\
		{\Large \textsc{Laboratoire : CEA Saclay/IRFU/LARSIM}} \\
		\\
		{\Large \textsc{Discipline : Physique Quantique}} \\
    \\
		\\
    \huge \textsc{Th\`ese de doctorat}\\
    \\
    \large{Soutenue le 11 décembre 2014 par} \\
    \\
    \\
    \large{\textbf{{Mohamed Issam Ibnouhsein}}}\\
    \\
		\\
		\\
    \huge\bf{Corrélations quantiques}\\
		\huge\bf{et structures causales}\\
		\\
		\\
\end{tabular}

\begin{tabular}{p{0cm} p{3.5cm} p{5.4cm} l }
	& \footnotesize\bf{Directeur de th\`ese} : & Daniel Estève & \footnotesize{Directeur de recherche au CEA Saclay} \\
	& \footnotesize\bf{Encadrant} : & Alexei Grinbaum & \footnotesize{Chercheur au CEA Saclay} \\
	& & &\\
	& \footnotesize\bf\underline{Composition du jury :}& &\\
	\\
	& \footnotesize{Pr\'esident du jury} : & Nicolas Gisin & \footnotesize{Professeur à l'Université de Genève} \\
	& \footnotesize{Rapporteurs} : & Giacomo Mauro D'Ariano & \footnotesize{Professeur à l'Université de Pavie} \\
  &	& Jonathan Oppenheim& \footnotesize{Professeur à University College London} \\
  & \footnotesize{Examinateurs} : & Alexia Auffèves & \footnotesize{Chercheuse à l'Institut Néel - CNRS} \\
  &	& Frédéric Grosshans & \footnotesize{Chercheur au CNRS} \\

\end{tabular}

\end{center}

\newpage
\null
\thispagestyle{empty}
\newpage

\chapter*{Résumé}

Les travaux récents en fondements de la théorie quantique (des champs) et en information quantique relativiste tentent de mieux comprendre les effets des contraintes de causalité imposées aux opérations physiques sur la structure des corrélations quantiques.

Le premier chapitre de cette thèse est consacré à l'étude des implications conceptuelles de la non-localité quantique, notion qui englobe celle d'intrication dans un sens précis. Nous détaillons comment les récentes approches informationnelles tentent de saisir la structure des corrélations non-locales, ainsi que les questions que ces dernières soulèvent concernant la capacité d'un observateur localisé à isoler un système de son environnement.

Le second chapitre détaille les effets de l'invariance de Poincaré sur la détection et la quantification de l'intrication. Cette invariance impose que tous les systèmes soient modélisés en dernière instance dans le cadre de la théorie des champs, ce qui implique qu'aucun système à énergie finie ne puisse être localisé, ainsi que la divergence de toute mesure d'intrication pour des observateurs localisés. Nous fournissons une solution à ces deux problèmes en démontrant l'équivalence générique qui existe entre une résolution spatiale finie des appareils de mesure et l'exclusion des degrés de liberté de haute énergie de la définition du système observé. Cette équivalence permet une interprétation épistémique du formalisme quantique standard décrivant les systèmes localisés non-relativistes et leurs corrélations, clarifiant ainsi l'origine des mesures finies d'intrication pour de tels systèmes.

Le dernier chapitre explore un cadre théorique récemment introduit qui prédit l'existence de corrélations quantiques sans ordre causal défini. Procédant par analogie avec le cas des corrélations non-locales, nous présentons quelques principes informationnels contraignant la structure de ces corrélations dans le but de mieux en comprendre l'origine physique.

\bigskip

\textbf{Mots-clés :} théorie quantique, théorie quantique relativiste des champs, informatique quantique, entropie d'intrication, invariance de Poincaré, schéma de localisation, courbes de type temps fermées, ordre causal, information quantique.

\newpage

\chapter*{Abstract}

Recent works in foundations of quantum (field) theory and relativistic quantum information try to better grasp the interplay between the structure of quantum correlations and the constraints imposed by causality on physical operations.

Chapter \ref{first-chapter} is dedicated to the study of the conceptual implications of quantum nonlocality, a concept that subsumes that of entanglement in a certain way. We detail the recent information-theoretic approaches to understanding the structure of nonlocal correlations, and the issues the latter raise concerning the ability of local observers to isolate a system from its environment.

Chapter \ref{second-chapter} reviews in what sense imposing Poincaré invariance affects entanglement detection and quantification procedures. This invariance ultimately forces a description of all quantum systems within the framework of quantum field theory, which leads to the impossibility of localized finite-energy states and to the divergence of all entanglement measures for local observers. We provide a solution to these two problems by showing that there exists a generic equivalence between a finite spatial resolution of the measurement apparatus and the exclusion of high-energy degrees of freedom from the definition of the observed system. This equivalence allows for an epistemic interpretation of the standard quantum formalism describing nonrelativistic localized systems and their correlations, hence a clarification of the origin of the finite measures of entanglement between such systems.

Chapter \ref{third-chapter} presents a recent theoretical framework that predicts the existence of correlations with indefinite causal order. In analogy to the information-theoretic approaches to nonlocal correlations, we introduce some principles that constrain the structure of such correlations, which is a first step toward a clear understanding of their physical origin.

\bigskip

\textbf{Keywords:} quantum theory, relativistic quantum field theory, quantum computing, entropy of entanglement, Poincaré invariance, localization scheme, closed timelike curves, causal order, quantum information.

\chapter*{Acknowledgements}

I address my warmest thanks to my dissertation advisors Daniel Estève and Alexei Grinbaum. Their trust and constant support allowed me to enjoy a rare freedom in the choice of topics to study. I am also indebted to \v{C}aslav Brukner, Professor at the University of Vienna and Director of the Institute for Quantum Optics and Quantum Information (IQOQI) of Vienna. The invitation to visit his group allowed me to better grasp where difficult questions lie in the topics of quantum foundations and quantum information.

I have learned a lot from the valuable discussions with Etienne Klein, Vincent Bontems, Ognyan Oreshkov, Magdalena Zych, Igor Pikovski, Mateus Ara{\'u}jo, Maël Pégny, Ämin Baumeler, Christina Giarmatzi and most importantly Fabio Costa, whose patience and will to answer my (numerous) questions were a valuable help throughout these years.

This dissertation is dedicated to my parents for their constant support.

Financial support and travel funds over the last five years were provided by the CEA Saclay, the European Commission Q-Essence Project, and the Ministry for Education, Research and Technology of France.

\chapter*{Note to the reader}

This is an updated version of the original Ph.D. dissertation, dated October 13\textsuperscript{th}, 2015. The original version can still be found on the \href{https://tel.archives-ouvertes.fr/tel-01146097}{national French dissertations archive} or on the \href{http://arxiv.org/abs/1510.01309v1}{arXiv}. Updates mainly consist of either the use of quotation marks or the rephrasing of technical results when attribution was not clear in a small number of sections concerning literature review. The complete list of modifications is included below. The original results of Chapter 2 were published in the \href{http://journals.aps.org/prd/abstract/10.1103/PhysRevD.90.065032}{Physical Review D}. The original results of Chapter 3 are faithful to the state of progress of my work around December 2014. An updated version of these results will soon be published in the Physical Review A.

\section*{List of modifications}

p.6: rephrasing of the Schrödinger cat example (source: \href{https://en.wikipedia.org/wiki/Measurement_problem}{Wikipedia}).
 
p.8: use of quotation marks for a sentence by Horodecki \emph{et al.}

p.9: reformulation of the comment on the definition of separable states with addition of an internal reference. 

p.10: reformulation of the comments on the definition of an entanglement measure, reminder of the source.

p.15: reformulation of part of the comments on the loopholes in experiments that attempt to violate the Bell inequality (source: \href{https://en.wikipedia.org/wiki/Bell_test_experiments}{Wikipedia}).

pp.17-18: reformulation of the definition of the quantum game and of some technical results, reminder of the source.

p.20: minor rephrasing and inclusion of a footnote.

p.22: rephrasing of the sentence concerning the consequences of the simplex structure, rephrasing of the footnote, reminder of the source.

p.23: clarification on the verbatim reproduction of a technical proof, reminder of the source.

p.28: clarification on the verbatim reproduction of a technical proof, reminder of the source. 

p.29: explicit quote using quotation marks for the comment on the proof.

p.30: reformulation of the introduction to Section \ref{section-121}, addition of a missing reference.

pp.31-33: clarification of the source of the mathematical definitions and examples.

pp.36-40: clarification of the sources of some mathematical definitions, use of quotation marks when required, addition of a missing reference.

pp.53-54: clarification of the sources, use of quotation marks for a technical conclusion on spin entropy.

p.56: use of quotation marks for a technical conclusion, reminder of the source.

p.58: rephrasing of the summary of the works of Saldanha and Vedral, reminder of the sources.

p.59: clarification concerning the verbatim reproduction of a technical introduction to the Unruh effect, reminder of the source. 

p.62: improved the readability of internal references.

p.64: added an internal reference to a previously cited source.

pp.76-77: rephrasing and use of quotation marks when required for the review of Jacobson's work.

p.81: use of quotation marks concerning the physical interpretation of two closed timelike curves models and addition of a missing reference.

p.102: rephrasing of the simplified version of the game with a reminder of the source.

\tableofcontents

\chapter*{Note synthétique}
\addcontentsline{toc}{chapter}{Note synthétique}
\markboth{Note synthétique}{}

\pagenumbering{roman} \setcounter{page}{1}

Cette thèse se situe à l’intersection des études sur les fondements de la théorie quantique (des champs) et de l’information quantique relativiste. Y sont discutées plusieurs problématiques autour de l’imbrication entre les notions de corrélations quantiques et de contrainte causale. Le caractère mathématique des travaux ici présentés est constamment motivé par les questions philosophiques qu’impose la théorie quantique, et il est essentiel d’avoir ces deux types d’investigations à l’esprit si l’on veut acquérir une compréhension adéquate des fondements de la théorie quantique.

\section*{Introduction}
\addcontentsline{toc}{section}{Introduction}
\markboth{Note synthétique}{}

Les notions de corrélation et de causalité sont depuis longtemps au centre de nombreux débats philosophiques et scientifiques, d'où le parti pris de cette thèse de se restreindre aux définitions opérationnelles de ces notions. Ainsi, une corrélation se définit en termes de probabilités jointes pour les variables observées, tandis que la dépendance causale nécessite une spécification des probabilités conditionnelles pour relier un changement de la variable conditionnante à un changement observable de la distribution de probabilité pour la variable conditionnée. Pourquoi donc relier ces deux notions dans le titre << corrélations quantiques et structures causales >> ? Les termes << quantiques >> et << structures >> sont naturellement cruciaux pour comprendre le sens de cette association.

La relativité restreinte représente le cadre formel de notre compréhension moderne de la notion de structure causale : elle définit à partir d'arguments opérationnels une géométrie des événements où les seules dépendances causales possibles sont les ordres << avant >> et << après >> pour une séparation de type temps entre les événements, ainsi que l'absence d'influence causale ou non-signalisation dans le cas d'une séparation de type espace. Ces types d'influence causale sont encodés dans une structure géométrique appelée espace de Minkowski. Un des premiers débats mettant en cause la compatibilité de la relativité et de la théorie quantique date du célèbre papier d’Einstein, Podolsky et Rosen de 1935 et qui discute le paradoxe portant leurs noms. Ce dernier peut être résumé de manière moderne comme suit : une paire de photons est produite de telle façon que le spin total vaut 0, i.e. les photons sont intriqués. Supposons maintenant que les deux photons soient arrivés à des laboratoires avec une séparation de type espace. Alors, si une mesure du spin selon un axe $x$ pour le premier donne +1, on peut prédire avec certitude qu’une mesure du spin pour l’autre photon selon le même axe donnera -1. EPR dénomment une propriété prédictible avec certitude << élément de réalité >> et lui accordent une existence indépendente de l'acte de la mesure. \'{E}tant donné que les photons doivent vérifier la condition de non-signalisation, la relativité affirme que la mesure effectuée sur le premier photon ne peut affecter l’état du second photon. On entend ici par état d’un système toute représentation des probabilités d’obtenir tel ou tel résultat après mesure de telle ou telle observable sur ledit système. Or la formulation d’avant 1935 du principe d’incertitude de Heisenberg affirmait qu’il est impossible de mesurer deux variables conjuguées telles que le spin selon un axe $x$ et le spin selon un axe orthogonal $z$ à cause de la perturbation irréductible de la propriété mesurable par observation du spin selon $x$ par toute tentative de mesurer le spin selon $z$. L’existence d’une propriété bien définie de spin selon $x$ avant la mesure n’est donc pas explicitement remise en cause, il s’agit simplement d’une << variable cachée >>. Revenons à notre expérience. Puisque l'état du deuxième photon ne peut être perturbé par la mesure effectuée sur le premier, on peut donc simultanément mesurer le spin selon $z$ pour ce deuxième photon, et si l’on obtient +1, alors on est certain qu’une mesure du spin selon $z$ pour le premier photon donnerait -1, i.e. le spin selon $z$ du premier photon est un élément de réalité. Mais alors, le premier photon possède simultanément des valeurs bien définies +1 et -1 pour des variables conjuguées telles que le spin selon $x$ et $z$, en contradiction avec le principe d’incertitude ! Notons qu'EPR n'accordent le statut de variables cachées qu'aux éléments de réalité, une hypothèse encore plus faible que celle sous-jacente à l'explication << par perturbation >> du principe d'incertitude.

Cette tension entre la condition de non-signalisation entre les photons et la structure des corrélations quantiques les unissant est à la base de la compréhension moderne de la théorie quantique. En effet, ce n’est qu’après 1935 et la formulation du fameux paradoxe EPR que la portée conceptuelle du principe d’incertitude (qui portait aussi le nom de principe d'indétermination) devint plus claire, poussant Bohr à hisser son principe de complémentarité du statut de simple impossibilité physique de mesurer simultanément des variables conjuguées à une position épistémologique où la définition d’une propriété est confondue avec la production d’un résultat par un appareil de mesure. Durant de nombreuses années, le choix entre l’incomplétude de la théorie quantique, au sens qu’une théorie plus fondamentale à variables locales cachées existe, et l’abandon de toute existence objective des propriétés non mesurées était d’ordre philosophique. Ce n’est qu’en 1967 que les bases d’une vérification scientifique de la validité de l’une ou l’autre option ont été jetées par Bell via la fameuse inégalité portant son nom. L’analyse de cette inégalité et de ses implications conceptuelles constitue le point de départ de cette thèse.

\section*{Chapitre 1}
\addcontentsline{toc}{section}{Chapitre 1}
\markboth{Note synthétique}{}

Le scenario imaginé par Bell, dont nous analysons une version moderne dans ce premier chapitre, est le même que celui du paradoxe EPR. En formalisant la notion de variable cachée locale, il montre que si le monde est décrit par de telles variables, la moyenne des résultats d'une certaine combinaison d'observables vérifie toujours une inégalité simple qu'il a introduite. Or la théorie quantique prévoit une violation de cette inégalité, phénomène confirmé par les expériences d'Aspect et d'autres durant les années 1980. Il n'existe donc pas de théorie à variables cachées locales qui compléterait la théorie quantique : l'indétermination des propriétés non mesurées est un fait scientifique dont il faudra dorénavant nous accomoder.

Durant cette même période, les premiers protocoles d'information quantique allaient voir le jour, à commencer par le protocole de cryptographie quantique introduit par Bennett et Brassard en 1984. Ce dernier prévoit la possibilité théorique d'une distribution parfaitement secrète d'une clé cryptographique en se basant précisément sur les corrélations quantiques de type EPR pour détecter toute tentative d'espionnage. Suivront d'autres utilisations astucieuses de ce type de corrélations telles que la téléportation des états quantiques, le codage super-dense, etc., ce qui imposa les corrélations de type EPR comme une ressource informationnelle non-locale et non comme une difficulté conceptuelle de la théorie quantique. Restait à identifier précisément à quel type de ressource on avait affaire. Les états quantiques utilisés par EPR et Bell étaient intriqués, mais des états intriqués ne violant pas l'inégalité de Bell et donc simulables par des variables cachées locales ont été rapidement identifiés.

Le domaine de l'information quantique a récemment produit le cadre conceptuel adéquat pour mieux cerner et analyser le ou les types de ressources mises en jeu, à savoir imposer des contraintes opérationnelles sur les partenaires d'un << jeu >> dont l'objectif est d'atteindre avec la meilleure probabilité possible une relation donnée entre les résultats de mesures conditionnés par des choix d'appareils de mesures. L'intrication n'est alors plus définie comme la ressource que possède tout état pur non-séparable (dans la représentation par ket) mais plutôt comme la ressource que des partenaires avec une séparation de type espace ne peuvent produire à partir d'états séparables grâce à des opérations locales et la communication de bits classiques. La non-localité est alors définie comme la ressource que ces mêmes partenaires ne peuvent produire à partir d'état séparables en ne partageant qu'une source de bits aléatoires classiques, toute communication leur étant interdite une fois le jeu commencé. L'utilisation de ces ressources quantiques se traduit toujours par une probabilité de gain du jeu plus élevée, et plusieurs résultats ont été obtenus sur la structure de l'intrication et de la	non-localité en utilisant ces méthodes, l'un des plus remarquables étant que tout état intriqué est non-local au sens opérationnel ci-dessus. La maturité des télécommunications classiques a naturellement placé le premier paradigme au centre du débat, et la compréhension complète de la structure des états intriqués est l'un des problèmes ouverts les plus importants de la théorie quantique à l'heure actuelle.

À nouveau, l'étude des liens qu'entretiennent les notions de corrélation quantique et de contrainte causale ont à nouveau permis une avancée conceptuelle importante. En effet, s'il est avéré que les corrélations quantiques respectent la condition de non-signalisation dans le scénario EPR, il est naturel de se demander s'il s'agit-il des corrélations les plus fortes à le faire. Autrement dit, peut-on déduire la borne quantique dans l'inégalité de Bell, dite borne de Tsirelson, à partir du respect de la condition de non-signalisation ? Cette question, posée par Popescu et Rohrlich, a reçu une réponse négative. L'introduction des boîtes Popescu-Rohrolich (PR) a fourni un exemple concret de distributions de probabilités (et donc d'états au sens opérationnel évoqué précédemment) généralisant le scénario EPR et respectant la contrainte de non-signalisation, mais qui violent de manière maximale l'inégalité de Bell. L'absence de lien causal entre les particules du scénario EPR, encodée grâce à la notion géométrique de séparation de type espace, est formulée pour les deux parties causalement indépendentes de la boîte PR de manière algébrique. Les boîtes PR constituent une des premières tentatives de reconstructions, même partielles, de la théorie quantique à partir de principes informationnels clairs. Si le principe de non-signalisation n'a pas permis pas de déduire la borne de Tsirelson, et encore moins la structure des corrélations quantiques, d'autres principes ont été formulés qui y parviennent partiellement, le plus abouti à cette date étant le principe de causalité informationnelle que nous détaillons dans ce premier chapitre.

La discussion qui précède a jonglé entre les définitions géométrique et algébrique d'absence d'influence causale entre deux variables et a considéré la structure des corrélations quantiques obéissant à une telle contrainte. La géométrie de l'espace de Minkowski étant bien plus riche, cette analyse n'est probablement pas suffisante. Le formalisme de la théorie algébrique des champs a été construit avec l'idée d'incorporer dans la structure algébrique même des observables quantiques la notion relativiste de localité. Nous détaillons ce formalisme dans le premier chapitre, et notamment comment le théorème de Reeh-Schlieder, au travers des contraintes dites de micro-causalité, impose que tous les états à énergie finie atteignent la violation maximale autorisée par la théorie quantique d'un avatar des inégalités de Bell, les systèmes considérés étant alors des régions d'espace à un instant fixe.

\section*{Chapitre 2}
\addcontentsline{toc}{section}{Chapitre 2}
\markboth{Note synthétique}{}

L'incorporation directe du principe de localité dans la théorie quantique mène inévitablement à des états décrivant les régions d'espace fortement corrélées. Le second chapitre s'ouvre sur la discussion d'une approche alternative qui consiste en une inclusion progressive des contraintes d'invariance de Poincaré au scénario EPR.

L’expérience de Bell considère le référentiel des observateurs, dénommés de manière standard Alice et Bob, comme fixe (le mouvement des particules arrivant dans les laboratoires séparés de type espace n'est alors pas important), tandis que le scénario relativiste considère que le référentiel de ces observateurs est lui aussi en mouvement. Par exemple, on peut considérer la situation où du point de vue d’un troisième observateur noté $O$ supposé fixe, Alice et Bob se déplacent selon un axe $z$ tandis que les particules se déplacent en sens opposé selon un axe $x$ avec spin total 0.

Des mesures de spin par Alice et Bob selon l'axe $y$ seront parfaitement anticorrélées indépendamment de leur vitesse selon $z$ et de celle des particules selon $x$ car il n’y a aucun mouvement relativiste selon la direction $y$. Par contre, si Alice et Bob décident d'effectuer des mesures selon l’axe $z$, le déplacement selon à la fois $x$ et $z$ de la particule dans leur référentiel commun implique que les résultats de mesure seront affectés par cette vitesse (qui correspond à l’opposé de la vitesse de déplacement commune d’Alice et Bob selon $z$ pour $O$). En effet, en régime relativiste, la partie << spin >> de l’état des particules subit sous l'effet des rotations dites de Wigner une transformation qui dépend de la partie << moment >>. Ainsi, le déplacement selon l'axe $z$ de la particule dans le référentiel commun d’Alice et Bob affecte les résultats des mesures de spin selon ce même axe $z$. Dans le régime non-relativiste, les transformées de Galilée ne couplent pas le spin et le moment, et donc de telles considérations n’entrent pas en jeu.

Que signifie concrètement le terme << affecter >> ? Il suffit de décrire ce que << voient >> Alice et Bob, i.e. d'expliciter l’état des particules dans leur référentiel. Il faut donc appliquer à l'état des particules dans le référentiel de $O$ la transformation de Lorentz permettant de passer du référentiel commun à Alice et Bob à celui de $O$. Les calculs montrent que l'état ainsi obtenu produit des résultats de mesures de spin selon $z$ qui ne seront pas parfaitement anticorrélés, mais qu’une rotation spécifique selon $y$ de l’appareil de mesure donnera des résultats parfaitement anticorrélés. Reste à définir un objet mathématique modélisant de manière covariante les résultats de ces mesures de spin. Le chapitre 2 discute quelques tentatives dans la littérature de définir une observable spin covariante ainsi que l'impossibilité d'une formulation covariante de la notion d’entropie d'intrication. Il s'avère donc que l'on ne peut donner un sens covariant à la notion de corrélation quantique.

Dans le cas d'observateurs non inertiels, des difficultés apparaissent à cause de la dépendence du nombre de particules en l'accélération des observateurs : c'est l'effet Unruh. Par ailleurs, la manipulation de l'intrication par des observateurs locaux est problématique étant donné le caractère global des modes de tout champ quantique. Ces difficultés peuvent être partiellement dépassées grâce à l'utilisation des détecteurs de type Unruh-DeWitt, mais les fortes corrélations évoquées précedemment entre degrés de liberté de type espace pour tout état d'énergie finie impliquent une divergence de toutes les mesures d'intrication, et notamment celle de l'entropie d'intrication. Comme tous les systèmes, y compris ceux de basse énergie, sont en dernière instance décrits par des champs quantiques, cette divergence entre en contradiction avec les résultats finis obtenus pour toute mesure d'intrication en théorie quantique des systèmes finis.

Une première contribution de cette thèse est de fournir une méthode de régularisation de cette divergence à basse énergie qui, contrairement à certains travaux récents que nous évoquerons, est indépendente du modèle de détecteur couplé au champ et valable pour tous les états à énergie finie. L'idée générale de cette méthode est de montrer la convergence à basse énergie entre le schéma standard de localisation, i.e. la méthode standard d'association des degrés de libertés du champ à des régions avec séparation de type espace, et le schéma de localisation dit de Newton-Wigner. Le vide étant séparable pour ce dernier à toutes les échelles d'énergie, nous obtenons des résultats finis pour l'entropie d'intrication en régime de basse énergie pour tous les états à énergie finie. Le formalisme ainsi déduit est naturellement équivalent à celui de la théorie quantique des systèmes finis, d'où une transition contrôlée entre la description de l'intrication en théorie quantique des champs et celle de la théorie quantique des systèmes finis. Notons que là encore, les liens forts qui existent entre la notion de corrélation quantique et de structure causale se manifestent. En effet, la séparabilité du vide pour le choix du schéma de localisation de Newton-Wigner est obtenue au prix d'une dynamique non-locale, et donc d'une violation des contraintes de causalité encodées dans la géométrie de l'espace de Minkowski. Tenter de préserver la causalité revient à choisir le schéma de localisation standard, et donc à admettre la divergence des mesures de corrélations. La régularisation de l'entropie d'intrication grâce à la granularité des appareils des mesures fournit un autre éclairage sur ces liens : si l'on se situe à une échelle d'énergie où les éventuelles violations de la structure relativiste de la causalité sont indétectables, alors les mesures d'intrication fournissent des résultats finis.

L'étape suivante naturelle est de considérer ce qui se passe lorsque l'on tente de régulariser la divergence de l'entropie d'intrication en régime de haute énergie. La dernière section de ce chapitre rappelle partiellement les travaux de Jacobson à propos de la relation << thermodynamique >> entre une régularisation à haute énergie de l'entropie d'intrication (avec la condition que la loi de proportionalité de l'entropie et de l'aire ainsi que la relation de Clausius tiennent) et la dynamique de l'espace-temps telle que décrite par l'équation d'Einstein de la relativité générale, autre indice des liens forts qui existent entre la structure des corrélations quantiques et les modèles relativistes de structures causales.

\section*{Chapitre 3}
\addcontentsline{toc}{section}{Chapitre 3}
\markboth{Note synthétique}{}

Nous avons jusqu'à maintenant passé en revue les difficultés qui peuvent se poser pour la détection et la mesure de l'intrication lorsque l'on incorpore de façon directe ou plus progressive des contraintes causales sous forme de contrainte géométrique, i.e. d'invariance relativiste. Le dernier chapitre de cette thèse s'intéresse à la possibilité de structurer dans un cadre algébrique cohérent les relations causales autres que celle de non-signalisation, approche récemment introduite par Oreshkov \& \emph{al.} et qui tente de contourner les difficultés que pose l'ajout de contraintes géométriques à la théorie quantique. La première tentative d'incorporer de manière algébrique des liens causaux complexes au sein de la théorie quantique date des travaux de Deutsch sur les courbes de type temps fermées. Cependant, la dynamique non-unitaire de ces modèles pose problème car elle ouvre la voie à des communications à vitesse supra-lumineuse. Le cadre formel introduit par Oreshkov \& \emph{al.} est particulièrement intéressant à cet égard car il préserve la structure linéaire de la théorie quantique. Relaxant la condition d'existence d'une structure causale globale, ce formalisme reproduit toutes les corrélations multipartites que peuvent posséder des partenaires dont les opérations locales sont régies par la théorie quantique. De manière analogue au paradoxe EPR, il est possible de définir un << jeu causal >> entre deux partenaires Alice et Bob pour lequel la probabilité de succès possède une borne spécifique si tous les événements locaux tels que le choix d'un appareil de mesure par Alice ou le résultat d'une mesure chez Bob sont ordonnés selon des ordres causaux bien définis ou une mixture de tels ordres causaux. Ce cadre formel prédit l'existence de corrélations qui violent une << inégalité causale >>, ce qui implique l'impossibilité d'identifier les événements locaux à une série d'événements causalement ordonnés. Assumer la validité de la théorie quantique à un niveau local n'implique donc pas l'existence d'une structure causale globale de type relativiste, i.e. un événement n'est pas forcément << avant >>, << après >> ou sans lien causal avec un autre, mais peut être dans une << superposition >> d'ordres causaux. Au contraire, si l'on assume que les opérations locales sont classiques, alors il est possible d'organiser les événements locaux bipartites au sein d'une structure causale globale.

Procédant par analogie avec les approches informationnelles qui tentent de caractériser la structure des corrélations quantiques vérifiant la condition de non-signalisation, la seconde contribution de cette thèse est de placer les corrélations sans ordre causal défini au sein d'un cadre probabiliste plus général afin de mieux comprendre les liens entre certaines contraintes sur l'ordre des opérations locales et la possibilité d'émergence de telles corrélations, ainsi que tenter de reconstruire à partir de principes clairs les bornes classiques et quantiques associées au jeu causal. Les spécificités du jeu causal font apparaître l'extension du principe de causalité informationnelle comme un candidat naturel à une telle reconstruction. La performance des partenaires dans le jeu telle que mesurée par l'information mutuelle servira d'outil de base à cette approche. Ainsi, une classe de jeux pour lesquels les corrélations causalement ordonnées performent de manière bornée est introduite, et la borne quantique associée au jeu causal initial est déduite d'une telle condition de performance bornée, démontrant ainsi qu'une formulation entropique des contraintes de signalisation imposées par les structures causales ordonnées impose la même limite que la théorie quantique à la probabilité de succès dans le jeu causal. Nous discutons enfin la possibilité de définir d'autres mesures de performance. L'introduction d'une telle mesure alternative, basée sur la `corrélation maximale de Hirschfeld-Gebelein-Rényi', s'avère en effet plus adaptée aux jeux causaux lorsque l'on tente de discriminer les corrélations classiques et quantiques des corrélations supra-quantiques.

\section*{Perspectives}
\addcontentsline{toc}{section}{Perspectives}
\markboth{Note synthétique}{}

L'exploration dans cette thèse des liens forts qui existent entre les notions de corrélation quantique et de structure causale nous amène à poser des questions plus générales. La granularité des appareils de mesure exploitée lors de la régularisation de l'entropie d'intrication correspondait à une opération d'ordre géométrique car elle se base sur une équivalence à basse énergie entre différentes façons d'affecter les degrés de libertés du champ à des régions d'espace à un instant fixe. Plus généralement, la théorie algébrique des champs est un encodage dans l'algèbre des observables quantiques de la structure géométrique de l'espace-temps. Cette approche est à la base de plusieurs tentatives d'unification de la théorie quantique et de la relativité générale : l'unification est cherchée au travers d'une << géométrisation >> plus importante de la théorie quantique. Nous avons, au fur et à mesure de cette thèse, défendu un point de vue complémentaire : il faut tenter de rendre plus algébrique notre compréhension des structures causales, en commençant par les relations simples de type << avant >>, << après >> et la condition de non-signalisation. Le cadre formel discuté dans le dernier chapitre, lui-même basé sur des suggestions récentes de Hardy, est une première étape.

Cependant, beaucoup de questions demeurent : à quoi correspondent les corrélations quantiques sans ordre causal défini ? Quelles voies explorer pour espérer une implémentation expérimentale de la violation de l'inégalité causale ? Existe-t-il une analogie plus ou moins formelle entre cette ressource et l'intrication ? Des réponses à ces questions permettraient de mieux comprendre la nature des corrélations en jeu, et ainsi d'approfondir encore davantage notre compréhension des liens entretenus par les notions de corrélation quantique et de structure causale.

\chapter*{Introduction}
\addcontentsline{toc}{chapter}{Introduction}
\markboth{Introduction}{}

\pagenumbering{arabic} \setcounter{page}{1}

This dissertation is at the intersection of foundations of quantum (field) theory and the theory of relativistic quantum information. It considers some issues that arise when trying to better grasp the interplay between the structure of quantum correlations and the constraints imposed by causality on physical operations. It goes without saying that the more mathematical sort of inquiry presented here is constantly motivated by the philosophical questions quantum theory imposes, and it is essential for gaining an adequate grasp of the foundations of physics to have both types of investigations in mind.

Chapter \ref{first-chapter} considers some foundational issues about quantum correlations in both the nonrelativistic and the relativistic settings. Very early in the development of quantum theory, worries arose about a potential conflict between the correlations predicted by quantum theory and our intuitions about the physical world rooted in classical physics---especially relativistic field theories---in which the state of a local system is independent of the state of distant systems. In 1964, Bell elucidated the peculiar nonlocality of quantum theory by showing rigorously that some correlations predicted by quantum theory cannot be reproduced by any local hidden variable model. Bell's approach clarified the sense in which quantum correlations respect causality, and later Popescu and Rohrlich showed that there exist theories that are more nonlocal than quantum theory but still respect the causality constraints on marginal probabilities describing measurement results of local observers. More recently, quantum correlations were thoroughly studied in the context of general non-signalling correlations. A partial derivation of bipartite quantum correlations was found based on a physically clear principle called `information causality' that extends the no-signalling condition. These results belong to the field of (partial or complete) reconstructions of quantum theory, which aims at a better understanding of the structure of quantum correlations and hopefully of quantum theory itself.

After the seminal works of Summers, Redhead, Clifton, Halvorson and others, it appeared that entanglement is much more deeply entrenched in any relativistic quantum theory. Indeed, results in algebraic quantum field theory such as the Reeh-Schlieder theorem or the ``natural" type III$_1$ constraint on local algebras of observables entail a generic entanglement in the state space of a field system. This deeply affects the ability of local observers to isolate a system from its environment.

Chapter \ref{second-chapter} focuses on technical aspects of entanglement detection and quantification in the relativistic setting. Observer-dependent entanglement arises when Poincaré invariance is imposed on quantum theory. Therefore, understanding the implications of this dependence is crucial for modern quantum information protocols. We review how Poincaré invariance couples the spin and momentum degrees of freedom for inertial observers, which in turn implies a transfer of entanglement between them and imposes a fine-tuning of local operations if observers are to detect the nonlocal character of correlations. For non-inertial observers, interpretational difficulties arise due to the observer-dependence of the number of particles, a phenomenon called Unruh effect. Moreover, entanglement manipulation by local observers is problematic because of the global character of field modes. These difficulties can be partially overcome by using the so-called Unruh-DeWitt detectors, but many problems arise when local observers try to quantify entanglement for any infinite-mode system. We detail these issues by presenting results on the area law for entanglement entropy and its divergence in the continuum limit for many field models, thus providing an alternative point of view on the ``invasive" character of entanglement in quantum field theory. Following the intuition gained in Chapter \ref{first-chapter} about the deep relationships between entanglement and the properties of local algebras, we present a novel regularization technique at low energy of entanglement entropy of infinite-mode systems. This is the first contribution of this dissertation. The idea consists in proving the convergence at low energy between the standard localization scheme, i.e. the standard way of assigning the field degrees of freedom to spacelike separated regions of spacetime, and the so-called Newton-Wigner localization scheme. The vacuum state is unentangled for the latter at all energy scales, which results in a finite entropy of entanglement at low energy for all finite-energy states. The derived low-energy formalism is as expected equivalent to standard nonrelativistic quantum theory, hence a controlled transition from the quantum field theory description of entanglement to the one by nonrelativistic quantum theory. A natural step is then to consider what happens if one tries to regularize entanglement entropy at high energy. The last section of this chapter partially reviews the work by Ted Jacobson on the ``thermodynamical" relationship between a high-energy regularization of entanglement entropy---under the conditions that the area law and the Clausius relation hold---and spacetime dynamics as described by Einstein's general relativity equation.

In Chapter \ref{third-chapter}, we expand our analysis of the interplay between quantum correlations and the causal structures ordering measurement events by examining basic concepts such as ``localization" and ``causal structure" from an operational point of view. Early approaches have tried to model exotic causal structures such as closed timelike curves (CTC) using the quantum formalism. For instance, Deutsch introduced a CTC model that extends the quantum formalism with non-unitary operations and avoids logical paradoxes. More recent approaches try to preserve the linear and unitary structure of quantum operations. Some of these models exhibit correlations beyond causally ordered ones. We focus on a recently introduced formalism where causal relations are defined in terms of the possibility of signalling. According to this definition, it is possible to find an operational task---a `causal game'---whose probability of success is bounded for operations performed in a definite or mixture of causal orders. All the possible multipartite correlations that can be produced by different agents whose operations are locally described by quantum mechanics are reproduced without making any prior assumption on a causal structure in which the operations are embedded. An example of such correlations was found that allows winning the causal game with a probability of success larger than the causally ordered bound. This shows that assuming the local validity of quantum mechanics does not imply the existence of a global causal structure. In contrast, if classical mechanics is assumed to hold locally, bipartite correlations can always be embedded in a global causal structure, while it was proven that this is not the case for three parties. Following the standard approach to entanglement characterization, the second contribution of this dissertation consists in placing such correlations in the context of a generalized probabilistic framework in order to better understand the connections between the local ordering of events and the emergence of an indefinite global order, and examining the relevance of various informational principles for a characterization of the classical and quantum bounds on correlations with indefinite causal order. The general aim is to provide a reconstruction of bipartite quantum correlations with indefinite causal order from clearly motivated physical principles. A possible extension of the information causality principle thus appears as a natural candidate. We reformulate the causal game as a random access code and introduce a class of causal games for which causally ordered correlations perform with a bounded efficiency as measured by mutual information. We then show that the quantum bound can be derived by taking bounded efficiency as an assumption, i.e. the entropic characterization of fixed causal structures imposes the same limit as quantum theory to the success probability in a family of causal games. This principle is very similar to the intuition provided by information causality, which states that the efficiency of a protocol using non-signalling correlations and one-way signalling cannot exceed the total amount of signalling if three natural conditions on mutual information hold. We also show that shifting the focus from mutual information to an alternative measure of dependence called `Hirschfeld-Gebelein-Rényi maximal correlation' is probably better suited in the context of causal games for discriminating classical and quantum correlations from supra-quantum ones. We end our discussion by reviewing the quantum switch framework, an instance of a higher-order computation using quantum supermaps which provides an alternative approach to correlations with indefinite causal order.

\chapter{Conceptual implications of entanglement}
\label{first-chapter}

The goal of this chapter is to analyze in which way quantum correlations have deep conceptual implications on our understanding of ``reality". We review in the first section the formalism of nonrelativistic quantum theory, and provide an operational definition of an intrinsic property of multipartite systems called entanglement. The Bell theorem is presented and its implications analyzed. The discussion then continues with the recent attempts to ground the structure of entanglement on clear physical principles. We will mainly focus on the partial reconstruction of bipartite non-signalling quantum correlations from the information causality principle. In the second section, we shift the focus to relativistic quantum fields. We begin by a review of the framework of algebraic quantum field theory, and following Halvorson's seminal work \cite{halvorson_locality_2001}, we analyze the constraints causality imposes on the properties of local algebras of observables, and the implications of such properties on the ability of local observers to isolate their part of the field system from its environment.

\section{Structure of quantum correlations}

The first quantum models were elaborated more than a century ago. The unification by Dirac and von Neumann of Heisenberg's matrix mechanics and Schrödinger's wave mechanics into a unique quantum theory gave a strong mathematical basis to all these models and ones still to be elaborated. Unlike classical physics (including thermodynamics and relativity), the postulates of quantum theory rest on no clear physical principles. During the last decades, many efforts have been made to find a satisfactory interpretation to these postulates, but no consensus was reached. Central to these difficulties is the \emph{measurement problem}, as illustrated by Schrödinger's cat paradox. The triggering of a mechanism that causes the death of a cat is conditioned on a quantum event, such as the decay of a radioactive atom. According to the Schrödinger equation, the cat evolves into a linear combination of ``alive cat" and ``dead cat" states, each of these states being associated with a nonzero probability amplitude. However, after the measurement the cat is either alive or dead. The question is: how can we characterize the transition from a probability distribution of possible outcomes into a well-defined outcome? More generally: how can one establish a correspondence between quantum and classical reality? Other ``paradoxes" such as the notorious Einstein-Podolski-Rosen (EPR) paradox further highlight the departure quantum theory imposes on our ``realist" view of the world based on concepts of classical physics. In this section we review the basic postulates of quantum theory, some of the conceptual novelties they convey and the recent attempts to (partially) ground these postulates on physically motivated principles.

\subsection{General postulates of quantum theory}

The standard postulates of quantum theory are \cite{nielsen_chuang,paris_modern_2012}:
\begin{description}
\item[$\Pos$] The space of states of an isolated quantum system corresponds to the set of positive trace class operators with trace 1 (called density matrices) on a Hilbert space  $\mathcal{H}$. The Hilbert space $\calH_{AB}$ associated to a composed system $AB$ is the tensor product $\calH_A\otimes\calH_B$ of the Hilbert spaces associated to the subsystems $A$ and $B$. If $\rho_{AB}$ is the state of the composed system, the partial traces
\begin{equation}
\rho_A=\trace_B(\rho_{AB}),\hspace{0.5cm}\rho_B=\trace_A(\rho_{AB})
\end{equation}are also density matrices describing subsystems $A$ and $B$ respectively.
\item[$\Poss$] Transformations of an isolated system that are associated with the action of a connected Lie group correspond---via a strongly continuous unitary representation in $\mathcal{H}$ of its universal cover---to a one-parameter group of unitary operators acting on $\calH$. In particular, if $\rho(0)$ is the initial state of the system, the transformation corresponding to time translation is described by a one-parameter group of unitary operators $\{U(t)\}_{t \in \mathbb{R}}$, and the state at time $t$ is given by:
\begin{equation}
\rho(t)=U(t)^{-1}\rho(0)\, U(t).
\end{equation}One can also describe this evolution by a differential equation called the Liouville-von Neumann equation:\begin{equation}
i\hbar\frac{\partial \rho}{\partial t}=[H(t),\rho],
\end{equation}where $\hbar$ is the reduced Planck constant and $H(t)$ is a one-parameter group of self-adjoint operators with dense and invariant domain in $\calH$ called the Hamiltonian of the system.
\item[$\Posss$] A measurement apparatus is described by a collection $\{M_m\}$ of bounded operators acting on $\calH$ verifying the following completeness relation:
\begin{equation}
\label{completeness}
\sum_m M_m^{\dagger}M_m= \id,
\end{equation} where the index $m$ refers to the measurement outcomes that may occur in the experiment. If $\rho$ is the state of the system immediately prior to the measurement, then the probability for outcome $m$ after the measurement is given by Born's rule:
\begin{equation}
\label{BornRule}
p(m)=\operatorname{tr}( M_m^{\dagger}\rho M_m),
\end{equation}and the state immediately after the measurement is either
\begin{equation}
\label{MeasPos}
\rho_m=\frac{M_m^{\dagger}\rho M_m}{\trace(M_m^{\dagger}\rho M_m)}
\end{equation}if the measurement outcome is $m$, or
\begin{equation}
\tilde{\rho}=\sum_m p(m)\rho_m=\sum_m M_m^\dag\rho M_m
\end{equation}if the measurement results are not recorded.
\end{description}

Note that:
\begin{itemize}
\item[(i)] The formalism presented above accounts for both pure and mixed states. One can retrieve the usual ``bra-ket" formalism whose primary elements are pure states by noting that the latter correspond to extremal projectors.
\item[(ii)] Probabilities are defined through observables $F_m=M_m^{\dagger}M_m$. These are positive self-adjoint operators called Positive-Operator Valued Measures (POVM). If we further impose that the $\{F_m\}$'s are orthogonal then they are called Projection-Valued Measures\footnote{See Appendix \ref{algebra-theory}.} (PVM).
\item[(iii)] Postulate 3 describes a measurement apparatus through operators $\{M_m\}$: there exists an infinite number of possible decompositions of the $F_m$'s into $M_m^{\dagger}M_m$'s, in correspondence with the infinite number of possible physical apparatus measuring the observables  $\{F_m\}$. Note that one cannot deduce posutlate 3 from postulate 2.
\item[(iv)] Postulates 2 and 3 can be grouped in the framework of quantum operations, which can describe any combination of unitary operations, interactions with an ancillary quantum system or with the environment, quantum measurement, classical communication and postselection. A quantum operation is a trace non-increasing linear map $\Phi:\mathfrak{B}(\calH)\rightarrow \mathfrak{B}(\calH)$ that is also completely positive, i.e. $\Phi\otimes \id_n$ is positive for all $n\in\mathbb{N}$. Stinespring's dilation theorem then provides the so-called Kraus representation of the quantum operation $\Phi$:
\begin{equation}
\Phi(\rho)=\sum_m M^\dag_m \rho M_m
\end{equation}
for any density matrix $\rho$, where $M^\dag_m,M_m \in\mathfrak{B}(\calH)$, called Kraus operators, verify the relation:
\begin{equation}
\label{completeness}
\sum_m M_m^{\dagger}M_m\leq \id.
\end{equation}
\end{itemize}

\subsection{Entanglement as a resource}

The effect of the replacement of the classical concept of phase space by the abstract Hilbert space makes a gap in the description of composite systems. As recognized by Einstein, Podolsky, Rosen, and Schrödinger, entanglement, which is a specific type of nonclassical\footnote{We detail the meaning of the word ``nonclassical" in Section \ref{nonloc-context}.} correlations between the subsystems of a composed system, is the essence of the quantum formalism. However, as commented by Horodecki \emph{et al.} in \cite{horodecki_quantum_2009}, ``[it] waited over 70 years to enter laboratories as a new resource as real as energy". This ``holistic" property of compound quantum systems has potential for many quantum processes such as quantum cryptography, quantum teleportation or dense coding, and many efforts were put in the study of entanglement characterization, detection, distillation, and quantification \cite{horodecki_quantum_2009}.

\subsubsection{LOCC paradigm}

A state $\rho_{ABC...}$ shared by parties $A,B,C,...$ is said to be
{\it separable} if and only if (iff) it can be written in the form
\begin{equation}
    \label{separable}
    \rho_{ABC...} = \sum_i p_i ~ \rho^i_{A} \otimes \rho^i_{B} \otimes
    \rho^i_{C} \otimes ...
\end{equation}
where $\{p_i\}$ is a probability distribution. These states can
trivially be created by local quantum operations (LO) and classical communication (CC) between parties: as argued in \cite{plenio_introduction_2005}, Alice can simply sample from the probability distribution
$p_i$ and then share the outcome $i$ with other parties. Subsequently, each party $X$
can locally create state $\rho^i_X$ and then discard the information
about outcome $i$. The crucial point is that the converse is also true : a quantum state $\rho$ may be generated perfectly using LOCC iff it is separable. This can be traced back to the fact that separable states are the ``endpoints" of LOCC transformations between quantum states \cite{nielsen_conditions_1999}. If the measurements results on a quantum system cannot be simulated using a separable state and LOCC, then its state will be considered as entangled. Thus, adopting a highly operational point of view, entanglement can be seen as a \emph{resource} that allows parties to overcome the LOCC constraint in solving certain multipartite tasks, generally referred to as `games', in the sense that parties win the game if their measurement outcomes (outputs) for the given choice of measurement settings (inputs) are correlated in a certain way\footnote{This rather involved definition of entanglement is justified because alternative paradigms define resources that are distinct from entanglement and absent from separable states. Therefore, defining entangled states as nonseparable ones can generate a confusion on what resource we are referring to if different paradigms are available.}.

\subsubsection{Entanglement measures}

In  \cite{plenio_introduction_2005} Plenio and Virmani characterize a good entanglement measure by the fact that it should capture ``the essential features that we associate with entanglement", i.e. it should be 0 for separable states and should not increase under LOCC. Ideally, it also should be related to some operational procedure. Many well known entanglement measures are based on the von Neumann entropy of quantum states, which extends the Shannon entropy of classical states. It is defined as follows:
\begin{equation}
S(\rho)=-\trace(\rho\log(\rho)),
\end{equation}for any density operator $\rho$.
The following is a list of possible postulates for an entanglement measure \cite{plenio_introduction_2005}:
\begin{enumerate}
    \item A {\it bipartite} entanglement measure $E(\rho)$ is a mapping from density
    matrices into positive real numbers:
    \begin{equation}
    \rho \mapsto E(\rho) \in \mathbb{R}^+
    \end{equation}
    defined for states of arbitrary bipartite systems. A
    normalization factor is also usually included
    such that the maximally entangled state $|\psi_d^+\rangle \lang \psi_d^+|$ where
\begin{equation}
    |\psi_d^+\rangle= \frac{|0,0\rangle + |1,1\rangle + .. +|d-1,d-1\rangle}{\sqrt{d}},
    \end{equation}
    where $d$ is the dimension of the Hilbert space, has $E(|\psi_d^+\rangle \lang \psi_d^+|)=\log d$.
    \item $E(\rho)=0$ iff the state $\rho$ is separable.
    \item $E$ does not increase on average
    under LOCC, i.e.
    \begin{equation}
        E(\rho) \ge \sum_{i} p_i E\left(\frac{A_i^{\dagger}\rho A_i}{\trace(A_i^{\dagger}\rho A_i)}\right),
    \end{equation}
    where the $\{A_i\}$ are the Kraus operators describing some LOCC protocol
    and the probability of obtaining outcome $i$ is given by $p_i=\trace(A_i^{\dagger}\rho A_i)$.
    \item  For a bipartite pure state $|\psi\rangle\langle\psi|_{AB}$ the measure
    reduces to the entropy of entanglement\footnote{The entropy of entanglement for a bipartite system $AB$ is defined by $S_A=S\circ \trace_B = S\circ \trace_A=S_B$.}:
    \begin{equation}
        E(|\psi\rangle\langle\psi|_{AB}) = (S\circ \trace_B)
    (|\psi\rangle\langle\psi|_{AB}).
    \end{equation}
\end{enumerate}
One can find in the literature authors that impose additional requirements such as convexity, additivity or continuity, depending on their needs. Other entanglement measures exist such as the entanglement of distillation, the entanglement cost, the relative entropy of entanglement and the squashed entanglement. Their definitions are generally based on operational considerations about entanglement quantification in quantum information protocols. Nonetheless, these measures not only have a practical purpose but also developed into powerful mathematical tools that contributed to the formalization of important open questions such as the additivity of quantum channel capacities or the bounding of quantum computing fault tolerance thresholds (See references in \cite{plenio_introduction_2005} for more details).

\subsubsection{Alternative paradigms}

The notion of entangement as a resource is of course implicitly related to our restriction of quantum operations to LOCC operations. Switching to a different set of restrictions on possible operations, one can define and study new resources. For instance, nonlocality is defined as the resource that cannot be simulated in the local operations and shared randomness (LOSR) paradigm, i.e. when parties are forbidden all sorts of communication, being allowed though to synchronize their local operations with respect to a common classical random variable shared in advance. Again, only separable states can be created from scratch using LOSR, and if we cannot simulate the measurements results on a quantum system by measurements on a separable state, then its state is called nonlocal. Nonlocality and entanglement are distinct resource \cite{gisin_hidden_1996}. For instance, one can show that maximally nonloca states are not maximally entangled \cite{methot_anomaly_2006}. One of the main difficulties in these paradigms is to define sufficiently `subtle' games to allow a detection of the resource under consideration. We will analyze this issue in more details in the next section.

\subsection{Bell's theorem}
\label{nonloc-context}

Besides its importance for quantum information processing, entanglement is at the root of many departures from important concepts of classical physics. Indeed, the combination of three of the most natural assumptions in classical physics---free will, realism and locality---is questioned by entanglement effects. Free will, encoded as the fact that the setting of the measurement apparatus can be chosen independently of the parameters that determine its future outcomes, is the most fundamental one. Zeilinger commented \cite{zeilinger_dance}:
\begin{quote}
``[W]e always implicitly assume the freedom of the experimentalist [...] This fundamental assumption is essential to doing science. If this were not true, then, I suggest, it would make no sense at all to ask nature questions in an experiment, since then nature could determine what our questions are, and that could guide our questions such that we arrive at a false picture of nature."
\end{quote}
Realism states that alongside the outcomes of actually performed measurements, the outcomes of potentially performed measurements also exist at the same time. For instance, position and velocity of a car exist, independently of measuring them or not. Locality states that the outcomes of a measurement on one part of a system cannot depend on the choice of the measurement to be performed on the other part of the system when the two parts are spacelike separated. This is a somewhat more elaborate hypothesis since it is based on the fact that there exists a maximum propagation velocity in nature, which is a well verified experimental fact at the basis of relativity theory. Whether one or more of these assumptions should be dropped in a world described by quantum theory was subject to heated philosophical debates, all brought to the realm of science by Bell's celebrated 1964 theorem \cite{Bell1964a}. We provide here the modern form of Bell's argument introduced in 1969 by Clauser, Shimony, Horne and Holt (CHSH) \cite{clauser_proposed_1969}.

Consider two parties Alice and Bob, and the simple case of measurements of two binary-valued observables, $x \in \{0, 1\}$ with outcomes $a \in \{0, 1\}$, performed by Alice in a region $A$, and $y \in \{0, 1\}$ with
outcomes $b \in \{0, 1\}$, performed by Bob in a spacelike separated region $B$. They are allowed to confer on a joint strategy before the game starts, but once the game starts they are separated and not allowed further communication.
In this setup, the notions of realism and locality are combined into the so-called \emph{local realism} hypothesis, which is formalized as follows:
\begin{itemize}
\item[(i)] There exists a probability space $\Lambda$ such that the observed outcomes by both Alice and Bob result by random sampling of the (unknown, ``hidden") parameter $\lambda \in \Lambda$. $\Lambda$ is assumed to be endowed with a probability measure of density $\mu$ such that the expectation of a random variable $X$ on $\Lambda$ with respect to $\mu$ is written:
\begin{equation}
E(X) = \int_\Lambda X(\lambda) \mu(\lambda) d \lambda.
\end{equation}
\item[(ii)] The values observed by Alice or Bob are functions of the local detector settings and the hidden parameter only. Thus:
\begin{itemize}
\item[(a)]Value observed by Alice with detector setting $x$ is $a(x,\lambda)$, abbreviated as $a_x$.
\item[(b)]Value observed by Bob with detector setting $y$ is $ b(y,\lambda)$, abbreviated as $b_y$.
\end{itemize}
\end{itemize}
Since each of the four quantities $a_0, a_1, b_0$ and $b_1$ is $\pm 1$, either $b_0+b_1$ or $b_0-b_1$ is 0, and the other $\pm 2$. From this it follows that:
\begin{equation}
a_0 b_0+ a_0 b_1 +a_1 b_0 - a_1 b_1 = a_0 (b_0+b_1) + a_1 (b_0 - b_1) \leq 2,
\end{equation}hence the so-called Bell or CHSH inequality\footnote{Other Bell inequalities are obtained by relabeling inputs and outputs, the latter being conditionned on the corresponding input.}:
\begin{equation}
\label{Bell-ineq}
E(a_0 b_0) + E(a_0 b_1) +E(a_1 b_0) - E(a_1 b_1) \leq 2.
\end{equation}
Suppose now that quantum theory is valid at Alice and Bob's laboratories, and that they perform spin measurements on electrons. Alice can choose between two measurement settings denoted by $A_0$  and $A_1$, and similarly, Bob chooses between two measurement settings denoted by $B_0$ and $B_1$. Take:
\begin{align}
\begin{aligned}
 A_0 &=\sigma_z \otimes \id,\\
 A_1 &= \sigma_x \otimes \id,\\
B_0 & = -\frac{1}{\sqrt{2}} \left[\id \otimes (\sigma_z + \sigma_x)\right],\\
B_1 & = \frac{1}{\sqrt{2}}\left[\id \otimes (\sigma_z - \sigma_x)\right],\\
\end{aligned}
\end{align}
where $\sigma_x = \begin{pmatrix}
0 & 1 \\
1 & 0
\end{pmatrix}$ and $\sigma_z = \begin{pmatrix}
1 & 0 \\
0 & -1
\end{pmatrix}$ are the usual Pauli matrices, and assume Alice and Bob share a singlet state in the eigenbasis of $\sigma_x$:
\begin{equation}
|\psi\rang =\frac{1}{\sqrt{2}}\left(|01\rang - |10\rang\right).
\end{equation}
One can easily check that:
\begin{align}
\begin{aligned}
\lang A_0 A_1\rang &= \lang A_1 B_0\rang = \lang A_1 B_1\rang = \frac{1}{\sqrt{2}},\\
\lang A_0 B_1 \rang &= -\frac{1}{\sqrt{2}},
\end{aligned}
\end{align}
so that we reach the so-called Tsirelson bound on correlations \cite{cirelson_quantum_1980}:
\begin{equation}
\lang A_0 A_1\rang + \lang A_0 B_1 \rang+ \lang A_1 B_0\rang + \lang A_1 B_1\rang =2 \sqrt{2}> 2.
\end{equation}
Therefore, Bell inequality \eqref{Bell-ineq} is violated for a specific choice of state and local measurements allowed by quantum theory, a result known under the name of EPR paradox \cite{einstein_can_1935}. Consequently, one needs to drop either free will, a choice called super-determinism, or the local realism hypothesis. The latter choice implies either a nonlocal hidden variables model such as Bohmian mechanics\footnote{Nonlocal hidden variable models are not compatible with relativity. For instance, Bohmian mechanics introduces a privileged foliation of spacetime, in clear violation of Lorentz invariance. Furthermore, a large class of nonlocal hidden variable models fulfill the Leggett inequality \cite{leggett_nonlocal_2003}, and a recent experiment has shown that this inequality is violated in accordance with the predictions of quantum theory \cite{groblacher_experimental_2007}.} or accepting the idea that a measurement does not reveal pre-existing values of the measured property, a choice that is at the basis of the standard or Copenhagen interpretation of quantum theory. This intepretation imposes a huge departure from our classical acception of the relation between systems and their properties. In the words of Mermin \cite{mermin_hidden_1993}:
\begin{quote}
`` [...] The outcome of a measurement is brought into being by the act of measurement itself, a joint manifestation of the state of the probed system and the probing apparatus. Precisely how the particular result of an individual measurement is brought into being---Heisenberg's `transition from the possible to the actual'---is inherently unknowable. Only the statistical distribution of many such encounters is a proper matter for scientific inquiry."
\end{quote}

Clearly, all quantum states respect the relativistic principle of locality in the Copenhagen interpretation, a concept that should not be confused with that of nonlocality of states in the LOSR context. To analyze this point more deeply, note that the hidden variable space corresponds to a source of shared randomness, and since parties cannot communicate, the above CHSH game can be understood using the LOSR paradigm. By definition, correlations that violate the Bell inequality are nonlocal in the LOSR sense, i.e. their statistics in the CHSH game cannot be reproduced using separable states. On the other hand, there are entangled states that do not violate the Bell inequality, and thus admit a local hidden variable model, i.e. their statistics can be reproduced using separable states  in the CHSH game. However, such states cannot be \emph{fully} generated from separable states using LOSR since even LOCC is not sufficient. Furthermore, Buscemi and later Rosset recently showed that \emph{all} entangled states are nonlocal in the LOSR sense by playing a \emph{distinct} game where the measurement settings of parties correspond to quantum systems to be measured jointly with their shared state\footnote{In the usual CHSH game, measurement settings are given by the values of classical bits.} \cite{buscemi_all_2012,rosset_correlations_2013}. Somehow, the CHSH game is not ``subtle" enough to detect the nonlocal character of all entangled states. The existence of nonlocality is already a suprising and useful feature of quantum theory, for example in quantum cryptography \cite{brunner_bell_2014}. But the truly remarkable fact---at least for the philosophy of science---is the existence of states that are \emph{sufficiently} nonlocal to violate the Bell inequality, and therefore do not conform to any local \emph{realistic} interpretation through a local hidden variable model\footnote{One can distill nonlocality under certain conditions \cite{forster_distilling_2009}, but there is strong evidence that one cannot distill a quantum system that admits a local hidden variable model into one that violates the Bell inequality \cite{dukaric_limit_2008}.}. Other forms of ``nonclassical" correlations exist in quantum theory, for instance those measured by quantum discord \cite{henderson_classical_2001,ollivier_quantum_2001}. Quantum discord has recently been interpreted as the difference in the performance of the quantum state merging protocol between a state and its locally decohered equivalent \cite{madhok_interpreting_2011}, or as quantifying the amount of entanglement consumption in the quantum state merging protocol \cite{cavalcanti_operational_2011}, and finally as the difference in superdense coding capacities between a quantum state and the best classical state when both are produced at a source that makes a classical error during transmission \cite{meznaric_quantifying_2012}. The role of quantum discord in a more general family of protocols has also been studied \cite{madhok_quantum_2013}.
For the purposes of our dissertation, we focus on states that violate the Bell inequality, generally refered to as nonlocal states in the literature. This term is rather unfortunate and dates back to the period where correlations that violate the Bell inequality were the only known example of nonlocal correlations\footnote{One may argue that this term is also confusing because all nonlocal correlations respect the relativistic locality principle.} in the modern LOSR sense, therefore the reader should keep these distinctions in mind.

In real world experiments, the ideal experimental protocol of particles leaving a source at definite times, and being measured at distant locations according to locally randomly chosen settings cannot be implemented. Therefore, experimental tests of the Bell inequality are subject to loopholes, i.e. problems of experimental design or setup that affect the validity of the experimental findings. The locality (or communication) loophole means that since in practical situations the two detections are separated by a timelike interval, the first detection may influence the second by some kind of signal. Avoiding this loophole imposes that the experimenter ensured a sufficient spatial separation of the particles, and that the measurement process is quick. The detection (or unfair sampling) loophole, which is due to the fact that particles are not always detected in both wings of the experiment, is a more serious challenge \cite{garg_detector_1987}: one can imagine that particles behave randomly, and by letting detection be dependent on a combination of local hidden variables and detector settings, the instruments can detect a subsample showing quantum correlations. An experiment showed that the Bell inequality can be violated while eliminating the detection loophole \cite{rowe_experimental_2001}, and loophole-free tests can be expected in the near future.

Our main goal in this section was to show how quantum theory respects the locality principle---or equivalently causality---imposed by the theory of relativity, and how the latter implies dropping the ``realism" part in the ``local realism" hypothesis. Therefore, we will not review in this dissertation related approaches in the topic of quantum foundations that also assume the realism hypothesis by using hidden variable models and that investigate various ways in which they contradict the predictions of quantum theory\footnote{Of course, different approaches do not necessarily rule out the same class of hidden variable models, but we here make the (reasonable) assumption that one can always build a set of observables and states that imply a contradiction between any class of hidden variable models and the predictions of quantum theory \cite{stairs_local_2007}.}\cite{spekkens_defense_2004,hardy_quantum_2004,pusey_reality_2012}.

\subsection{Reconstructions of quantum theory}
\label{reconstruct}

Bell's theorem helped clarify the sense in which nonlocal quantum correlations respect causality: to keep locality, one has to drop realism. Popescu and Rohrlich reversed the question and asked whether it is possible to derive the quantum bound on nonlocality from causality, a result that would provide a physically clear explanation of why quantum correlations exist \cite{popescu_quantum_1994,popescu_causality_1997}. They introduced an example of non-signalling models generalizing bipartite nonlocal quantum correlations, the so-called Popescu-Rohrlich (PR) boxes, and showed that they are maximally nonlocal, in the sense thay they maximally violate the Bell inequality\footnote{One can also show that PR-boxes correspond to an asymptotic unit of bipartite nonlocality \cite{forster_bipartite_2011}.}. Thus the answer to their question is no. However, these models share many of the nonclassical features of quantum theory, and thus shed new light on which properties are uniquely quantum. More recently, Paw\l owski \emph{et al.} introduced a principle called information causality which partially discriminates bipartite non-signalling quantum correlations from supra-quantum ones \cite{pawlowski_information_2009,allcock_recovering_2009}. This approach, triggered by developments in quantum information theory, belongs to a more general ``school" where one tries to derive (a subset of) the rules of quantum theory from clear informational principles. \emph{Reconstructing} quantum theory then means that one should look for clearly motivated constraints on the correlations between experimental records, such that they (partially) reproduce the predictions of the quantum formalism.

This approach is to be contrasted with the one that consists in \emph{modifying} some of the rules of quantum theory and comparing the predictions of the modified theory with those of the original. One then expects to gain a deeper understanding of standard quantum theory by isolating a subset of more fundamental features. Previous attempts include quaternionic models \cite{adler} or a model with non-linear terms in the Schrödinger equation \cite{weinberg2}. While the latter are notoriously problematic because of arbitrarily fast communication effects \cite{gisin_weinbergs_1990}, leading to the so-called ``preparation problem" \cite{cavalcanti_preparation_2012}, quaternionic quantum mechanics are formally well motivated: the propositional calculus---as formalized in the quantum logic framework---implies that it is possible to represent the pure states of a quantum system by rays on a Hilbert defined on any associative division algebra \cite{pitowski,pitowsky_quantum_2005}. This includes the quaternion algebra as the most general case \cite{soler_characterization_1995}. However, one of the main problems in the theory of quaternion quantum mechanics has been the construction of a tensor product of quaternion Hilbert modules, an essential requirement in order to formalize the notion of entanglement.

\subsubsection{PR-boxes}

The scenario again consists of two parties Alice and Bob that choose between two possible inputs $x,y$ with binary outputs $a$ and $b$ respectively\footnote{We below summarize the presentation of the scenario given in \cite{bub_quantum_2010}.}. Parties are allowed to confer on a joint strategy before they are separated. Once the game starts, they are not allowed further communication. The winning correlations are as follows:
\begin{itemize}
\item[(i)] $a\cdot b=0$ if $x,y$ are 00, 01, or 10,
\item[(ii)] $a\oplus b=0$ if $x,y$ are 11,
\end{itemize}where $\cdot$ denotes the Boolean product and $\oplus$ denotes the Boolean sum. The no-signalling condition can be defined using the marginal probabilities as follows:
\begin{align}
\begin{aligned}
\sum_{b\in\{0;1\}} p(a,b|x,y)&=p(a|x), \quad a,x,y\in \{0;1\},\\
\sum_{a\in\{0;1\}} p(a,b|x,y) &= p(b|y), \quad b,x,y\in \{0;1\},
\end{aligned}
\end{align}i.e., measurement results at Alice's laboratory are independent of Bob's choice of input, and vice versa. Alice and Bob are supposed to be symmetrical players, and we additionally require that the marginal probability of a particular output for a player should be independent of the the same player's input:
\begin{align}
\label{no-signalling}
\begin{aligned}
p(a=0|x=0)&=p(a=0|x=1)=p(b=0|y=0)=p(b=0|y=1),\\
p(a=1|x=0) &= p(a=1|x=1) = p(b=1|y=0)=p(b=1|y=1).
\end{aligned}
\end{align}Denote the marginal probability of output 1 by $p$ (See Table  \ref{table-chsh} for a summary of the winning correlations).
\begin{table}[h!]
\begin{center}
\begin{tabular}{|ll||ll|ll|} \hline
   &$x$&$0$ & &$1$&\\
   $y$&&&&&\\\hline\hline
  $0$ &&$p(00|00) = 1-2p$&$ p(10|00) = p$  & $p(00|10) = 1-2p$&$ p(10|10) = p$     \\
   &&$p(01|00) = p$&$p(11|00) = 0$  & $p(01|10)=p$&$ p(11|10) = 0$  \\\hline
   $1$&&$p(00|01)=1-2p$&$ p(10|01)=p$  & $p(00|11)=1-p$&$ p(10|11)=0$   \\
  &&$p(01|01)=p$&$ p(11|01)=0$  & $p(01|11)=0$&$ p(11|11)=p$   \\\hline
\end{tabular}
\end{center}
 \caption{Winning correlations for the game with marginal probability p for the outcome 1.}
 \label{table-chsh}
\end{table}The probability $p(00|00)$ is to be read as $p(a = 0,b = 0|x = 0,y = 0)$, etc.
Under the condition of random inputs and for a fixed strategy $S$, the probability of winning the game with marginal $p$ is:
\begin{equation}
p_{S}\mbox{(win)} = \frac{1}{4}[p_{S}(a\cdot b = 0|00) + p_{S}(a\cdot b = 0|01) + p_{S}(a\cdot b = 0|10) + p_{S}(a\oplus b = 0|11)].
\end{equation}
This probability can be related to the quantities that appear in the Bell inequality \eqref{Bell-ineq} by defining $\lang xy\rang_S$ as the expectation value, for the strategy $S$, of the product of the outputs for the input pair $x,y$, where now possible outputs take values $\pm 1$ instead of 0 or 1. We have:
\begin{align}
\begin{aligned}
\lang xy\rang_S &=p_S(1,1|xy)-p_S(1,-1|xy)-p_S(-1,1|xy)+p_S(-1,-1|xy)\\
& = p_S(\mbox{same outputs}|xy)-p_S(\mbox{different outputs}|xy).
\end{aligned}
\end{align}One can then show that:
\begin{equation}
p_S(\mbox{win})=\frac{1}{2}-\frac{K_S}{8}+\frac{3(1-2p)}{4},
\end{equation}where
\begin{equation}
K_S=\lang 00\rang_S +\lang 01\rang_S +\lang 10\rang_S -\lang 11\rang_S.
\end{equation}Bell's argument implies that if Alice and Bob share classical resources, i.e. correlations based on either shared randomness or common causes established before the game starts, then $|K_C|\leq 2$, where the low index $C$ stands for `classical'. A winning strategy $p_C(\mbox{win})=1$ is therefore impossible if $p>\frac{1}{3}$, and one can show that a wining classical strategy exists if $p\leq \frac{1}{3}$. If Alice and Bob share quantum resources, i.e. if they are allowed to perform measurements on shared entangled states prepared before the game starts, then the Tsirelson bound $|K_Q|\leq 2\sqrt{2}$ applies, and a winning quantum strategy $p_Q(\mbox{win})=1$ is impossible if $p>\frac{1+\sqrt{2}}{6}$.

Consider now the game for $p=\frac{1}{2}$, so that it is impossible to have a classical nor a quantum winning strategy. The winning correlations are as in Table \ref{1/2-game}. Under the condition of random inputs, the probability of winning the $p=\frac{1}{2}$ game is:
\begin{equation}
p_S(\mbox{win})=\frac{1}{2}-\frac{K_S}{8}.
\end{equation}Thus, the optimal classical and quantum success probabilities are:
\begin{equation}
p_{\mbox{optimal},C}(\mbox{win})=\frac{3}{4},\quad p_{\mbox{optimal},Q}(\mbox{win})=\frac{2+\sqrt{2}}{4}.
\end{equation} The winning correlations for the $p=\frac{1}{2}$ game are exactly those of a PR-box. They are usually represented as in Table \ref{PR-box} \footnote{Correlations of Table \eqref{1/2-game} are obtained by relabeling the $x$ input, the $a$ output conditionally on the $x$ input, the $y$ input, and the $b$ output conditionally on the $y$ input.}. Popescu and Rohrlich introduced the PR-box as a hypothetical device representing a nonlocal information channel that is more nonlocal than quantum mechanics. In fact, it maximally violates the Bell inequality:
\begin{equation}
K_{PR}=|\lang 00\rang_{PR} +\lang 01\rang_{PR} +\lang 10\rang_{PR} -\lang 11 \rang_{PR}| =4.
\end{equation}Correlations in Table \ref{PR-box} can be defined by the relation:
\begin{equation}
a\oplus b=x \cdot y,
\end{equation}with marginal probabilities equal to $\frac{1}{2}$ for all inputs and all outputs, i.e. PR-boxes are the optimal nonlocal devices respecting the no-signalling constraint.

\begin{table}[h!]
\begin{center}
\begin{tabular}{|ll||ll|ll|} \hline
   &$x$&$0$ & &$1$&\\
   $y$&&&&&\\\hline\hline
  $0$ &&$0$&$1/2$  & $0$&$1/2$     \\
   &&$1/2$&$0$  & $1/2$&$0$  \\\hline
   $1$&&$0$&$1/2$  & $1/2$&$0$   \\
  &&$1/2$&$0$&$0$&$1/2$   \\\hline
\end{tabular}
\end{center}
 \caption{Correlations for the $p=1/2$ game.}
 \label{1/2-game}
\end{table}

\begin{table}[h!]
\begin{center}
\begin{tabular}{|ll||ll|ll|} \hline
   &$x$&$0$ & &$1$&\\
   $y$&&&&&\\\hline\hline
  $0$ &&$1/2$&$0$  & $1/2$&$0$     \\
   &&$0$&$1/2$  & $0$&$1/2$  \\\hline
   $1$&&$1/2$&$0$  & $0$&$1/2$   \\
  &&$0$&$1/2$&$1/2$&$0$   \\\hline
\end{tabular}
\end{center}
 \caption{Correlations for the PR-box.}
 \label{PR-box}
\end{table}

\subsubsection{Polytope of non-signalling correlations}

The analysis of the previous section showed that bipartite classical and quantum correlations appear as only a subset of possible bipartite non-signalling correlations. The convex set of bipartite classical probability distributions forms the so-called `local polytope'  \cite{pitowski,barrett_nonlocal_2005}. It has 16 vertices defined as the non-signalling deterministic states and facets corresponding to the Bell inequalities. The fundamental fact is that the local polytope has the structure of a simplex, an $n$-simplex being defined as a convex polytope generated by $n+1$ vertices that are not confined to any $(n-1)$-dimensional subspace. For example, a triangle corresponds to a simplex while a rectangle does not. Therefore, the 16-vertex simplex represents the correlations polytope of probabilistic states of a bipartite \emph{classical} system with two binary-valued observables for each subsystem\footnote{One can associate a Boolean algebra to these classical probabilistic states, thus providing an alternative representation of the classical events structure.}.  Probability distributions over these extremal states define mixed states and are represented by points in the interior of the simplex. These points (representing an experiment statistics) are such that there exists a strategy with shared randomness (a local variable model) that produces the same probability distribution. If on the contrary a point lies outside the local polytope, then the experiment cannot be reproduced with shared randomness only. As stated previously, we call `nonlocal region' the region which lies outside the local polytope. The `quantum region', defined by the requirement that the probability distributions must be obtained from measurements on \emph{quantum} bipartite systems, contains the local polytope but is larger than it: measurements on quantum states can give rise to nonlocal correlations (the Bell inequalities are violated). The no-signalling polytope\footnote{This polytope can also be defined for no-signalling bipartite probability distributions with arbitrary inputs $x,y\in\{1,\cdots n\}$ and binary outputs 0 or 1.}, which contains the quantum region, is the set of bipartite probability distributions respecting the no-signalling constraint. Its vertices correspond to the 16 vertices of the classical simplex together with the 8 nonlocal vertices corresponding to the PR-box and the 7 other probability distributions obtained from the PR-box by relabeling the $x$ inputs, the $a$ outputs conditionally on the $x$ inputs, the $y$ inputs, and the $b$ outputs conditionally on the $y$ inputs (See Fig. \ref{polytope} and \cite{bub_why_2012} for more details).

\begin{figure}[!ht]
\label{fig-rindler}
\begin{center}
\includegraphics[scale=0.4]{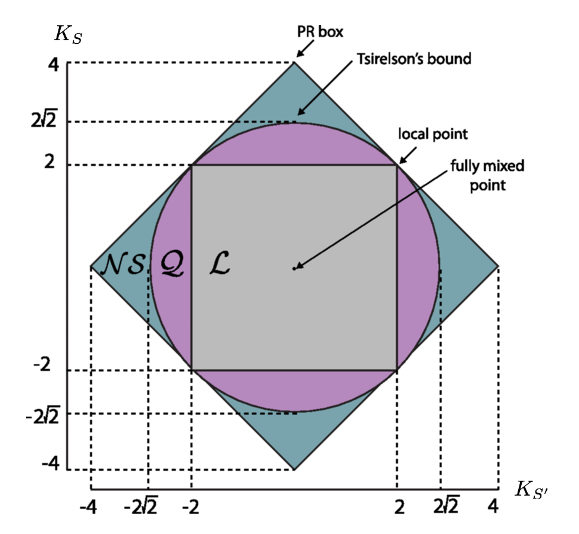}
\end{center}
\caption{A two-dimensional section of the no-signalling polytope in the Bell scenario with two inputs and two outputs. The vertical axis represents the degree of violation of the Bell inequality, while the horizontal axis represents the degree of violation the Bell inequality after relabeling of inputs. Local correlations satisfy $|K_S|\leq 2$ and $|K_{S'}|\leq 2$. The PR-box is the maximally nonlocal model, achieving the maximum violation $K_{PR}=4$. The Tsirelson bound corresponds to the point where $K_Q=2\sqrt{2}$, i.e. the maximum violation one can achieve using quantum systems and operations (From \cite{brunner_bell_2014}).}
\label{polytope}
\end{figure}

The polytope structure of correlations allows us to better understand the transition from classical to quantum theory. In the words of Bub \cite{bub_quantum_2010}:
\begin{quote}
``What is fundamental in the transition from classical to quantum physics is the recognition that information in the physical sense has new structural features, just as the transition from classical to relativistic physics rests on the recognition that spacetime is structurally different than we thought. Hilbert space, the event space of quantum systems, is interpreted as a kinematic (i.e., pre-dynamic) framework for an indeterministic physics, in the sense that the geometric structure of Hilbert space imposes objective probabilistic or information-theoretic constraints on correlations between events, just as the geometric structure of Minkowski spacetime in special relativity imposes spatio-temporal kinematic constraints on events."
\end{quote}Thus, if we were to highlight one ``new structural feature" of quantum correlations it would be the fact that they lie within a polytope that is not a simplex. The 1-simplex consists of two pure or extremal deterministic states
\begin{equation}
\vec{0} = \left( \begin{array}{c}1 \\ 0 \end{array} \right )\quad \mbox{and} \quad \vec{1} = \left ( \begin{array}{c}0 \\ 1 \end{array} \right ),
\end{equation}
 represented by the vertices of the simplex, with mixed states represented by the line segment between the two vertices:
\begin{equation}
\mathbf{p} = p\,\mathbf{0} + (1-p)\,\mathbf{1}
\end{equation}for $0 \leq p \leq 1$. This is simply the classical 1-bit space\footnote{One can think of the state space of classical physics as an infinite-dimensional simplex where extremal states correspond to deterministic states. The dynamics of a classical system correspond to the canonical transformations generated by Hamiltonians that act on its set of states, therefore the simplex associated to the space of states of classical physics must support a representation of these dynamics via transformations acting on its vertices \cite{bub_quantum_2010}.}. More generally, mixed states of a simplex have a unique decomposition over the vertices, i.e. the extremal states, and ``no other state space has this feature" \cite{bub_quantum_2010}. This property is at the basis of various other features that were first considered as specifically quantum but which turned out to be generic features of nonclassical (i.e. non-simplex theories). Examples of such nonclassical features include the impossibility of a universal cloning machine that can copy the extremal states of any probability distribution \cite{masanes_general_2006,barrett_information_2007}, the no-broadcasting theorem \cite{barnum_generalized_2007}, the monogamy of nonclassical correlations \cite{masanes_general_2006}, and information-disturbance trade-offs \cite{barrett_no_2005,scarani_secrecy_2006}.

\subsubsection{Characterizing the quantum region}

PR-boxes and other supra-quantum models share some nonclassical features with quantum theory, but they also have powerful communication and nonlocal computation properties unobserved in nature \cite{van_dam_nonlocality_2000,buhrman_nonlocality_2010,linden_quantum_2007}. One of the great unsolved problems in quantum information theory is to determine the boundary of quantum correlations. A possible approach is semi-definite programming which aims at deciding (preferably after a finite number of steps) whether a probability distribution is quantum or not \cite{wehner_tsirelson_2006,navascues_bounding_2007,navascues_convergent_2008}. Alternative approaches look for physical principles that would constrain general non-signalling bipartite correlations to the quantum region. Various principles were introduced among which relaxed uncertainty relations \cite{steeg_relaxed_2008,oppenheim_uncertainty_2010}, nonlocality swapping \cite{skrzypczyk_couplers_2009,skrzypczyk_emergence_2009}, macroscopic locality \cite{navascues_glance_2010} or information causality. The latter is able to partially characterize all bipartite (non-signalling) quantum correlations, thus providing a partial reconstruction of quantum theory.

To see how to arrive at this result, consider the following game: Alice and Bob are again two spacelike separated parties. At each round of the game, Alice receives $N=2^n$ random and independent bits $\vec{a}_N=(a_0,a_1,...,a_{N-1})$, and Bob receives a value of a random uniformly distributed variable $b\in\{0,...,N-1\}$. With the help of one classical bit communicated by Alice to Bob, Bob is required to guess the value of the $b$-th bit in Alice's list $a_b$ for some value $b\in\{0,1,...,N-1\}$. We denote Bob's guess by $g$. It is again assumed that Alice and Bob can communicate and plan a strategy before the game starts, but only one bit can be communicated once the game starts. They win the game if Bob always guesses correctly over any succession of rounds. If we are to quantify the strength of non-signalling correlations between Alice and Bob, Alice must decide which bit she sends to Bob independently of $b$ at each round. The efficiency of Alice's and Bob's strategy can be quantified by
\begin{equation}
I(N)=\sum_{k=0}^{N-1} I(a_k:g|b=k),
\end{equation}where $I(a_k:g|b=k)$ is the Shannon mutual information between $a_k$ and $g$ under the condition that Bob has received $b=k$.

Suppose now Alice and Bob share a supply of PR-boxes at each round. One can show that for any round and for any $b\in\{0,...,N-1\}$, Bob will be able to correctly guess the value of any designated bit in Alice's list. We below reproduce close to verbatim the adaptation of the original proof in \cite{pawlowski_information_2009} by Bub \cite{bub_why_2012}. Consider the simplest case $N=2$ where Alice receives two bits $a_0,a_1$. The strategy in this case involves a single shared PR-box. Alice inputs $a_0\oplus a_1$ into her part of the box and obtains an output $A$. She send the bit $x=a_0\oplus A$ to Bob, who inputs the value of $b$ into his part of the box and obtains an output $B$. Bob's final guess is $g=x\oplus B=a_0\oplus A\oplus B$. Now, correlations between inputs and outputs of a PR-box are such that $A\oplus B= a\cdot b=(a_0\oplus a_1)\cdot b$. So Bob's guess is $a_0\oplus((a_0\oplus a_1)\cdot b)$, therefore it follows that if $b=0$ Bob correctly guesses $a_0$, and if $b=1$, Bob correctly guesses $a_1$. Suppose now Alice's set is composed of four bits $a_0,a_1,a_2,a_3$ ($N=4$)(See Fig. \ref{information-caus}) \footnote{The winning strategy for $N=4$ also applies for $N=3$, therefore it is justified to restrict our attention to $N=2^n$.}. Bob's random variable can be specified by two bits $b_0,b_1$:
\begin{equation}
b=b_02^0+b_1 2^1.
\end{equation}The strategy in this case involves an inverted pyramid of PR-boxes: two shared PR-boxes $L$ and $R$ at the first stage, and one shared PR-box at the final second stage. Alice inputs $a_{0}\oplus a_{1}$ into the $L$ box and gets an output $A_L$, and inputs $a_{2}\oplus a_{3}$ into the $R$ box and gets an output $A_R$.  Bob inputs $b_{0}$ into both the $L$ and $R$ boxes and we consider the output $B_{0}$ of one of the boxes depending on what bit Bob is required to guess. At the second stage, Alice inputs $(a_{0}\oplus A_{L})\oplus(a_{2}\oplus A_{R})$ and obtains the output $A$. Bob inputs $b_{1}$ into this box and obtains the output $B_{1}$.  Alice then sends to Bob the bit $x = a_{0}\oplus A_{L}\oplus A$. Now, Bob could correctly guess either $a_{0}\oplus A_{L}$ or $a_{2}\oplus A_{R}$, using the elementary $N = 1$ strategy, as $x\oplus B_{1} = a_{0}\oplus A_{L}\oplus A\oplus B_{1}$. Here $A\oplus B_{1} = (a_{0}\oplus A_{L}\oplus a_{2}\oplus A_{R})\cdot b_{1}$. If $b_{1} = 0$, Bob would guess $a_{0}\oplus A_{L}$. If $b_{1} = 1$, Bob would guess $a_{2}\oplus A_{R}$. So if Bob is required to guess the value of $a_{0}$ (i.e., $b_{0} = 0, b_{1} = 0$) or $a_{1}$ (i.e., $b_{0} = 1, b_{1} = 0$)---the input to the PR-box $L$---he guesses $a_{0} \oplus A_{L} \oplus A \oplus B_{1} \oplus B_{0}$, where $B_{0}$ is the Bob-output of the $L$ box. Then:
\begin{eqnarray}
a_{0} \oplus A_{L} \oplus A \oplus B_{1} \oplus B_{0} & = & a_{0} \oplus A_{L} \oplus B_{0} \nonumber \\
& = & a_{0} \oplus (a_{0} \oplus a_{1})\cdot b_{0}.
\end{eqnarray}
If $b_{0} =0$, Bob correctly guesses $a_{0}$, and if $b_{0} =1$, Bob correctly guesses $a_{1}$. If Bob is required to guess the value of $a_{2}$ (i.e., $b_{0} = 0, b_{1} = 1$) or $a_{3}$ (i.e., $b_{0} = 1, b_{1} = 1$)---the input to the PR-box $R$---he guesses $a_{0} \oplus A_{L} \oplus A \oplus B_{1} \oplus B_{0}$, where $B_{0}$ is the Bob-output of the $R$ box. Then:
\begin{eqnarray}
a_{0} \oplus A_{L} \oplus A \oplus B_{1} \oplus B_{0} & = & a_{2} \oplus A_{R} \oplus B_{0} \nonumber \\
& = & a_{2} \oplus (a_{2} \oplus a_{3})\cdot b_{0}.
\end{eqnarray}
If $b_{0} =0$, Bob correctly  guesses $a_{2}$ and if $b_{0} =1$, Bob correctly guesses $a_{3}$. This procedure can be generalized to any value of $N$ by using an inverted pyramid of $N(N+1)/2$ PR-boxes\footnote{Note that the game can be modified to allow Alice to send $m$ classical bits of information to Bob, in which case Bob is required to guess $m$ bits from Alice's list. In this case, the procedure we detailed applies by using $m$ inverted pyramids of PR-boxes at each round.}. 

\begin{figure}[!htbp]
\label{fig-rindler}
\begin{center}
\includegraphics[scale=0.25]{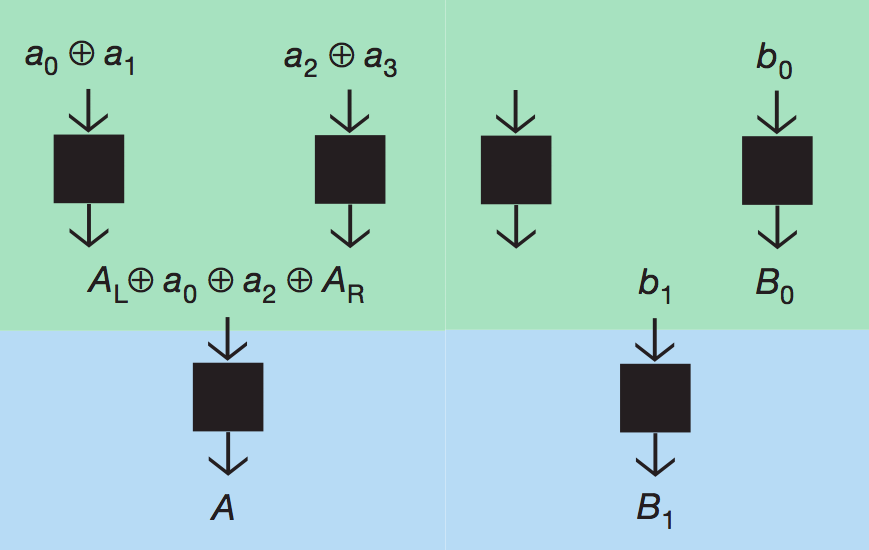}
\end{center}
\caption{Information causality identifies the strongest quantum correlations: the protocol when Alice's set has size $N=4$ (From \cite{pawlowski_information_2009}).}
\label{information-caus}
\end{figure}

One can now use this result to retrieve the Tsirelson bound under the condition that mutual information verifies three ``natural" properties \cite{pawlowski_information_2009}. Denote by $p=\frac{1+E}{2}$ the probability of simulating a PR-box, where $E$ depends on the nature of the non-signalling (NS) box. Consider the $N=4$ game where Alice and Bob share NS-boxes, and Alice is allowed to communicate one bit to Bob. Bob's guess $x \oplus B_{1} \oplus B_{0}$ will be correct if $B_{1}$ and $B_{0}$ are both correct or both incorrect. The probability of being correct at both stages is:
\begin{equation}
\frac{1}{2}(1+E)\cdot \frac{1}{2}(1+E) = \frac{1}{4}(1+E)^{2}.
\end{equation}
The  probability of being incorrect at both stages is:
\begin{equation}
(1-\frac{1}{2}(1+E))\cdot(1-\frac{1}{2}(1+E)) = \frac{1}{2}(1-E)\cdot \frac{1}{2}(1-E) = \frac{1}{4}(1-E)^{2}.
\end{equation}
So the probability $P_{success}$ that Bob guesses correctly, i.e., the probability that $g = a_{k}$ when $b=k$, is:
\begin{equation}
P_{k} = \frac{1}{4}(1+E)^{2}+\frac{1}{4}(1-E)^{2} = \frac{1}{2}(1+E^{2}).
\end{equation}
In the general case $N=2^{n}$, Bob guesses correctly if he makes an even number of errors over the $n$ stages ($B_{0}, B_{1}, B_{2}, \ldots$). The probability that $g = a_{k}$ when $b=k$ is then:
\begin{equation}
P_{k} = \frac{1}{2^{n}}(1+E)^{n}    +  \frac{1}{2^{n}} \sum_{j=1}^{\lfloor\frac{n}{2}\rfloor} {n \choose 2 j} (1-E)^{2j}  (1+E)^{n-2j}         = \frac{1}{2}(1+E^{n}),
\end{equation}
where $\lfloor \frac{n}{2} \rfloor$ denotes the integer value of $\frac{n}{2}$.

Assume now that mutual information obeys the following properties in the toy universe of NS boxes:
\begin{itemize}
\item[(i)] \emph{Consistency:} if two systems $A$ and $B$ are both classical, then $I(A : B)$ should coincide with Shannon's mutual information.
\item[(ii)] \emph{Data processing inequality:} if $B\rightarrow B'$ is a permissible map between systems, then $I(A : B) \geq I(A : B')$. This amounts to saying that any local manipulation of data can only decay information.
\item[(iii)] \emph{Chain rule:} there exists a conditional mutual information $I(A : B|C)$ such that the following identity is satisfied for all states and triples of parts: $I(A : B, C) = I(A : C) + I(A : B|C)$.
\end{itemize}
Under a local randomization assumption, one can then show that\footnote{The local randomization assumption is meant to simplify computations. One can drop this assumption at the end of calculations (See \cite{pawlowski_information_2009}).}:
\begin{equation}
\label{ineq-pro}
\frac{(2E^2)^n}{2\ln(2)} \leq I(N).
\end{equation}Furthermore, in the case of independent Alice's input bits, using the chain rule we have
\begin{equation}
I(a_0,...,a_{N-1}:x,B)=I(a_0:x,B)+I(a_1,...,a_{N-1}:x,B|a_0),
\end{equation} where $B$ denotes Bob's part of previously shared correlations\footnote{In case $m$ bits of communication are allowed, the bit $x$ should be replaced by a vector of $m$ bits $\vec{x}_m$.}. Using again the chain rule for the second term on the right-hand side, we have:
\begin{equation}
I(a_1,...,a_{N-1}:x,B|a_0)=I(a_1,...,a_{N-1}:x,B,a_0)-I(a_1,...,a_{N-1}:a_0),
\end{equation}
 and since Alice's inputs are independent, $I(a_1,...,a_{N_1}:a_0)=0$. Applying the data processing inequality to the first term then implies:
\begin{equation}
I(\vec{a}_N:x,B)\geq I(a_0:x,B)+I(a_1,...,a_{N-1}:x,B).
\end{equation}Iterating this procedure, we have:
\begin{equation}
I(\vec{a}_N:x,B)\geq \sum_{k=0}^{N-1} I(a_k:x,B).
\end{equation}Finally, we observe that Bob's guess $g$ is obtained from $b,x$ and $B$. Hence, the data processing inequality implies that $I(a_k:g|b=k)\leq I(a_k:x,B)$. Therefore, we have:
\begin{equation}
I(N) \leq I(\vec{a}_N:x,B).
\end{equation}Now:
\begin{eqnarray}
\label{lim-mes}
I(\vec{a}_N:x,B) &=&I(\vec{a}_N:B)+I(\vec{a}:x|B)\nonumber\\
&=&I(\vec{a}_N:x|B)\nonumber\\
&=&I(x:\vec{a}_N,B)-I(x:B)\nonumber\\
& \leq & I(x:\vec{a}_N,B)\nonumber\\
&\leq & I(x:x)=1,
\end{eqnarray}a property named ``information causality" by Pawlowski \emph{et al.} \cite{pawlowski_information_2009}. Using \eqref{ineq-pro} and \eqref{lim-mes}, one can show that $E\leq \frac{1}{\sqrt{2}}$, i.e. the Tsirelson bound on non-signalling correlations is derived from three natural assumptions on the behavior of mutual information.

A related recent result shows that two concepts of entropy can be defined in the context of general probabilistic theories: measurement entropy and mixing entropy, which coincide with the Shannon entropy in the context of classical theory and with the von Neumann entropy in quantum theory respectively. It appears that for any non-signalling theory where these two entropies coincide, measurement entropy is strongly subadditive and a version of the Holevo bound is satisfied (or equivalently the data processing inequality is satisfied) and the theory is informationally causal. It follows that monoentropic theories, like quantum theory, obey the Tsirelson bound \cite{barnum_entropy_2010}. It is still an open question whether information causality can completely retrieve the quantum region. Moreover, it was shown that all bipartite principles fail to account for correlations with more than two parties \cite{gallego_quantum_2011}.  Thus, the information causality principle has fundamental limitations and intrinsically multipartite information concepts are needed for the full understanding of quantum correlations.

We complete this section by recalling that various other attempts to reconstruct the quantum formalism exist. For instance, some partial reconstructions focus on explaining the set of contextual quantum correlations instead of nonlocal ones \cite{cabello_graph-theoretic_2014,amaral_exclusivity_2014}. More ambitious attempts exist that aim at retrieving \emph{all} the postulates of quantum theory from physically clear principles \cite{rovelli_relational_1996,zeilinger_foundational_1999,brukner_malus_1999,hardy_quantum_2001,dakic_quantum_2009,chiribella_informational_2011,chiribella_quantum_2012}.
In these reconstructions, it appears that the transition from classical to quantum probabilities is deeply connected to the existence of \emph{continuous} reversible transformations between the states of the system.

\section{Entanglement and open systems}
\label{entanglement-AQFT}

The distinction between separability and locality, or equivalently between their opposites\footnote{The opposite of ``separable" would be ``nonseparable", but separable states can be built using various constraints, therefore a nonseparable state does not specify what kind of resource we are considering. For instance, nonseparable states under LOCC are entangled while nonseparable states under LOSR are nonlocal. As stated previously, we are focusing in this dissertation on states that violate the Bell inequality.} nonlocal ``outcome-outcome" correlation and ``measurement-outcome" correlation respectively, is crucial to unraveling the conceptual implications of Bell's theorem. In nonrelativistic quantum theory and quantum field theory (QFT), ``measurement-outcome" correlations are excluded because of the commutation relations between the observables associated to distinct and spacelike separated systems respectively, while entanglement can give rise to ``outcome-outcome" correlations. Such correlations puzzled Einstein, as it appears in the following famous passage \cite{einstein_quanten-mechanik_1948}:

\begin{quote}
``If one asks what is characteristic of the realm of physical ideas independently of the quantum theory, then above all the following attracts our attention: the concepts of physics refer to a real external world, i.e. ideas are posited of things that claim a `real existence' independent of the perceiving subject (bodies, fields, etc.) [...] It appears to be essential for this arrangement of the things in physics that, at a specific time, these things claim an existence independent of one another, insofar as these things ``lie in different parts of space". Without such an assumption of the mutually independent existence (the ``being-thus") of spatially distant things, an assumption which originates in everyday thought, physical thought in the sense familiar to us would not be possible. Nor does one see how physical laws could be formulated and tested without such clean separation [...] For the relative independence of spatially distant things ($A$ and $B$), this idea is characteristic: an external influence on $A$ has no immediate effect on $B$; this is known as the ``principle of local action", which is applied consistently in field theory. The complete suspension of this basic principle would make impossible the idea of the existence of (quasi-)closed systems and, thereby, the establishment of empirically testable laws in the sense familiar to us."
\end{quote}

Following Halvorson \cite{halvorson_locality_2001}, we assume that Einstein's ``principle of local action" corresponds to locality while the ``assumption of the mutually independent existence (the `being-thus') of spatially distant things" corresponds to separability. We also assume that Einstein relates the thermodynamical concept of open systems to the possibility for quantum systems to sustain nonlocal ``outcome-outcome" correlations with other quantum systems when they are in an entangled state rather than ``measurement-outcome" correlations. Therefore, Einstein seems to point out that separability of systems is threatened by entanglement.

Clifton and Halvorson argued in \cite{clifton_entanglement_2001} that such a worry is unjustified from the point of view of nonrelativistic quantum theory. Their argument goes as follows\footnote{This is a close to verbatim summary of various definitions and proofs in \cite{clifton_entanglement_2001}.}: consider the simplest toy universe consisting of two nonrelativistic quantum systems, represented by a tensor product of two-dimensional Hilbert spaces $\mathbb{C}^2(A)\otimes \mathbb{C}^2(B)$, where system $A$ is the `object' system, and $B$ its `environment'. Let $x$ be any state vector for the composite system $AB$, and $\rho_A(x)$ be the reduced density operator $x$ determines for system $A$. The von Neumann entropy of $A$, $S_A(x) = -\trace(\rho_A(x)\ln \rho_A(x))$, is a measure of the degree to which $A$ and $B$ are entangled: if $x$ is a product vector with no entanglement, $S_A(x) = 0$, whereas $S_A(x) = \ln2$ when $x$ is, say, a singlet state. Only the presence of a nontrivial interaction Hamiltonian between members of the pair can change their degree of entanglement, and hence the entropy of either system $A$ or $B$. Alice can destroy entanglement of system $A$ with $B$ by performing a standard von Neumann measurement on $A$. If $P_\pm$ are the eigenprojections of the observable Alice measures, and the initial density operator of $AB$ is $\rho = P_x$ (where $P_x$ is the projection onto the ray $x$ generates), then the post-measurement joint state of $AB$ will be given by the new density operator:
\begin{equation}
\rho \rightarrow \rho' = (P_+\otimes \id)P_x(P_+\otimes \id)+(P_-\otimes \id)P_x(P_-\otimes \id).
\end{equation}
Since the projections $P_\pm$ are one-dimensional, and $x$ is entangled, there are nonzero vectors $a_x^\pm\in\mathbb{C}^2_A$ and $b_x^\pm\in\mathbb{C}^2_B$ such that $(P_\pm\otimes \id)x=a_x^\pm\otimes b_x^\pm$, and one can straightforwardly show that:
\begin{equation}
\rho'=\trace\left[ (P_+\otimes \id)P_x\right] P_+\otimes P_{b_x^+}+\trace\left[ (P_-\otimes \id)P_x\right] P_-\otimes P_{b_x^-}.
\end{equation}
Thus, $\rho'$ will always be a convex combination of product states independently of the degree of entanglement of the initial state $x$, and there will no longer be any entanglement between $A$ and $B$. This result can be generalized as follows. Given any finite- or infinite-dimensional Hilbert spaces $\mathcal{H}(A)$ and $\mathcal{H}(B)$, it suffices for Alice to measure any non-degenerate observable of $A$ with a discrete spectrum to destroy the entanglement with $B$ of the initial (pure or mixed) state $\rho$. The final state $\rho'$ will then be a convex combination of product states, each of which is a product density operator obtained by ``collapsing" $\rho$ using some particular eigenprojection of the measured observable. They conclude that ``Alice's operation on $A$ has the effect of isolating $A$ from any further EPR influences from $B$". Thus, it seems that the methodological issue Einstein is trying to rise does not concern (non-interacting) nonrelativistic systems. 

What about relativistic systems? Various results show that quantum field systems are unavoidably and intrinsically open to entanglement, even in the free case, in the sense that one cannot ``efficiently" disentangle subsystems associated with spacelike separated regions of spacetime through local operations. As we shall see in this section, this behavior is deeply rooted in the algebraic structure of QFT. We start by providing a quick review of the algebraic quantum field theory (AQFT) formalism, a mathematically precise description of the structure of quantum field theories incorporating the idea of domain localization, and analyze how algebras of infinite type and the Reeh-Schlieder theorem entail severe practical obstacles to isolating spacelike separated regions from entanglement with other field systems. We also analyze how the type III property of the local algebras associated with localized subsystems imposes a fundamental limitation on isolating field systems from entanglement in relativistic QFT, and the importance of the split property for neutralizing Einstein's methodological worry.

\subsection{Algebraic quantum field theory}
\label{section-121}

AQFT is a general mathematical framework for QFTs incorporating the idea of domain localization \cite{haag_algebraic_1964,haag}. In this approach, subsystems of a given system are associated with subalgebras of a given algebra. States are secondary objects that arise via Hilbert space representations, or alternatively as linear functionals on the algebra of observables defining the system. Since linear functionals can be interpreted as expectation values, states are assumed to be positive and normalized\footnote{See Appendix \ref{algebra-theory} for details on the theory of operators.}. AQFT thus shifts the focus from the level of states (at which the nonlocal features such as entanglement appear) to the level of observables by basing the theory on the algebra of observables, not on the Hilbert space of states, an approach that clarifies in what sense QFTs do not violate the relativistic principle of locality \cite{brunetti_algebraic_2006}. 

\subsubsection{Abstract representation of a system}

A nontrivial conceptual step consists in identifying the algebraic structure of possible operations on a physical system $S$ with a C$^*$-algebra $\mathfrak{A}$, and its state with a state on $\mathfrak{A}$. This assumption can be heuristically related to the fact that one can perform operations on this system, including measurements whose results can be added, mutliplied, substracted etc. and encode information concerning the physical state of $S$. The first assumption can be weakened to Jordan-Banach or to Segal algebras, which then leads to loosing much of the deductive power of the theory. The second input is more peculiar and often overlooked: linearity and positivity can be related to the algebraic structure of the observer's statistics, not to an intrinsic property of $S$. Therefore, such an algebraic approach hides in two seemingly innocent axioms much of the conceptual difficulties of quantum theory \cite{grinbaum_significance_2004}.

In most situations one associates a specific kind of C$^*$-algebra to the system called a von Neumann algebra. Formally, the von Neumann algebra $\mathfrak{M}$ associated to the system can be taken as the universal enveloping von Neumann algebra of $\mathfrak{A}$, a further nontrivial step regarding the information-theoretic interpretation of the algebraic framework\footnote{A heuristic justification can be given in view of the spectral theorem: it is the theory of von Neumann algebras, not C$^*$-algebras, that is the natural extension of classical measure/probability theory (See Appendix \ref{algebra-theory}).}.

The specific properties of individual models of any relativistic QFT (particle structure, scattering cross sections, bound states, superselection charges...) are then encoded in the net structure
\begin{equation}
\mathcal{O}\mapsto \mathfrak{A}(\mathcal{O}),
\end{equation}
i.e. in the way C$^*$-algebras representing local systems are embedded into each other for different subsets $\mathcal{O}$ ($\calO$ for ``Open" or ``Observer") of Minkowski spacetime $M$ \cite{brunetti_algebraic_2006}.
One naturally assumes
\begin{itemize}
\item[]\emph{Isotony:} If $\calO_1 \subseteq \calO_2$, then $\mathfrak{A}(\calO_1)\subseteq \mathfrak{A}(\calO_2).$
\end{itemize}As a consequence, the collection of all local algebras $\mathfrak{A}(\calO)$ defines a net whose limit points can be used to define algebras associated with unbounded regions, and in particular $\mathfrak{A}(M)$, which is identified with $\mathfrak{A}$ itself. Other conditions that any physically reasonable QFT should satisfy include:
\begin{itemize}
\item[]\emph{Microcausality:} $\mathfrak{A}(\calO')\subseteq \mathfrak{A}(\calO)'$.
\item[]\emph{Translational covariance:} One assumes that there is a faithful representation $\vec{x} \mapsto \alpha_\vec{x}$ of the spacetime translation group of $M$ in the group of automorphisms of $\mathfrak{A}$ such that
\begin{equation}
\alpha_\vec{x}(\mathfrak{A}(\calO))=\mathfrak{A}(\calO+\vec{x}).
\end{equation}
\item[]\emph{Weak additivity:} For any $\calO\subseteq M$, $\mathfrak{A}$ is the smallest von Neumann algebra containing $\bigcup_{\vec{x}\in M} \mathfrak{A}(\calO+\vec{x})$.
\item[]\emph{Spectrum condition:} $\mathfrak{A}$ is represented as $\mathfrak{B}(\calH)$ for some Hilbert space $\calH$, and the generator of spacetime translations, the energy-momentum of the field, has a spectrum confined to the forward light-cone.
\end{itemize}
Brunetti and Fredenhagen argued in \cite{brunetti_algebraic_2006} that this mapping between regions of spacetime and algebras naturally asks for a covariant formalism ``in the spirit of general relativity". Following Einstein's `point coincidence' argument against the `hole problem', systems corresponding to isometric regions must be isomorphic, and ``since isometric regions may be embedded into different spacetimes", their conclusion is that category theory is the natural formalism to tackle this problem. Objects of the category then correspond to the systems and the morphisms to the embeddings of a system as a subsystem of other systems. The framework of AQFT may then be described as a covariant functor $\euA$ between two categories \cite{brunetti_generally_2001,fredenhagen_locally_2004, brunetti_algebraic_2006}. Brunetti and Fredenhagen define a first category, denoted $\Loc$ for `locality category', that ``contains the information on local relations [between systems] and is crucial for the interpretation" \cite{brunetti_algebraic_2006}. Objects of this category are topological spaces with additional structures and its morphisms structure preserving embeddings. They provide examples of such categories connected to well-known physical examples (globally hyperbolic Lorentzian spaces under the condition that the embeddings are isometric and preserve the causal structure, spin bundles with connections, etc.). The formal definition of $\Loc$ is the following:

\begin{description}
\item[$\Loc$] The class of objects
$\obj(\Loc)$ is formed by all (smooth) $d$-dimensional
($d\ge 2$ is held fixed), globally
hyperbolic Lorentzian spacetimes $M$ which are oriented and time-oriented.
Given any two such objects $M_1$ and $M_2$, the morphisms
$\psi\in\Hom_{\Loc}(M_1,M_2)$
are taken to be the isometric embeddings
$\psi: M_1 \to M_2$ of $M_1$ into
$M_2$ but with the following constraints:
\begin{itemize}
\item[$(i)$] If $\gamma : [a,b]\to M_2$ is any
causal curve and
$\gamma(a),\gamma(b)\in\psi(M_1)$ then the whole curve must be
in the image $\psi(M_1)$, i.e. $\gamma(t)\in\psi(M_1)$ for all $t\in ]a,b[$.
\item[$(ii)$] Any morphism preserves orientation and time-orientation of the embedded spacetime.
\end{itemize}
Composition is composition of maps, the unit element in $\Hom_{\Loc}(M,M)$ is given by
the identical embedding ${\rm id}_M : M\mapsto M$ for
any $M\in\obj(\Loc)$.
\end{description}
The second category, denoted $\Obs$ for `observables category', describes the algebraic structure of  observables, common to all systems.

\begin{description}
\item[$\Obs$] The class of objects
  $\obj(\Obs)$ is formed by all C$^*$-algebras possessing unit
  elements, and the morphisms are faithful (injective) unit-preserving
  $*$-homomorphisms.
  The composition
  is again defined as the composition of maps,
  the unit element in
  $\Hom_{\Obs}(\calA,\calA)$ is for any $\calA \in \obj(\Obs)$ given
  by the identical map ${\rm id}_{\calA}: A \mapsto A$, $A\in\calA$.
\end{description}
One can use alternative categories according to particular needs. For instance, von Neumann algebras are better suited when selection of a state is possible. More examples are analyzed in \cite{brunetti_algebraic_2006} :
\begin{quote}
 ``In classical physics one looks instead at Poisson algebras, and in perturbative quantum field theory one admits algebras which possess nontrivial representations as formal power series of Hilbert space operators, or more general topological $*$-algebras [...] On the other side, one might consider for $\Loc$ bundles over spacetimes, or one might (in conformally invariant theories) admit conformal embeddings as morphisms [...] In case one is interested in spacetimes which are not globally hyperbolic one could look at the globally hyperbolic subregions (where care has to be payed to the causal convexity condition (i) above)."
 \end{quote}
The leading principle in AQFT is that the functor $\euA$---on which further ``natural" constraints must be imposed such as isotony or microcausality---contains all physical information. In particular, one defines two theories as being equivalent iff the corresponding functors are naturally equivalent. A familiar object such as a field $F$ arises as a natural transformation from the category of test function spaces to the category of observable algebras via their functors related to the locality category. More generally, natural transformations from functors on the locality category are central to the understanding of functor $\euA$. Since we are not primarily concerned with a covariant formulation of AQFT, we refer the reader to \cite{brunetti_generally_2001,fredenhagen_locally_2004, brunetti_algebraic_2006} and the references therein for more details.

\subsubsection{Concrete representation of a system}

The usual Hilbert space formalism is derived from the GNS representation theorem, and for a given system inequivalent representations can exist. We recall the definition of unitary equivalence of representations of a C$^*$-algebra $\mathfrak{A}$: two representations $\pi$ and $\phi$ are said to be unitarily equivalent iff there exists an isometry $U:\mathcal{H}_\pi \longrightarrow \mathcal{H}_\phi$ such that:
\begin{equation}
U\pi(A) U^{-1}=\phi(A),\quad \forall A \in \mathfrak{A}.
\end{equation}
If representations $(\pi_i,\calH_i)$ are unitarily equivalent to some fixed representation $(\pi,\calH)$, we say that $\phi=\bigoplus_i \pi_i$ is a multiple of the representation $\pi$. Two representations $\pi$ and $\phi$, where $\pi$ is irreducible and $\phi$ factorial\footnote{A representation $\phi$ is factorial iff the von Neumann algebra $\phi(\mathfrak{A})''$ is a factor.}, are said to be quasi-equivalent  iff $\phi$ is a multiple of $\pi$. They are said to be disjoint iff they are not quasi-equivalent\footnote{The general definition is more cumbersome, see \cite{halvorson_locality_2001}.}.

\begin{Def}
Let $\omega$ be a state on a C$^*$-algebra $\mathfrak{A}$ inducing a GNS representation $\pi_\omega:\mathfrak{A} \rightarrow \mathfrak{B}(\calH_\omega)$ of $\mathfrak{A}$ on a Hilbert space $\mathcal{H}_\omega$. A state $\mu$ on $\mathfrak{A}$ is said $\omega$-normal iff there exists a density matrix $\rho\in\mathfrak{B}(\mathcal{H}_\omega)$ such that:
\begin{equation}
\mu(A)=\trace \left(\rho\pi_\omega(A)\right),\quad \forall A\in\mathfrak{A}.
\end{equation}
The folium of the representation $\pi_\omega$, denoted by $\mathfrak{F}(\pi_\omega)$, corresponds to the set of $\omega$-normal states.
\end{Def}
 Intuitively, the folium of a state representation $\pi_\omega$ corresponds to the set of (pure or mixed) states that ``live" in the same Hilbert space as the state $\omega$ that generated the representation. We have the following equivalences:
\begin{align*}
\pi \,\, \mbox{and}\,\, \phi \,\, \mbox{are quasi-equivalent}\quad   &\Longleftrightarrow \quad \mathfrak{F}(\pi)=\mathfrak{F}(\phi),\\
\pi\,\, \mbox{and} \,\,\phi\,\, \mbox{are disjoint} \quad &\Longleftrightarrow\quad \mathfrak{F}(\pi)\cap \mathfrak{F}(\phi)=0.
\end{align*}
Let $\mathfrak{A}$ be a C$^*$-algebra, $\pi$ a representation of $\mathfrak{A}$ and $\mathfrak{F}(\pi)$ its folium.
We say that a state $\omega$ on $\mathfrak{A}$ is weak*-approximated\footnote{One can define the $\sigma(\mathfrak{A}^*,\mathfrak{A})$-topology on the space of states, also called ultraweak or weak*-topology (See Appendix \ref{algebra-theory}).} by a state in $\mathfrak{F}(\pi)$ iff for each $\epsilon>0$ and for each finite collection $\{A_i:i=1,...,n\}$ of operators in $\mathfrak{A}$, there is a state $\omega'\in\mathfrak{F}(\pi)$ such that
\begin{equation}
|\omega(A_i)-\omega'(A_i)|<\epsilon,\quad i=1,...,n.
\end{equation}
Two representations $\pi,\phi$ are then said to be weakly equivalent iff all states in $\mathfrak{F}(\pi)$ may be weak*-approximated by states in $\mathfrak{F}(\phi)$ and vice versa.

In summary, we have the following implications for any two representations $\pi,\phi$:
\begin{center}
Unitarily equivalent $\Longrightarrow$ Quasi-equivalent $\Longrightarrow$ Weakly equivalent.
\end{center}If $\pi$ and $\phi$ are both irreducible, then the first arrow is reversible.

The following theorem is crucial to understanding the practical implications of the existence of inequivalent representations:
\begin{Thm}[Fell's theorem]
Let $\mathfrak{A}$ be a C$^*$-algebra and $\pi_1, \pi_2$ two $*$-representations on a same Hilbert space $\mathcal{H}$. Any state of $\pi_1$ is a limit for the ultraweak topology of a state in $\pi_2$ iff $\ker(\pi_2)$ is an ideal of $\ker(\pi_1)$.
\end{Thm}If a representation $\pi$ of a C$^*$-algebra is faithful, this theorem implies that any state on $\mathfrak{A}$ is weak*-approximated by states in $\mathfrak{F}(\pi)$. In particular, it follows that \emph{all} representations of $\mathfrak{A}$ are weakly equivalent. Adopting a practical point of view\footnote{See \cite{halvorson_locality_2001} for a detailed analysis of physical equivalence of representations.}, we follow Haag and Kastler \cite{haag_algebraic_1964} who have argued that the weak equivalence of all representations of the C$^*$-algebra entails their physical equivalence. Their argument, which entails an operational definition of physical equivalence, is based on the fact that measuring the expectations of a finite number of observables $\{A_i\}$ in the C$^*$-algebra, each to a finite degree of accuracy $\epsilon$, we can determine only a weak*-neighborhood of the algebraic state. But by Fell's theorem, states from the folium of any representation lie in this neighborhood. So we can never determine \emph{in practice} which representation is the physically ``correct" one and they all, from a practical point of view, carry the same physical content. As a corollary, choosing a representation is therefore simply a matter of convention.

Note that the assumption of finite degree of accuracy in discriminating states amounts to saying that one can only extract a finite amount of information from the physical system through a finite sequence of measurements. This in turn means that if we associate a von Neumann algebra $\mathfrak{M}$ to the system, $\mathfrak{M}$ should be the ultraweak closure of finite-dimensional algebras, i.e. a hyperfinite algebra. One can show that the general postulates of quantum mechanics and relativity, together with assumptions about the existence of a scaling limit and bounds on the local density of states, imply a unique structure of the local algebras: they are isomorphic to the unique hyperfinite type III$_1$ factor \cite{haag_algebraic_1964,haag,yngvason_role_2005}.

\subsection{Infinite algebras and the Reeh-Schlieder theorem}

We start by defining the notion of Bell correlation in AQFT. Let $\calH$ be a Hilbert space, let $\calS$ denote the set of unit vectors on $\calH$ and let $\mathfrak{B}(\calH)$ denote the set of bounded linear operators on $\calH$. If $x\in\calS$, we let $\omega_x$ denote the state of $\mathfrak{B}(\calH)$ induced by $x$. Let $\mathfrak{M}_1,\mathfrak{M}_2$ be two von Neumann algebras acting on $\calH$ such that $\mathfrak{M}_1 \subseteq \mathfrak{M}_2'$, and let $\mathfrak{M}_{12}=(\mathfrak{M}_1\cup\mathfrak{M}_2)''$ be the algebra generated by $\mathfrak{M}_1$ and $\mathfrak{M}_2$. A state $\omega$ on $\mathfrak{M}_{12}$ is called a normal product state iff $\omega$ is a normal state and there are states $\omega_1$ on $\mathfrak{M}_1$ and $\omega_2$ on $\mathfrak{M}_2$ such that:
\begin{equation}
\omega(AB)=\omega_1(A)\omega_2(B),\quad \forall A\in\mathfrak{M}_1,\forall B\in\mathfrak{M}_2.
\end{equation}A state belonging to the weak* closed convex hull of normal product states will be called a separable state across $\mathfrak{M}_{12}$. An entangled state across $\mathfrak{M}_{12}$ is defined as a nonseparable state. Following Summers \cite{summers_vacuum_1985,summers_bells_1996}, we define: 
\begin{align}
\begin{aligned}
\mathcal{T}_{12}=\bigg\lbrace & \frac{1}{2}[A_1(B_1+B_2)+A_2(B_1-B_2)]:\\ & A_i=A_i^*\in\mathfrak{M}_1, B_i=B_i^*\in\mathfrak{M}_2, -\id \leq A_i,B_i\leq \id\bigg\rbrace.
\end{aligned}
\end{align}Elements of $\mathcal{T}_{12}$ are called Bell operators for $\mathfrak{M}_{12}$. For a given state $\omega$ of $\mathfrak{M}_{12}$, define:
\begin{equation}
\label{def-beta}
\beta(\omega)=\sup \{ |\omega(A)|:A\in\mathcal{T}_{12}\}.
\end{equation}If $\beta(\omega)>1$, we say that $\omega$ violates a Bell inequality, or is Bell correlated.  In the AQFT context, Bell's theorem states that the correlations $\omega$ induces between $\mathfrak{M}_1$ and $\mathfrak{M}_2$ admit a local hidden variable model iff $\beta(\omega)=1$. This result has important conceptual implications for the interpretation of states in QFT. For instance, Summers showed that every pair of commuting non-abelian von Neumann algebras possesses \emph{some} normal state with maximal Bell correlation \cite{summers_independence_1990}. He also showed that in most standard QFT models, \emph{all} normal states are maximally Bell correlated across spacelike separated tangent wedges or double cones \cite{summers_independence_1990}. Such Bell correlations decrease exponentially with spacelike separation. Finally, Halvorson showed that for \emph{any} pair of mutually commuting von Neumann algebras, if both algebras are of infinite type then there is an open dense subset of vectors of $\calS$ which induce Bell correlated states for $\mathfrak{M}_{12}$ \cite{halvorson_locality_2001}. Since algebras of local observables in quantum field theory are always of infinite type, this result shows that for \emph{any} pair of spacelike separated systems, a dense set of field states violate Bell's inequalities relative to measurements that can be performed on the respective subsystems. Clifton and Halvorson contrast this result with the behaviour of finite-dimensional (elementary) quantum mechanics ``where decoherence will most often drive a pair of systems into a classically correlated state" \cite{clifton_entanglement_2001}.

Focusing on the character of correlations of particular states, Halvorson showed that if a vector state $x\in\calH$ is cyclic for $\mathfrak{M}_1$ acting on $\calH$, i.e. if by applying elements of $\mathfrak{M}_1$ to $x$ one can generate the entire state space, then the induced state $\omega_x$ is entangled across $\mathfrak{M}_{12}$ \cite{halvorson_locality_2001}. He also showed that such entangled states are locally highly mixed, in the sense that they have a norm dense set of components\footnote{A state $\omega$ on a von Neumann algebra $\mathfrak{M}$ is said to be a component of another state $\rho$ iff there is a third state $\tau$ such that $\rho=\lambda\omega+(1-\lambda)\tau$ with $\lambda\in [0;1]$.} in the state space of $\mathfrak{M}(\calO)$, paralleling our intuition about systems with a finite number of degrees of freedom. Consequently, the Reeh-Schlieder theorem---which entails that finite-energy states are induced by vectors which are cyclic for local algebras \cite{reeh_bemerkungen_1961}---implies that ``although there is an upper bound on the Bell correlation of the Minkowski vacuum (in models with a mass gap) that decreases exponentially with spacelike separation, the vacuum state remains nonseparable at all distances" \cite{halvorson_locality_2001}. Halvorson also showed that if $\mathfrak{M}_1$ and $\mathfrak{M}_2$ are commuting non-abelian von Neumann algebras each possessing a cyclic vector, and if $\mathfrak{M}_{12}$ possesses a separating vector\footnote{A vector $x\in\calH$ is called separating for a von Neumann algebra $\mathfrak{M}$ acting on $\calH$ iff $Ax=0$ implies $A=0$ for all $A\in\mathfrak{M}$.}, then the previous property is generic in the sense the the set of states on $\mathfrak{M}_{12}$ which are entangled across $\mathfrak{M}_{12}$ is norm dense in the state space of $\mathfrak{M}_{12}$. The Reeh-Schlieder theorem guarantees that the conditions of this result are satisfied whenever we consider spacelike separated regions satisfying $(\calO_1\cup\calO_2)'\neq\emptyset$, for instance when both regions are bounded in spacetime. Consequently, if a local observer Alice wants to be certain that by performing a local operation in $\calO_1$ she produces a disentangled state, she ``would need extraordinary ability to distinguish the state of [$\mathfrak{M}_{12}$] which results from her operation from the generic set of states of [$\mathfrak{M}_{12}$] that are entangled!"

\subsection{Type III algebras and split property}

We recall that a factor is type III iff its projections are infinite and equivalent. Such factors contain no abelian projections\footnote{We recall that a nonzero projection $P\in\mathfrak{M}\subseteq\mathfrak{B}(\calH)$ is called abelian iff the von Neumann algebra $P\mathfrak{M}P$ acting on the subspace $P\calH$ (with identity $P$) is abelian.} (which are finite), therefore the associated projection lattice has no atoms. Consequently, type III algebras possess no pure states, despite the fact that they always possess a dense set of vectors states, that are incidentally both cyclic and separating.

An important result shows that there are no product states across any two commuting type III algebras\footnote{The debate about ``improper" mixtures that arise by restricting an entangled state to a subsystem, and ``proper" mixtures that do not is therefore irrelevant in this context.}. Moreover, as stated in Section \ref{section-121}, in most known AQFT models the local algebras are type III$_1$. A characterization of type III$_1$ factors was established by Connes and St\o rmer \cite{connes_homogeneity_1978}: a factor $\mathfrak{M}$ acting on a (separable) Hilbert space is type III$_1$ iff for any two states $\rho,\omega$ of $\mathfrak{B}(\calH)$, and any $\epsilon>0$, there are unitary operators $U\in\mathfrak{M}$ and $U'\in\mathfrak{M}'$ such that $||\rho-\omega^{UU'}||<\epsilon$. This result immediately implies that there are no separable states across $(\mathfrak{M},\mathfrak{M}')$. Indeed, as argued by Clifton and Halvorson in \cite{clifton_entanglement_2001}, ``if some $\omega$ were not entangled, it would be impossible to act on this state with local unitary operations in [$\mathfrak{M}$] and [$\mathfrak{M}'$] and get arbitrarily close to the states that are entangled across [$(\mathfrak{M},\mathfrak{M}')$]". More interestingly, the Connes-St\o mer characterization implies that invariance under unitary operations on separate entangled systems and norm continuity force triviality of any measure of entanglement across $(\mathfrak{M},\mathfrak{M}')$. Note that this result does not contradict with the fact that entanglement entropy is norm continuous and invariant under unitary operations because there are no states represented by density operators for type III factors, therefore entropy is not available \cite{clifton_entanglement_2001}. Consequently, local observers cannot distinguish \emph{even in principle} between the different degrees of entanglement that states might have across $(\mathfrak{M},\mathfrak{M}')$, unlike the results of the previous section where this was only a matter of ``extraordinary ability".

Clifton and Halvorson claim in \cite{clifton_entanglement_2001} that the consequences of the above considerations are particularly strong when considering local algebras associated with diamond regions in $M$. Diamond regions are defined as the intersection of the timelike future of a given spacetime point $p$ with the timelike past of another point in $p$'s future. Indeed, consider a diamond region $\Diamond \in M$. Most AQFT models verify a property called Haag's duality:
\begin{equation}
\mathfrak{M}(\Diamond ')=\mathfrak{M}(\Diamond )'.
\end{equation}Every global state of the field is \emph{intrinsically} entangled across $(\mathfrak{M}(\Diamond '),\mathfrak{M}(\Diamond )')$, therefore ``it is \emph{never possible} to think of a field system in a diamond region $\Diamond$ as disentangled from its spacelike complement" \cite{clifton_entanglement_2001}. Thus, Einstein's methodological worry seems justified for field systems.

However, from a \emph{practical} point of view, the region in which a local observer is manipulating a field system is defined only approximately. Therefore, it is natural to consider the possibility of disentangling a state of the field across some pair of \emph{strictly} spacelike separated regions $(\mathfrak{M}(\calO_1),\mathfrak{M}(\calO_2))$, i.e. regions which remain spacelike separated when either is displaced by an arbitrarily small distance. They showed that because of the lack of abelian projections, there is a norm dense set of entangled states across $(\mathfrak{M}(\calO_1),\mathfrak{M}(\calO_2))$ that cannot be disentangled by any pure local operation performed in $\mathfrak{M}(\calO_1)$ \cite{clifton_entanglement_2001}. In summary, using Halvorson's words \cite{halvorson_locality_2001}:
\begin{quote}
``There are many regions of spacetime within which no local operations can be performed that will disentangle that region's state from that of its spacelike complement, and within which no pure or projective operation on any one of a norm dense set of states can yield disentanglement from the state of any other strictly spacelike separated region."
\end{quote}
Is any attempt at a \emph{practical} disentanglement of field systems doomed? The answer is no! Indeed, in most AQFT models the local algebras verify the so-called \emph{split property}: for any bounded open $\calO$ in Minkowski spacetime, and any larger region $\tilde\calO$ whose interior contains the closure of $\calO$, there is a type I factor $\mathfrak{I}$ such that:
\begin{equation}
\mathfrak{M}(\calO)\subset \mathfrak{I}\subset \mathfrak{M}(\tilde\calO).
\end{equation}
Clifton and Halvorson use this property to demonstrate the operational possibility of disentangling field systems. We below reproduce close to verbatim their proof in \cite{clifton_entanglement_2001}. Suppose that Alice wants to prepare some state $\rho$ on $\mathfrak{M}(\calO_1)$. The split property implies that there is a type I factor $\mathfrak{I}$ satisfying
\begin{equation}
\mathfrak{M}(\calO_1)\subset \mathfrak{I}\subset \mathfrak{M}(\tilde\calO_1)
\end{equation}for any region $\tilde\calO_1$ that contains the closure of $\calO_1$. The density operator associated with the vector representation on a Hilbert space $\calH$ of $\rho$ can be extended to a density operator $D_\rho$ in the type I algebra $\mathfrak{I}$, which can be written as a convex combination $\sum_i \lambda_iP_i$ of mutually orthogonal atomic projections in $\mathfrak{I}$ satisfying $\sum_i P_i=\id$ with $\sum_i\lambda_i=1$. Since each projection $P_i$ is equivalent, in the type III algebra $\mathfrak{M}(\tilde\calO_1)$, to the identity operator, there is a partial isometry $W_i\in\mathfrak{M}(\tilde\calO_1)$ satisfying
\begin{equation}
W_iW_i^*=P_i,\quad W_i^* W_i=\id.
\end{equation}Now consider the operator $T$ on $\mathfrak{M}(\tilde\calO_1)$ defined by
\begin{equation}
T(X)=\sum_i \lambda_i W_i^* X W_i,\quad \forall X\in \mathfrak{M}(\calO_1).
\end{equation}We claim that $T(X)=\rho(X)\id$ for all $X\in\mathfrak{M}(\calO_1)$. Indeed, because each $P_i$ is abelian in $\mathfrak{I}\supseteq \mathfrak{M}(\calO_1)$, the operator $P_iXP_i$ acting on $P_i\calH$ is a multiple of the identity operator $P_i$ on $P_i\calH$:
\begin{equation}
\label{local-prep}
P_iXP_i=c_i P_i.
\end{equation}Taking the trace on both sides of this equation, we get $c_i=\trace(P_i X)$. Moreoever, acting on the left of equation \eqref{local-prep} with $W^*_i$ and on the right with $W_i$, we obtain $W_i^*XW_i=\trace(P_iX)\id$, which yields the desired conclusion when multiplied by $\lambda_i$ and summer over $i$. Finally, for any initial state $\omega$ of $\mathfrak{A}(\calO_1)$, we have:
\begin{equation}
\omega(T(X))=\rho(X),\quad \forall X\in\mathfrak{A}(\calO_1),
\end{equation}which means that by performing an approximately local operation $T$ to $\mathfrak{M}(\calO_1)$ (choosing $\tilde\calO_1$ to approximate $\calO_1$ as close as we like), Alice can prepare any state on $\mathfrak{M}(\calO_1)$. Furthermore, if we define $T$ to be local to $\mathfrak{M}(\tilde\calO_1)$ as the requirement that $T$ leave the expectations of observables outside $\mathfrak{M}(\tilde\calO_1)$ and those in its center $\mathfrak{M}(\tilde\calO_1)\cap \mathfrak{M}(\tilde\calO_1)'$ unchanged \cite{halvorson_locality_2001}, then the result of Alice preparing $T$ will always produce a product state across $(\mathfrak{M}(\calO_1),\mathfrak{M}(\calO_2))$ when $\calO_2\subseteq (\tilde\calO_1)'$, i.e. for any initial state $\omega$ across $(\mathfrak{M}(\calO_1),\mathfrak{M}(\calO_2))$:
\begin{equation}
\omega(T(XY))=\omega(T(X)Y)=\rho(X)\omega(Y),\forall X\in\mathfrak{M}(\calO_1),\forall Y\in\mathfrak{M}(\calO_2).
\end{equation}
In Clifton and Halvorson's words \cite{clifton_entanglement_2001}: 
\begin{quote}
``[Whenever] we allow Alice to perform approximately local operations on her field system, she \emph{can} isolate it from entanglement with other strictly spacelike separated field systems, while simultaneously preparing the local state she wants. God is subtle, but not malicious."
\end{quote}Therefore, for all practical purposes, Einstein's methodological worry is neutralized by the split property.

\bigskip

In summary, we clarified the sense in which nonlocal correlations do not contradict the relativistic principle of locality. We reviewed the recent information-theoretic (partial) reconstructions of non-signalling quantum correlations, an approach that shed new light on what features are specifically quantum. We also showed that entanglement pervades QFT \emph{precisely} because of the relativistic constraints encoded in field models, and that local observers can isolate field systems despite the correlations the latter share with their environment. After this conceptual discussion, we now focus on understanding entanglement in the relativistic setting, most notably issues concerning entanglement detection and quantification.

\chapter{Entanglement and relativity}
\label{second-chapter}

In this chapter, we analyze the effects of relativistic constraints on entanglement detection and quantification. We start by adding a minimal relativistic requirement to nonrelativistic quantum theory, viz. Poincaré invariance. We review how this requirement couples the spin and momentum degrees of freedom for inertial observers through the so-called Wigner rotations. This in turn implies a transfer of entanglement between them for bipartite systems and imposes a fine-tuning of local operations if observers are to detect the nonlocal character of correlations. For non-inertial observers, the full machinery of QFT is necessary. In addition to the interpretational difficulties that arise because of the Unruh effect, local observers cannot efficiently manipulate entanglement because of the global character of the field modes. These difficulties are partially overcome by using the so-called Unruh-DeWitt detectors. Regarding entanglement quantification in QFT, many results in the literature show that entanglement entropy follows a divergent area law in the continuum limit for most field models. Since all systems are ultimately described by QFT, these results are in sharp contrast with the finite results of entanglement measures in nonrelativistic quantum theory. Following the intuition we gained from Chapter \ref{first-chapter} on the deep relationships between entanglement and the properties of local algebras, we present a novel regularization technique at low energy for the entanglement entropy of infinite-mode systems. The derived low-energy formalism is as expected equivalent to standard nonrelativistic quantum theory. This provides a controlled transition from the QFT picture of entanglement to the nonrelativistic quantum theory one. We end this chapter with a discussion of the thermodynamical relationship between the possibility of a high-energy regularization of entanglement entropy and spacetime dynamics.

\section{Poincaré invariance}

We define Minkowski spacetime, denoted by $\mathcal{M}$, as a four-dimensional real vector space equipped with a non-degenerate, symmetric bilinear form with signature $(+,-,-,-)$. The basic elements of this space are the 4-vectors $x$. These vectors have contravariant (denoted by upper indices) components $x^\mu$ and covariant components (lower indices) $x_\mu$. The relation between them is given by the metric tensor $g$ in the following way:
\begin{equation}
x_\mu=g_{\mu\nu}x^\nu=g^{\mu\nu}x_\nu,\hspace{0.3cm} \mu,\nu\in\{0,1,2,3\}.
\end{equation}The geometry of Minkowski spacetime is fully determined by specifying the metric tensor. It has the form:
\begin{equation}
g^{\mu\nu}=g_{\mu\nu}=\begin{bmatrix}
1 &0 &0 & 0\\
0 & -1 & 0 & 0\\
0 & 0 & -1 & 0\\
0 & 0 & 0 & -1
\end{bmatrix}.
\end{equation}The line element $ds$ is given by
\begin{equation}
ds^2=g_{\mu\nu} dx^\mu dx^\nu=dx^\mu dx_\mu=dx_\nu dx^\nu.
\end{equation}

\subsection{Lorentz group}

We consider the group $\mathcal{L}$ of transformations $\Lambda$, called the Lorentz group,  that leave the origin and the elements of the metric tensor invariant. It is a six-dimensional non-compact non-abelian real Lie group. It contains as a subgroup the group of orientation-preserving spatial rotations, which is isomorphic to the three dimensional Lie group SO(3) consisting of all $3 \times 3$ orthogonal real matrices with determinant one. The remaining three dimensions of the Lorentz group come from the boosts representing transformations from one inertial system to another that is moving in parallel to the first one with constant velocity.  In the  $4\times 4$ matrix representation of the Lorentz group $O(1,3)$, elements are denoted by $L^\mu_\nu$ and we write:
\begin{equation}
\label{matrix-rep}
x'^{\mu}=L^\mu_\nu x^\nu,
\end{equation}and require the identity:
\begin{equation}
g_{\mu\nu}=L^\alpha_\mu L^\beta_\nu g_{\alpha \beta}.
\end{equation}\\
One can show that in the $4 \times 4$ representation:
\begin{equation}
(L_0^0)^2=1+\sum_i(L^i_0)^2 \geq 1,
\end{equation}and check that if two Lorentz transformation $\Lambda_1$ and $\Lambda_2$ have the property $L^0_0(\Lambda_1)\geq 1$ and $L^0_0(\Lambda_2)\geq 1$ respectively, then their product has the same property. Therefore the Lorentz group have four components that are not simply connected:
\begin{itemize}
\item[(i)] $\mathcal{L}^\uparrow_+$: the set of transformations that preserve orientation ($\det(L)=1$ in the $4\times 4$ representation) and the direction of time ($L^0_0\geq 1$), called proper orthochronous.
\item[(ii)]$\mathcal{L}^\uparrow_-$: the set of transformations that reverse orientation ($\det(L)=-1$) and preserve the direction of time, called improper orthochronous.
\item[(iii)]$\mathcal{L}^\downarrow_+$: the set of transformations that preserve orientation and reverse the direction of time ($L^0_0\leq -1$), called proper non-orthochronous.
\item[(iv)]$\mathcal{L}^\downarrow_-$: the set of transformations that reverse orientation and the direction of time, called improper non-orthochronous.
\end{itemize}The quotient $\mathcal{L}/\mathcal{L^\uparrow_+}\simeq O(1,3)/SO^+(1,3)$, where $SO^+(1,3)$ denotes the identity component in the $4 \times 4$ representation, is isomorphic to the Klein group $\mathbb{Z}_2\times \mathbb{Z}_2$, whose element are denoted by $\id, P, T$ and $PT$. $P$ and $T$ correspond to parity inversion and time reversal respecively, and can be written in the $4\times 4$ representation as:
\begin{equation}
P^\mu_\nu=\begin{bmatrix}
1 &0 & 0 & 0\\
0 & -1 & 0 & 0\\
0 & 0 & -1 & 0\\
0 & 0 & 0 & -1
\end{bmatrix},\quad T^\mu_\nu=\begin{bmatrix}
-1 & 0 & 0 & 0\\
0 & 1 & 0 & 0\\
0 & 0 & 1 & 0\\
0 & 0 & 0 & 1
\end{bmatrix}.
\end{equation}

Any relativistic theory should be invariant under Lorentz transformations, therefore we should look for the various possible representations of the Lorentz group, besides the usual $4\times 4$ matrix representation. Thus we need to study the Lie algebra associated with the Lorentz group. The generators of the Lorentz group are denoted by $M^{\mu\nu}$ where $\mu,\nu\in \{0,1,2,3\}$ and satisfy the following commutation relations:
\begin{equation}
[M^{\mu\nu},M^{\rho\sigma}]=-i(M^{\mu\rho}g^{\nu\sigma}+M^{\nu\sigma}g^{\mu\rho}-M^{\nu\rho}g^{\mu\sigma}-M^{\mu\sigma}g^{\nu\rho}),
\end{equation}where $\mu,\nu,\rho,\sigma \in \{0,1,2,3\}$. These relations can be written in terms of the more familiar operators $\vec{J}=(J^1,J^2,J^3)$ and $\vec{K}=(K^1,K^2,K^3)$ (generators of rotations and boosts respectively) as follows:
\begin{equation}
\label{com-lorentz}
[J^j,J^k]=i\epsilon_{jkl} J^l,\hspace{0.2cm} [J^j,K^k]=i\epsilon_{jkl}K^l, \hspace{0.2cm}[K^j,K^k]=-i\epsilon_{jkl}J^l,
\end{equation}where $j,k,l \in\{1,2,3\}$ and $\epsilon_{jkl}$ is the fully antisymmetric tensor in three dimensions normalized according to $\epsilon^{123}=\epsilon_{123}=1$.

For finite or compact groups, every representation is equivalent to a unitary one. As the Lorentz group is not compact (the hyperbolical angles modeled by $\overrightarrow{\chi}$ range over the whole real axis), we can expect some representations to be non-unitary. Commutation relations \eqref{com-lorentz} can be decoupled by taking the following complex linear combinations:
\begin{equation}
\vec{C}=\frac{1}{2}(\vec{J}+i\vec{K}),\quad \vec{D}=\frac{1}{2}(\vec{J}-i\vec{K}).
\end{equation}These operators verify the following commutation relations:
\begin{align}
\begin{aligned}
{[}C^j,C^k]&=i\epsilon_{jkl}C^l,\\
[D^j,D^k]&=i\epsilon_{jkl}D^l,\\
[C^j,D^k]&=0.
\end{aligned}
\end{align}We recognize two independent sets of generators of $SU(2)$. We know the possible irreducible finite-dimensional representations\footnote{Unitary representations of the Lorentz group are not useful for our purposes.} of $\mathfrak{su}(2)$ are of the form:
\begin{equation}
D^{[j_1,j_2]}=D^{[j_1]}\otimes D^{[j_2]},\quad 2j_1,2j_2\in\mathbb{N},
\end{equation}where representation $D^{[j]}$ is $(2j+1)$-dimensional. Both $\vec{C}$ and $\vec{D}$ may be extended to the whole space by writing for a specific representation $[j]$ of $\mathfrak{su}(2)$:
\begin{equation}
\vec{C}=\vec{S}^{[j_1]}\otimes \id^{[j_2]},\quad \vec{D}=\id^{[j_1]}\otimes \vec{S}^{[j_2]},
\end{equation}where operators $\vec{S}^{[j]}$ are $(2j+1)\times (2j+1)$ matrices representing the generators of $SU(2)$ and $\id^{[j]}$ is the identity operator in the representation space $D^{[j]}$. As $\vec{J}=\vec{C}+\vec{D}$ and $\vec{K}=-i(\vec{C}
-\vec{D})$, we find for any element $A\in SL(2,\mathbb{C})$ characterized by $\overrightarrow{\theta}$ and $\overrightarrow{\chi}$ the following $SU(2)\otimes SU(2)$ representation:
\begin{align}
\begin{aligned}
D^{[j_1,j_2]}(A)  = & \exp(-i\overrightarrow{\chi}\cdot \vec{K})\exp(-i \overrightarrow{\theta}\cdot \vec{J})\\
= & \exp(-i\overrightarrow{\chi}\cdot[\vec{S}^{[j_1]}\otimes \id^{[j_2]}-\id^{[j_1]}\otimes \vec{S}^{[j_2]}])\\
& \cdot \exp(-i \overrightarrow{\theta}\cdot [\vec{S}^{[j_1]}\otimes \id^{[j_2]}+\id^{[j_1]}\otimes \vec{S}^{[j_2]}]).
\end{aligned}
\end{align}These matrices transform states with label $(j_1,j_2)$ into states with $(j_1,j_2)$. Boosts transform states of different $j$ into each other, therefore $j$ does not label irreducible subspaces, whereas $j_1,j_2$ do. Since $SL(2,\mathbb{C})$ is the universal covering of $SO^+(1,3)$, we provided projective representations of $\mathcal{L}^\uparrow_+$, which can be representations only if $j_1+j_2\in\mathbb{N}$ because
\begin{equation}
D^{[j_1,j_2]}(-e)=(-1)^{2(j_1+j_2)}\id^{[j_1]}\otimes \id^{[j_2]},
\end{equation}where $e$ is the neutral element of $\mathcal{L}$. An element $L \in \mathcal{L}^\uparrow_+$ can be uniquely written as a rotation followed by a boost:
\begin{equation}
\label{formula-lorentz}
D^{[j_1,j_2]}(L) = \exp(-i\overrightarrow{\chi}\cdot \vec{K})\exp(-i\overrightarrow{\theta}\cdot \vec{J}),
\end{equation}where $\overrightarrow{\chi},\overrightarrow{\theta}\in \mathbb{R}^3$. Elements in $\mathcal{L}^\uparrow_-,\mathcal{L}^\downarrow_+,\mathcal{L}^\downarrow_-$ can be obtained from those in $\mathcal{L}^\uparrow_+$ by multiplication by $P,T$ and $PT$ respectively.

As a side note, operators $\vec{C}^2$ and $\vec{D}^2$ commute with any of the generators $C^i$ and $D^j$, i.e. they are Casimir operators, and their eigenvalues $j_1(j_1+1)$ and $j_2(j_2+1)$ label the representations of the algebra. Alternative Casimir operators in terms of generators $M^{\mu\nu}$ are defined as:
\begin{align}
\begin{aligned}
C_1&=\frac{1}{2}M^{\mu\nu}M_{\mu\nu}=\vec{J}^2-\vec{K}^2=\frac{1}{2}(\vec{C}^2+\vec{D}^2),\\
C_2&=\frac{1}{4}\epsilon_{\mu\nu\alpha\beta}M^{\mu\nu}M^{\alpha\beta}=2\vec{J}\cdot\vec{K}=
\frac{1}{2i}(\vec{C}^2-\vec{D}^2).
\end{aligned}
\end{align}The fact that there are two quadratic Casimir operators can be explained as follows: the complexification $\mathfrak{sl}(2,\mathbb{C})\otimes \mathbb{C}$ of the Lie algebra $\mathfrak{sl}(2,\mathbb{C})$ (viewed as a real Lie algebra) associated to $SL(2,\mathbb{C})$ is isomorphic to the direct sum $\mathfrak{su}(2)\oplus \mathfrak{su}(2)$ of two rank-1 algebras.

\subsection{Poincaré group}

The Lorentz group is a subgroup of the Poincaré group $\mathcal{P}$  of isometries of Minkowski spacetime, whose $4\times 4$ representation is
\begin{equation}
\mathbb{R}(1,3) \rtimes O(1,3).
\end{equation}
It is a non-abelian Lie group with 10 generators: the generators of translations denoted by $P^\mu=(P^0, \vec{P})$ and the generators of $M^{\mu\nu}$ (or $\vec{J}$ and $\vec{K}$) of the Lorentz group.
We denote a transformation belonging to the Poincaré group as $(a,\Lambda)$, where $a$ is a four-vector specifying a translation and $\Lambda$ is a Lorentz transformation. Since we now have defined the generators of the Poincaré group, any Lorentz transformation $\Lambda$ can be written as:
\begin{equation}
(0,\Lambda)=\exp\left(-i\overrightarrow{\chi}\cdot \vec{K}\right) \exp\left(-i\overrightarrow{\theta}\cdot\vec{J}\right).
\end{equation}
  A translation along four-vector $a$ is given by
\begin{equation}
(a,e)=\exp(-i a_\mu P^\mu).
\end{equation}Hence, a general Poincaré transformation can be written as:
\begin{equation}
(a,\Lambda)=\exp\left(-ia_\mu P^\mu\right) \exp\left(-i\overrightarrow{\chi}\cdot \vec{K}\right) \exp\left(-i\overrightarrow{\theta}\cdot\vec{J}\right).
\end{equation}
The generators of the Poincaré group verify:
\begin{equation}
\label{poincare-symmetry}
[P^\mu,P^\nu] =0, \hspace{0.5cm}
[M^{\mu\nu},P^\sigma]=i(P^\mu g^{\nu\sigma}-P^\nu g^{\mu\sigma}).
\end{equation}In terms of the boosts and angular momentum operators, the second relation reads:
\begin{align}
\begin{aligned}
{[}K^j,P^0] &=-iP^j,\\
[K^j,P^k]&=iP^0g^{jk}=-iP^0\delta_{jk},\\
[J^l,P^0]&=0,\\
[J^l,P^k]&=i\epsilon^{lkm}P_m.
\end{aligned}
\end{align}
$\mathcal{P}$ splits into four non-simply connected components $\mathcal{P}^\uparrow_+,\mathcal{P}^\uparrow_-,\mathcal{P}^\downarrow_+,$ and $\mathcal{P}^\downarrow_-$.

Define the Pauli-Lubanski (PL) operator $W^\mu$ as:
\begin{equation}
\label{def-PL}
W^\mu=-\frac{1}{2}\epsilon^{\mu\nu\rho\sigma} P_\nu M_{\rho\sigma},
\end{equation}where $\epsilon^{\mu\nu\rho\sigma}$ denotes the fully antisymmetric tensor in four dimensions normalized by $\epsilon^{0123}=\epsilon_{0123}=1$. It can alternatively be written as:
\begin{align}
\begin{aligned}
W^0&=\vec{P}\cdot \vec{J},\\
\vec{W}&=P^0\vec{J}-\vec{P}\times\vec{K}.
\end{aligned}
\end{align}The PL operator verifies:
\begin{equation}
\label{pauli-lub-com}
 [W^\mu,P^\nu]=0,\quad [M^{\mu\nu},W^\sigma]=i(W^\mu g^{\nu\sigma}-W^\nu g^{\mu\sigma}).
 \end{equation}The second equality means that $W^\mu$ transforms under the Lorentz group like $P^\mu$, i.e. it is a four-vector, written as $(W^0,\vec{W})$. Because the tensor $\epsilon$ is fully antisymmetric, the vector $W^\mu$ is orthogonal to momentum: $W_\mu P^\mu=0$.

The first Casimir operator of the Poincaré group corresponds to the mass operator $M$ defined by:
\begin{equation}
M^2=P_\mu P^\mu=E^2-\vec{P}^2.
\end{equation}
 From relations \eqref{pauli-lub-com}, one can easily derive:
 \begin{equation}
 [P^\mu, W_\nu W^\nu]= [M^{\mu\nu},W_\sigma W^\sigma]=0.
 \end{equation}Therefore $W^2=W_\mu W^\mu$ is the second Casimir operator, and one can show that there are only two of them.

Poincaré invariance---as abstactly encoded in commutation relations \eqref{poincare-symmetry} between elements of the Lie algebra associated with the Poincaré group $\mathcal{P}$---is a fundamental requirement in building any relativistic (quantum) theory. In particular, if $U_g:\mathcal{H}\rightarrow\mathcal{H}$ represents the (non-necessarily linear) action of a Poincaré transformation $g$ on a Hilbert space $\mathcal{H}$, transition probabilities have to be preserved
\begin{equation}
\label{unitary}
|\langle U_g(\psi) | U_g(\phi) \rangle|^2 = |\langle \psi | \phi \rangle|^2, \quad \forall \psi, \phi \in {\cal H}\:, ||\psi||=||\phi||=1.
\end{equation}
A celebrated theorem due to Wigner establishes that a (bijective) map $U:\mathcal{H}\rightarrow\mathcal{H}$ verifying \eqref{unitary} must necessarily be either linear and unitary or antilinear and anti-unitary. Representations of the Poincaré group should further satisfy
\begin{equation}
U_{(a,\Lambda)\cdot (a',\Lambda')}=e^{i\theta\left((a,\Lambda),(a',\Lambda')\right)}U_{(a,\Lambda)}U_{(a',\Lambda')}
\end{equation} and $U_{(0,e)}=\id$ where $\cdot$ is the group product in $\mathcal{P}$ (taken as an abstract group) and $(0,e)$ the neutral element with respect to this product. If $(a,\Lambda)\in\mathcal{P}^\uparrow_+$, it can always be decomposed as $(a,\Lambda)=(a',\Lambda')\cdot (a',\Lambda')$ where $(a',\Lambda')\in\mathcal{P}^\uparrow_+$. Since $U_{(a,\Lambda)}=U_{(a',\Lambda')} U_{(a',\Lambda')}$ and the product of two unitary or antiunitary operators is unitary, $U_{(a,\Lambda)}$ must be unitary. Therefore we should look for projective unitary representations of $\mathcal{P}^\uparrow_+$ when considering its action on states, or equivalently for unitary representations of its universal covering whose $4\times 4$ representation is
\begin{equation}
\mathbb{R}(1,3)\rtimes SL(2,\mathbb{C}).
\end{equation}
Since the group of translation $T_4$ is abelian, all its irreducible representations are one-dimensional. Thus we can consider the set $|p,s\rang$ of vectors, each generating such a representation. These vectors correspond to eigenstates of the four-momentum operator:
\begin{equation}
P^\mu|p,s\rang = p^\mu |p,s\rang,
\end{equation}where the $p^\mu$'s are the eigenvalues. The index $s$ is an unspecified degeneracy label. There are three types of representations corresponding to $p^2=p_\mu p^\mu >0, p^2 =0$ and $p^2 <0$. These classes can be further divided as follows:
\begin{itemize}
\item[(i)] $p^2>0, p^0>0$, whose characteristic representative is $p=(m,\vec 0)$.
\item[(ii)] $p^2>0, p^0<0$,  whose characteristic representative is $p=(-m,\vec 0)$.
\item[(iii)] $p^2=0, p^0>0$,  whose characteristic representative is $p=(1,0,0,1)$.
\item[(iv)] $p^2=0, p^0=0$,  whose characteristic representative is $p=(0,0,0,0)$.
\item[(v)] $p^2=0, p^0<0$,  whose characteristic representative is $p=(-1,0,0,1)$.
\item[(vi)] $p^2<0$,  whose characteristic representative is $p=(0,0,0,1)$.
\end{itemize}We do not consider the possibility of tachyons, therefore we impose that $p^2\geq 0$ and $p^0\geq 0$. Thus we only need to consider three representations:

\

\textbf{Vacuum:} Representation (iv) is a single state. It has the property that all generators of the Poincaré group are represented by one-dimensional null matrices. Consequently, all group elements are represented by identity matrices. So, this state has momentum and angular momentum equal to zero and is invariant under any Poincaré transformation. We interpret it as the vacuum.

\

\textbf{Massive case:} Representation (i) fixes a mass $m$. However, there is an infinite number of eigenvalues $p^\mu$ of $P^\mu$ that satisfy $p^2=m^2$, i.e. the mass $m$ representation is infinite-dimensional. Looking at the characteristic representative $\mathring{p}=(m,\vec{0})$, equation \eqref{def-PL} implies
\begin{equation}
W^0=0,\quad W^i=-mJ^i,
\end{equation}therefore the $W^\mu$ generate $SU(2)$---which is then called the little group of massive particles---and the Casimir operator $W_\mu W^\mu$ is the Casimir operator of $SU(2)$. The irreducible representations at fixed $m$ are then given by the well-known $(2j+1)\times (2j+1)$ dimensional matrices $D^{[j]}$, with $j=0,\frac{1}{2},1,\frac{3}{2},\cdots$, corresponding to the spinor representation of $SU(2)$. Label $s$ can thus be decomposed into two labels $j,\lambda$ corresponding to the common eigenvectors of the Casimir operator $\vec{J}^2=-\frac{1}{m^2}W_\mu W^\mu$ with eigenvalue $j(j+1)$, and $J^3$ with eigenvalue $\lambda$. Since there is no orbital motion in the rest frame of the particle, $\vec J^2$ and $J^3$ correspond to the total intrinsic spin and spin along $\hat z$ axis of the particle respectively. We have the following relations:
\begin{align}
\begin{aligned}
P^\mu | \mathring{p},j,\lambda\rang &= \mathring{p}^\mu | \mathring{p},j,\lambda\rang,\\
P^2 | \mathring{p},j,\lambda\rang &= m^2 | \mathring{p},j,\lambda\rangle,\\
\vec{J}^2 | \mathring{p},j,\lambda\rang &= j(j+1) | \mathring{p},j,\lambda\rangle,\\
J^3 | \mathring{p},j,\lambda\rang &= \lambda |\mathring{p},j,\lambda\rang.
\end{aligned}
\end{align}
We now want to further decompose label $s$ in states $|p,s\rang$ which do not correspond to the rest frame state of the particle. Introduce three four-vectors $n[i], i=1,2,3$ that have the following properties:
\begin{itemize}
\item[(i)] \emph{Transversality:} $n[i]\cdot p =n[i]^\mu p_\mu=0$.
\item[(ii)] \emph{Orthogonality:} $n[i]\cdot n[j]=n[i]^\mu n[j]_\mu = g_{ij}=-\delta_{ij}$.
\item[(iii)] \emph{Closure:} $\sum_{i=1}^3 n[i]^\mu n[i]^\nu= -\left( g^{\mu\nu} - \frac{p^\mu p^\nu}{p^2}\right)$.
\item[(iv)] \emph{Handedness:} $\epsilon_{\mu\nu\alpha\beta} n[i]^\mu n[j]^\nu n[k]^\alpha p^\beta=m\epsilon_{ijk}$.
\end{itemize}Writing conditions (iv) in the rest frame $\mathring{p}=(m,\vec{0})$ we have
\begin{equation}
n[i]^\mu=\mathring{n[i]^\mu}=\delta_{\mu i},
\end{equation}which means that in the rest frame the $n$'s are just the unit vectors along the (right-handed) space coordinate axes. The vectors $n[1],n[2],n[3]$ and $p$ form a set of mutually orthogonal vectors called a tetrad or vierbein. The next step is to project $W$ onto these vectors:
\begin{equation}
W[i]=W\cdot n[i] = W_\mu n[i]^\mu.
\end{equation}One can check that:
\begin{equation}
[W[j],W[k]] = im\epsilon_{jkl} W[l],
\end{equation}and $W[1],W[2]$ and $W[3]$ satisfy the canonical commutation relations of a rotation algebra. One can show that
\begin{equation}
W^2=-(W[1]^2+W[2]^2+W[3]^2),
\end{equation} therefore the eigenvalues of $W^2$ are again the $-m^2 j(j+1)$ with $j=0,\frac{1}{2},1,\frac{3}{2},\cdots$, and label $s$ can be decomposed into labels $j,\lambda$ corresponding to the common eigenvectors of the Casimir operator
\begin{equation}
 \vec{S}^2=-\frac{1}{m^2} W_\mu W^\mu=-\frac{1}{m^2}(W[1]^2+W[2]^2+W[3]^2)
 \end{equation} with eigenvalue $j(j+1)$) and
\begin{equation}
 S^3=\frac{1}{m}W[3]
 \end{equation} with eigenvalue $\lambda$. For a particular choice of the tetrad we have $S^3=\vec{J}\cdot \frac{\vec{P}}{||\vec{P}||}$, and $\lambda$ can be interpreted as the helicity of the particle. We have the following relations:
\begin{align}
\begin{aligned}
P^\mu | p,j,\lambda\rang &= p^\mu |p,j,\lambda\rang,\\
P^2 | p,j,\lambda\rang &= m^2 | p,j,\lambda\rangle,\\
\vec{S}^2 | p,j,\lambda\rang &= j(j+1) | p,j,\lambda\rangle,\\
S^3 | p,j,\lambda\rang &= \lambda |p,j,\lambda\rang.
\end{aligned}
\end{align}Operators
\begin{equation}
\label{spin-def}
S^i=\frac{1}{m}W[i],\quad i=1,2,3
\end{equation} are called the intrinsic spin (along three distinct space directions) of the particle, and may now differ from the $J^i$'s because of the possible orbital motion of the particle. Thus, projecting the PL four-vector on suitable vectors one can define a notion of intrinsic spin in all reference frames.

If we fix a representation $[j]$, we can denote a state $|p,j,\lambda\rang$ simply by $|p,\lambda\rang$. Define $L(p)$ as the Lorentz transform verifying
\begin{equation}
U^{[j]}(L(p))|\mathring{p},\lambda\rang = |p,\lambda\rang,
\end{equation}and, for simplicity, denote $L(\Lambda)^\mu_\nu p^\nu$ by $\Lambda p$.
One can then show that the $[j]$ representation of a Poincaré transform is
\begin{equation}
\label{wigner-little-massive}
U^{[j]}(a,\Lambda) |p,\lambda\rang = \exp\left(-i a \cdot \Lambda p\right) \sum_{\lambda'=-j}^j D^{[j]}_{\lambda\lambda'}(W(\Lambda,p)) |\Lambda p,\lambda'\rang,
\end{equation}where $D^{[j]}_{\lambda\lambda'}(W(\Lambda,p))$ is the $(2j+1)\times(2j+1)$ matrix representation of $W(\Lambda ,p)=L^{-1}(\Lambda p)\Lambda L(p)$, the so-called Wigner's little rotation that leaves the rest momentum invariant and rotates the helicity. The rotation of helicity comes from the fact that $\Lambda L(p)$ and $L(\Lambda p)$ are not equal even though both of them bring the momentum $\mathring{p}$ to $\Lambda p$.

\

\textbf{Massless case:} For representation (iii) there is again an infinite number of eigenvalues $p^\mu$ of $P^\mu$ satisfying $p_\mu p^\mu=0$, therefore it is infinite-dimensional. Similarly to the massive case, consider the characteristic representative $\mathring{p}=(1,0,0,1)$. The 	four-vector becomes:
\begin{align}
\begin{aligned}
W^0&= W^3=-m\vec{J}\cdot \frac{\vec{P}}{||\vec{P}||},\\
W^1 &=-m(J^1+K^2),\\
W^2 &=-m(J^2-K^1).
\end{aligned}
\end{align}These have algebra
\begin{align}
\begin{aligned}
{[}W^1,W^2]=0,\quad [W^2,W^3]=-imW^1,\quad [W^3,W^1]=-imW^2,
\end{aligned}
\end{align}which is precisely the algebra of the euclidian group $E(2)$. Again, this was expected since the $W^\mu$ also form the generators of the little group in the massless case. The covering group of $E(2)$ is the semi-direct product of the group $T_2$ of translations in $\mathbb{R}^2$ and rotations $\exp(i\phi/2), 0\leq \phi <4\pi$. Since $T_2$ is abelian, its irreducible representations are one-dimensional and can be written as
\begin{equation}
\chi^{[k_1,k_2]}=\exp(i(k_1x_1+K_2x_2)),\quad k_1,k_2\in\mathbb{R}.
\end{equation}Two cases must be distinguished:
\begin{itemize}
\item[(i)] $k_1^2+k_2^2>0$: there are infinitely many possible values for $(k_1,k_2)$, therefore representations are infinite-dimensional representations and require a continuous spin quantum number. Such representations are not physically interesting.
\item[(ii)] $k_1=k_2=0$: this is the identity representation of $T_2$, and the little group consists of two-dimensional rotations only, whose unitary irreducible representations are one-dimensional:
\begin{equation}
D^{[j]}(\phi)=\exp(ij\phi/2),\quad j=0,\pm\frac{1}{2},\pm 1,\cdots.
\end{equation}
\end{itemize}
We restrict our attention to the physically important case (ii). To label states, we again define a vierbein\footnote{Note that we here impose only linear independence as a constraint on the vierbein. The four constraints of the massive case are here irrelevant because there is no rest frame vierbein with which the boosted vierbein should coincide.} for the frame\footnote{In a general Lorentz frame $p=(p,0,0,p)$, the corresponding four-vectors of the vierbein are found by applying a suitable Lorentz transformation.} $\mathring{p}=(1,0,0,1)$:
\begin{align}
\begin{aligned}
\mathring{p} &=(1,0,0,1),\\
n[1]&=(0,1,0,0),\\
n[2]&=(0,0,1,0),\\
s&=(s^0,0,0,s^3), \quad s^0>|s^3|.
\end{aligned}
\end{align}One can easily check that $W^\mu$ can be expanded on this vierbein as follows:
\begin{equation}
\label{tetra-decomp}
W^\mu=W^1 n[1]^\mu + W^2n[2]^\mu+W^0\mathring{p}^\mu + W' s^\mu.
\end{equation}The transversality condition $W_\mu \mathring{p}^\mu=0$ (obtained by applying $W_\mu P^\mu$ to $|\mathring{p},s\rang$) implies that $W'=0$. Therefore $W_\mu W^\mu = (W^1)^2+(W^2)^2$. Wigner argues in \cite{wigner_unitary_1939} that in order to have finite-dimensional representations, one should further impose $W_\mu W^\mu=0$. Therefore, we obtain $W^\mu=W^0 \mathring{p}^\mu$, and we can equate $j$ and $\lambda$. In the case where space inversions are included in the symmetry group, one must group $D^{[j]}$ and $D^{[-j]}$ together to form an irreducible representation. The states $|p,j\rang$ and $|p,-j\rang$ where  $j>0$ are then interpreted as states of a massless particle with spin $j$ and helicity $j$ and $-j$ respectively. Basis states, denoted by $|p,\lambda\rang$, satisfy:
\begin{align}
\begin{aligned}
P^\mu|p,\lambda\rang &= p^\mu |p,\lambda\rang,\\
W^0|p,\lambda\rang &= \lambda m|p,\lambda\rang,\\
W^i|p,\lambda\rang &= 0,\quad i=1,2.
\end{aligned}
\end{align}The generator of rotations is $W^0$, therefore the helicity of massless particles is invariant under rotations. Even though $[K^j,W^0]=-iW^j$, one can show that helicity is also invariant under boosts. If $p$ and $p'$ are two light-like four-momenta differing by a rotation of angle $\phi$ about the direction of $\vec{p}$ only, they verify:
\begin{equation}
|p',\lambda\rang = e^{i\lambda\phi}|p,\lambda\rang = D_{\lambda\lambda}^{[|\lambda|]} (L(p')^{-1}\Lambda L(p))|p,\lambda\rang.
\end{equation}Here there is no summation over the helicities because of the invariance of the helicity under boosts and rotations for massless particles.

\section{Observer-dependent entanglement}

Understanding entanglement in relativistic settings has been a key question in relativistic quantum information. Early results show that entanglement is observer-dependent for inertial observers \cite{czachor_einstein-podolsky-rosen-bohm_1997,alsing_lorentz_2002,terashima_einstein-podolsky-rosen_2002,rembielinski_einstein-podolsky-rosen_2002,peres_quantum_2002,gingrich_quantum_2002,ahn_relativistic_2003,li_relativistic_2003,lee_quantum_2004} (See \cite{peres_quantum_2004,czachor_two-spinors_2010,alsing_observer-dependent_2012} for reviews). For non-inertial observers, the Unruh effect is responsible for an observer-dependent particle number \cite{unruh_notes_1976}, and the second quantization framework becomes necessary. It was shown that entanglement between two field modes is degraded by the Unruh effect when observers are in uniform acceleration \cite{terashima_einstein-podolsky-rosen_2004,fuentes-schuller_alice_2005}. We address these questions in this section.

\subsection{Inertial observers}

The field of relativistic quantum information is focused on describing how relativistic particles behave in a regime where the nature and the number of particles do not change during an experiment, such that not all the machinery of quantum field theory is necessary. This simplified view of the problems may have applications in the near future if we use the spin of relativistic particles to encode quantum information. In this section, we consider a version of the EPR problem in which measurements are performed by moving inertial observers. The aim of such an approach is to explore effects of the relative motion between the sender and receiver in quantum information protocols.

In order to build one-particle states, we must first fix a value for mass and a representation $[j]$. States can then be labeled by $|\vec{p},\lambda\rang$ instead of $|m,\vec{p},j,\lambda\rang$. These are the so-called helicity states, and Wigner's little rotation now reads:
\begin{equation}
\label{LT-one}
U(\Lambda)|\vec{p},\lambda\rang =\sum_{\lambda'=-j}^j D^{(j)}_{\lambda \lambda'} (W(\Lambda,p)) |\vec{p}_\Lambda,\lambda'\rang,
\end{equation}where $\vec{p}_\Lambda$ are the spatial components of $\Lambda p$.
The familiar spin states can be obtained by suitable rotation
\begin{equation}
|\vec{p},\sigma\rang = \sum_\lambda D^{[j]}_{\lambda\sigma}(R^{-1}(\vec{p}/||\vec{p}||)|\vec{p},\lambda\rang,
\end{equation}where $R^{-1}(\vec{p}/||\vec{p}||)$ rotates the 3-momentum direction $\vec{p}/||\vec{p}||$ to the $\hat{z}$-axis.
Spin states $|\vec{p},\sigma\rang$ are normalized as
\begin{equation}
\lang \vec{p'},\sigma' |\vec{p},\sigma \rang = (2\pi)^3(2p^0)\delta_{\sigma' \sigma} \delta^3(\vec{p'}-\vec{p}),
\end{equation}
 and one-particle states are given by
\begin{equation}
\label{one-particle-WP}
|\Psi\rang = \sum_\sigma \int_{-\infty}^{\infty} d\mu(p)\psi_\sigma(p)|\vec{p},\sigma\rang,
\end{equation}where $\psi_\sigma(p)=\lang\vec{p},\sigma|\Psi\rang$ verifies $\sum_\sigma \int_{-\infty}^{\infty} d\mu(p)|\psi_\sigma(p)|^2=1$, and $d\mu(p)=d^3\vec{p}/(2\pi)^3(2p^0)$ is the usual Lorentz invariant measure.

Applying the Lorentz transform $\Lambda$ to state state $|\Psi\rang$ yields the following description of this state in the boosted frame: 
\begin{equation}
|\Psi'\rang = \sum_\sigma \int_{-\infty}^{\infty} d\mu(p)\psi'_\sigma(p) |\vec{p},\sigma\rang,
\end{equation}
 where
\begin{equation}
 \psi'_\sigma(p)=\sum_\eta D^{(1/2)}_{\sigma\eta}(W(\Lambda,\Lambda^{-1}p))\psi_\sigma(\Lambda^{-1}p) = \lang \vec{p},\sigma |U(\Lambda)| \Psi\rang.
\end{equation}
 This result can be derived from a straightforward calculation using \eqref{one-particle-WP} and \eqref{LT-one}, the normalization property of states, a change of variables $p\rightarrow \Lambda^{-1} p$ and the fact that $d\mu(\Lambda^{-1}p)=d\mu(p)$, as well as relabeling of indices.

The work by Peres, Terno and Rohrlich \cite{peres_quantum_2002} considered wave packets as in \eqref{one-particle-WP} with the density matrix $\rho = |\Psi\rang \lang \Psi|$. The components of this density matrix are $\rho_{\sigma \sigma'}(\vec{p},\vec{p'})= \psi^*_\sigma(p)\psi_{\sigma'}(p')$, and one can obtain the reduced spin density matrix by after tracing out over the momentum
\begin{equation}
\rho^{red}_{\sigma \sigma'} = \int_{-\infty}^{\infty} d\mu(p)\psi^*_\sigma (p) \psi_{\sigma'}(p').
\end{equation}
They showed that unlike the complete density operator $\rho$, this reduced quantity has no covariant transformation law because of the momentum-dependent Wigner rotations that apply when switching between descriptions associated with different reference frames. Consequently, spin entropy is not a relativistic scalar and therefore has no invariant meaning, except in the limiting case of sharp momentum \cite{alsing_observer-dependent_2012}. Alsin and Fuentes commented on this result as follows \cite{alsing_observer-dependent_2012}:
\begin{quote}
``The conclusion is that even though it may be possible to formally define a spin state in any Lorentz frame by tracing out momentum variables, there will be no relationship between the observable expectation values in different Lorentz frames."
\end{quote}
Following this idea that Lorentz transformations entangle the spin and momentum degrees of freedom within a \emph{single} particle, Gingrich and Adami \cite{gingrich_quantum_2002}, Alsing and Milburn \cite{alsing_lorentz_2002} and Terashima and Ueda \cite{terashima_einstein-podolsky-rosen_2002} showed that Lorentz transformations also affect entanglement between the spins of \emph{different} particles. We summarize below the Terashima and Ueda demonstration of this result. Consider a pair of spin-$\frac{1}{2}$ particles with total spin zero moving away from each other in the $\hat x$ direction, each with velocity $v=\tanh(\xi)$. This situation is described by the state
\begin{equation}
|\Psi\rang = \frac{1}{\sqrt{2}}\left[\vec p_+,\uparrow\rang |\vec p_-,\downarrow\rang - |\vec p_+,\downarrow\rang |\vec p_-,\uparrow\rang \right],
\end{equation}where $p_\pm^\mu=(m\cosh(\xi),\pm m \sinh(\xi),0,0)$. In this case, the Lorentz transformation associated with each particle's rest frame is
\begin{equation}
L(p)=
\begin{bmatrix}
\cosh(\xi) & \pm\sinh(\xi) & 0 & 0\\
\pm\sinh(\xi) & \cosh(\xi) & 0 & 0\\
0 & 0 & 1 & 0\\
0 & 0 & 0 & 1
\end{bmatrix}.
\end{equation}
Assume that the two observers who perform measurements on the particles are moving in the $\hat z$ direction at the same velocity $v'=\tan(\chi)$. The corresponding Lorentz transformation reads
\begin{equation}
U(\Lambda)=
\begin{bmatrix}
\cosh(\chi) & 0 & 0 & -\sinh(\chi) \\
0 & 1 & 0 & 0\\
0 & 0 & 1 & 0\\
-\sinh(\chi) & 0 & 0 & \cosh(\chi) \\
\end{bmatrix}.
\end{equation}
The Wigner rotation \eqref{LT-one} is then reduced to a rotation about the $\hat y$-axis of angle $\pm\delta$ defined by
\begin{equation}
\label{angle}
\tan(\delta) = \frac{\sinh(\xi)\sinh(\chi)}{\cosh(\xi)+\cosh(\chi)}.
\end{equation}For spin-$\frac{1}{2}$ particles, rotations are represented by the Pauli matrices. Using the Pauli matrix $\sigma_y$, the transformation law \eqref{LT-one} thus becomes

\begin{align}
\begin{aligned}
U(\Lambda)|\vec{p_+},\uparrow\rang &= \cos(\frac{\delta}{2}) |\vec{p_+}_\Lambda ,\uparrow\rang + \sin(\frac{\delta}{2}) |\vec{p_+}_\Lambda,\downarrow\rang,\\
U(\Lambda)|\vec{p_-},\downarrow\rang &= -\sin(\frac{\delta}{2}) |\vec{p_-}_\Lambda,\uparrow\rang + \cos(\frac{\delta}{2}) |\vec{p_-}_\Lambda,\downarrow\rang,
\end{aligned}
\end{align}hence:
\begin{align}
\label{relativistic-epr}
\begin{aligned}
U(\Lambda)|\Psi\rang =\frac{1}{\sqrt{2}}&\left[\cos(\delta)\left(| \vec{p}_{+\Lambda},\uparrow\rang |\vec{p}_{-\Lambda},\downarrow\rang - |\vec{p}_{+\Lambda},\downarrow\rang |\vec{p}_{-\Lambda},\uparrow\rang\right) \right.\\
&\left.+\sin(\delta)\left(|\vec{p}_{+\Lambda},\uparrow\rang |\vec{p}_{-\Lambda},\uparrow\rang + |\vec{p}_{+\Lambda},\downarrow\rang |\vec{p}_{-\Lambda},\downarrow\rang\right)\right],
\end{aligned}
\end{align}where $\delta$ is given by equation \eqref{angle}. The spins of the two particles are rotated about the $\hat y$-axis through angles $\delta$ and $-\delta$ respectively because they are moving oppositely. From \eqref{relativistic-epr} one finds that measurements of the spin $\hat z$-component do not show a perfect anti-correlation in the relativistic setting. That is, the perfect anti-correlation in the same direction is not Lorentz invariant. Let $Q$ and $R$ be operators on the first particle corresponding to the spin $\hat z$- and $\hat y$-components, respectively. Similarly, let $S$ and $T$ be operators on the second particle corresponding to the spin component in the directions $(0,-\frac{1}{\sqrt{2}},-\frac{1}{\sqrt{2}})$ and $(0, -\frac{1}{\sqrt{2}}, \frac{1}{\sqrt{2}})$, respectively. Then, one obtains:
\begin{equation}
\lang QS \rang + \lang RS \rang + \lang RT\rang -\lang QT\rang = 2\sqrt{2}\cos^2(\delta).
\end{equation}The violation of Bell's inequality decreases with increasing velocity of the observers and that of the particles.

However, since the Lorentz transformation is a local unitary operation, the perfect entanglement should be preserved. This is indeed the case: if the directions of measurements are rotated about the $\hat y$-axis through $\delta$ for the first particle and through $-\delta$ for the second in accordance with their spin rotations, the measurements of spin are perfectly anti-correlated and Bell’s inequality remains maximally violated.

These works generated an intense study of relativistic EPR correlations for both spin-$\frac{1}{2}$ particles and photons. The primary focus of works on spin-$\frac{1}{2}$ particles was on the relevant covariant observable(s) for spin such that the expectation values obtained from measurements are the same in all inertial frames, which includes as a sub-problem the meaning and validity of the reduced density matrix, especially in the case where one traces out the momentum from the complete quantum state. Czachor and Wilczewski \cite{czachor_relativistic_2003} for massless particles \cite{czachor_relativistic_2003} and later Caban and Rembieli{\'n}ski for massive particles \cite{caban_lorentz-covariant_2005} showed that it is indeed possible to define a Lorentz-covariant reduced spin density matrix. How can one interpret this reduced spin operator? Alsin and Fuentes summarize the findings concerning this issue as follows \cite{alsing_observer-dependent_2012} :
\begin{quote}
``Such an object contains information about the average polarization of the particle, as well as information about its average kinematical state. For sharp momentum, the reduced density matrix does not change under Lorentz transforms. However, in the case of an arbitrary momentum distribution the entropy of the reduced density matrix is in general not a Lorentz invariant. Therefore, while one can define a Lorentz-covariant finite-dimensional matrix describing the polarization of a massive particle, one cannot completely separate out kinematical degrees of freedom when considering entanglement in the relativistic setting." 
\end{quote}

It should be clear by now that most complications in defining an appropriate reduced spin density matrix arise from how one should define a spin operator for massive particles, a long appreciated problem in relativistic quantum theory \cite{fleming_covariant_1965}. Indeed, spin is only unambiguously defined in the rest frame of the particle, where it coincides with the total angular momentum. For observers in an arbitrary inertial frame, we defined in the last section an intrinsic spin operator by choosing a vierbein. However, this definition is not unique. To understand why, we consider a more general definition based our understanding of spin in the nonrelativistic case, namely that spin is defined as the difference between the total angular momentum $\vec{J}$ and the orbital angular momentum $\vec{L}$: $\vec{S}=\vec{J}-\vec{L}$. In the relativistic case, $\vec{J}$ is well defined as a generator of the Poincaré group, while $\vec{L}=\vec{R}\times \vec{P}$ is not. Indeed, momentum $\vec{P}$ is a well-defined generator of the Poincaré group but there is no well-defined notion of a position operator (hence no concept of localization) in relativistic quantum mechanics. Consider for instance a single nonrelativistic particle described by a state $|\psi\rangle$. Its localization is determined by the domain of the position representation (or wave function) $\psi(x)$ of the state $|\psi\rangle$: one says that the particle is localized in a domain $\Delta\subset \mathbb R^3$ iff the support of $\psi(x)$ lies in $\Delta$, which means that $E_\Delta\psi(x)=\psi(x)$ with $E_\Delta$  the multiplication operator by the characteristic function of $\Delta$. However, time evolution generated by the nonrelativistic Hamiltonian
\begin{equation}
H=\frac 1{2m}\mathbf P^2=-\frac1{2m}\nabla^2
\end{equation}spreads out the localization in the following sense:
\begin{equation}
\exp(itH)E_\Delta\exp(-itH)E_{\Delta'}\neq 0, \forall t\neq 0,
\end{equation}where $\Delta, \Delta'$ are disjoint domains verifying $E_\Delta E_{\Delta'}=0$.
The spread of the wave-packet also holds in relativistic theory where time evolution is generated by the Hamiltonian
\begin{equation}
H=(c^2 \mathbf  P^2+m^2c^4)^{1/2}.
\end{equation}A general theorem shows that due to the analyticity implied by the relativistic spectrum condition, localization as encoded in terms of position operators is incompatible with causality:

\begin{Lem}
\label{no-loc-thm}
Suppose there is a mapping $\Delta\mapsto E_\Delta$ from subsets of spacelike hyperplanes in Minkowski spacetime into projectors on $\mathcal H$ such that
\begin{itemize}
\item[(i)] $U(a)E_\Delta U(a)^{-1}=E_{\Delta+a}$.
\item[(ii)] $E_\Delta E_{\Delta'}=0$ if $\Delta, \Delta'$ spacelike separated.
\end{itemize}
Then $E_\Delta=0$ for all $\Delta$.
\end{Lem}

\begin{proof} The spectrum condition implies that  for every $\Psi\in\mathcal H$, the function  $a\mapsto U(a)\Psi$ has an analytic continuation into $\mathbb R^4+\mathrm i {\rm V}^+\subset \mathbb C^4$. Condition $(ii)$ means that $\langle E_\Delta \Psi, U(a)E_\Delta \Psi\rangle=\langle \Psi, E_\Delta E_{\Delta+a}U(a)\Psi\rangle=0$ on an open set in Minkowski spacetime. But, according to the ``edge of the wedge" theorem, an analytic function that is continuous on the real boundary of its analyticity domain  and vanishing on an open subset of this boundary vanishes identically.\end{proof}
Consequently, if we are to define operators $\vec{L}$ and $\vec{S}$ one has to drop one of the two conditions of Lemma \ref{no-loc-thm}. Different choices were favored by different researchers. Popular choices include the center-of-mass operator $\vec{R}_{cm}$ and the Newton-Wigner position operator $\vec{R}_{NW}$ \cite{newton_localized_1949} defined by
\begin{equation}
\vec{R}_{cm}=-\frac{1}{2}\left[ \frac{1}{P^0}\vec{K}+\vec{K}\frac{1}{P^0}\right],\hspace{0.3cm} \vec{R}_{NW} =\vec{R}_{cm} - \frac{\vec{P}\times \vec{K}}{mP^0(m+P^0)}.
\end{equation}This leads to the center-of-mass and Newton-Wigner spin operators\footnote{Note that a choice of vierbein is necessary in order to define the spin directions.}
\begin{equation}
\vec{S}_{cm}=\frac{\vec{W}}{P^0},\hspace{0.5cm} \vec{S}_{NW}=\frac{1}{m}\left(\vec{W}-\frac{W^0\vec{P}}{P^0+m}\right).
\end{equation}
The general conclusion from studies using these operators as a covariant representation of spin measurement is that maximum EPR correlations can be recovered in any inertial frame if both the state and the spin observable are Lorentz transformed \cite{alsing_observer-dependent_2012}.

The above considerations on spin observable measurements are highly theoretical, which led Saldanha and Vedral to recently question their physical realizability \cite{saldanha_physical_2012,saldanha_spin_2012}. Their argument can be summarized as follows \cite{alsing_observer-dependent_2012}: consider a covariant description of the interaction $H_{PL}$ of a measurement apparatus with spin observable $\vec{S}$ constructed from the PL four-vector $W^\mu=(W^0,\vec{W})$. Depending on the chosen definition of spin observable, one has either
\begin{equation}
W^\mu=(\vec{S}_{cm}\cdot \vec{P},P^0\vec{S}_{cm})
\end{equation}
or
\begin{equation}
W^\mu=\left(\vec{S}_{NW}\cdot\vec{P},m\vec{S}_{NW}+(P^0-m)(\vec{S}_{NW}\cdot \vec{P})\frac{\vec{P}}{||\vec{P}||^2}\right).
\end{equation}Expectation values of spin measurements are covariant only if the interaction is a Lorentz scalar, i.e. if the Hamiltonian can be written in the form $H_{PL}=W^\mu G_\mu=W^0G^0-\vec{W}\cdot \vec{G}$ for some four-vector $G^\mu=(G^0,\vec{G})$. The crucial point is that no such coupling is known to exist in nature. In particular, the standard measurements where spin couples to the electromagnetic field cannot be written in this form. Focusing on how physical spin is known to couple to a measurement appartus, they show that violation of the Bell inequality in this setting is also dependent on momentum and make different predictions in relation to the previous works on the subject \cite{alsing_observer-dependent_2012}.

We conclude this review by recalling that several results were obtained in the case of massless particles \cite{alsing_lorentz_2002,gingrich_entangled_2003,caban_helicity_2008}. Since massless particles have a different little group, the observables covariance issue was analyzed for the relevant objects, e.g. linear polarizations, and works on the effects of relative motion on entanglement yielded globally similar results \cite{alsing_observer-dependent_2012}.

\subsection{Non-inertial observers}

In flat spacetime, all inertial observers agree on the number of particles and the observer-dependent nature of entanglement can be simply traced back to a transfer of entanglement between the particles degrees of freedom. However, non-inertial observers do not agree on the number of particles, which is known as the Unruh effect \cite{unruh_notes_1976}. In this section, we analyze how the Unruh-DeWitt detectors allow for an unambiguous understanding of the implications of the Unruh effect on entanglement between modes.

\subsubsection{Unruh effect}

For the convenience of the reader unfamiliar with the Unruh effect, we below reproduce close to verbatim an introductory discussion of the subject that can be found in \cite{thiffeault_what_1993}. Consider as an example a two-dimensional Minkowski spacetime, and define the coordinates $\bar u$ and $\bar v$ by
\begin{align}
\begin{aligned}
\bar u &= t-x,\\
\bar v &= t+x.
\end{aligned}
\end{align}These correspond to the null rays going through the origin. The line element is written
\begin{equation}
\label{line-element}
ds^2=dt^2-dx^2=d\bar u d\bar v.
\end{equation}
We define the following coordinate transformation:
\begin{align}
\label{coord-transf-1}
\begin{aligned}
t &= a^{-1}\me^{a\xi}\sinh(a\eta),\\
x &= a^{-1}\me^{a\xi}\cosh(a\eta),
\end{aligned}
\end{align}where $a>0$ is a constant and $-\infty < \eta,\xi < \infty$. Inverting the transformation
\begin{align}
\begin{aligned}
\bar u&=-a^{-1}\me^{-au},\\
\bar v&= a^{-1}\me^{av},
\end{aligned}
\end{align}where $u=\eta - \xi, v=\eta+\xi$. The line element \eqref{line-element} then becomes
\begin{equation}
\label{coord-transf-2}
ds^2=\me^{2a\xi}du dv = \me^{2a\xi}(d\eta^2-d\xi^2).
\end{equation}
The coordinates $(\eta,\xi)$ cover only a quadrant of Minkowski spacetime. Lines of constant $\eta$ are straight while lines of constant $\xi$ are hyperbolae corresponding to the world lines of uniformly accelerated observers with acceleration $\alpha$ given by
\begin{equation}
\alpha=a\me^{-a\xi}.
\end{equation}
The system $(\eta,\xi)$ is known as the Rindler coordinate system, and the portion $x>|t|$ of Minkowski spacetime, labeled by $R$, is called the Rindler wedge. A second Rindler wedge $x<-|t|$ labeled $L$ may be obtained by changing the signs of the right-hand sides of the transformation equations \eqref{coord-transf-1} and \eqref{coord-transf-2}. The null rays act as \emph{event horizons} for Rindler observers: an observer in $R$ cannot see events in $L$ and vice versa. $L$ and $R$ thus represent two causally disjoint universes. We mark also the remaining future $(F)$ and past $(P)$ regions on Fig.\ref{fig-rindler}. Any event in $P$ or $F$ can be connected by null rays to both $L$ and $R$.

\vspace{0.5cm}

\begin{figure}[!htbp]
\begin{center}
\includegraphics[scale=0.27]{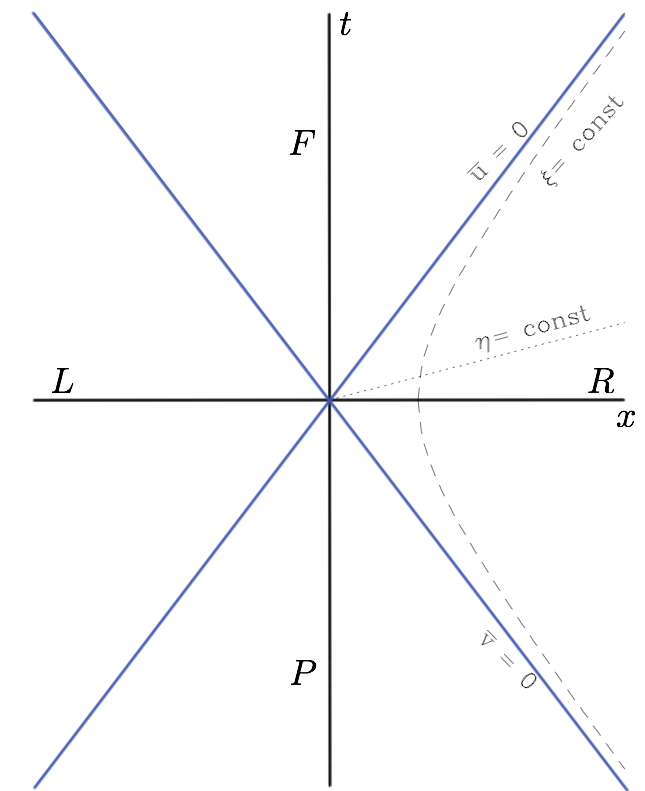}
\end{center}
\caption{Rindler coordinatization of Minkowski spacetime. In $R$ and $L$, time coordinates $\eta=\mbox{constant}$ are straight lines through the origin, space coordinates $\xi=\mbox{constant}$ are hyperbolae (corresponding to the world lines of uniformly accelerated observers) with null asymptotes $\bar u=0, \bar v=0$ acting as event horizons. The four regions $R, L, F$ and $P$ must be covered by separate coordinates patches. Rindler coordinates are non-analytic across $\bar u=0$ and $\bar v=0$ (From \cite{thiffeault_what_1993}).}
\label{fig-rindler}
\end{figure}

Now consider the quantization of a massless scalar field $\hat \Phi$ in two-dimensional Minkowski spacetime. The wave equation
\begin{equation}
\label{field-eq}
\Box\hat \Phi\equiv \left(\frac{\partial^2}{\partial t^2}-\frac{\partial^2}{\partial x^2}\right)\hat \Phi=\frac{\partial^2\hat \Phi}{\partial \bar u\partial \bar v}=0
\end{equation}
has the standard orthonormal mode solutions
\begin{equation}
\label{modes}
\bar{u}_k=(4\pi\omega)^{-1/2}\me^{i(kx-\omega t)},
\end{equation}where $\omega =|k|>0$ and $-\infty <k<\infty$. The modes with $k>0$ consist of right-moving waves
\begin{equation}
\label{pos-modes}
(4\pi\omega)^{-1/2}\me^{-i\omega\bar u}
\end{equation}along the rays $\bar u=\mbox{constant}$, while for $k<0$ one has left-moving waves along $\bar{v}=\mbox{constant}$
\begin{equation}
\label{neg-modes}
(4\pi\omega)^{-1/2}\me^{-i\omega\bar v}.
\end{equation}
Since the modes \eqref{modes} form a complete set, we can expand the field $\hat \Phi$ as
\begin{equation}
\hat \Phi=\sum_{k=-\infty}^\infty \left(\hat a_k\bar{u}_k+\hat a_k^\dagger \bar{u}_k^*\right).
\end{equation}
The operator $\hat a_k$ is the annihilation operator of mode $k$, while $\hat a_k^\dag$ is the corresponding creation operator. The Minkowski vacuum state $|\Omega_M\rang$ is then defined by
\begin{equation}
\label{M-vacuum}
\hat a_k|\Omega_M\rang =0.
\end{equation}
Now we wish to solve wave equation \eqref{field-eq} in the Rindler coordinates $(\eta,\xi)$:
\begin{equation}
\Box\hat\Phi=\me^{-2a\xi}\left(\frac{\partial^2}{\partial\eta^2}-\frac{\partial^2}{\partial\xi^2}\right)\hat\Phi = \me^{-2a\xi}\frac{\partial^2\hat\Phi}{\partial u\partial v}=0.
\end{equation}This has the same form as \eqref{field-eq}, so the mode solutions are
\begin{equation}
\label{rindler-modes}
u_k=(4\pi\omega)^{-1/2}\me^{i(k\xi\pm\omega\eta)},
\end{equation}with $\omega$ defined as in \eqref{modes}. The upper sign in \eqref{rindler-modes} applies in region $L$, the lower in region $R$.

Define
\begin{equation}
\label{R-modes}
u_{k,R}=
\left\lbrace
\begin{array}{ccc}
(4\pi\omega)^{-1/2}\me^{i(k\xi-\omega\eta)}  & \mbox{in} & R\\
0 & \mbox{in} & L\\
\end{array}\right.
\end{equation}
\begin{equation}
\label{L-modes}
u_{k,L}=
\left\lbrace
\begin{array}{ccc}
(4\pi\omega)^{-1/2}\me^{i(k\xi+\omega\eta)}  & \mbox{in} & L\\
0 & \mbox{in} & R\\
\end{array}\right.
\end{equation}
The set \eqref{R-modes} is complete in region $R$, while \eqref{L-modes} is complete in $L$, but neither set is separately complete on all of Minkowski spacetime. However, the modes \eqref{R-modes} and \eqref{L-modes} can be analytically continued into regions $F$ and $P$ ($a$ becomes imaginary in \eqref{coord-transf-1} and \eqref{coord-transf-2}). Thus these Rindler modes are as valid as the Minkowski spacetime basis \eqref{modes}.
We can thus expand the field as
\begin{equation}
\hat\Phi = \sum_{k=-\infty}^{\infty}\left(\hat b^{(1)}_k {u_{k,L}} + \hat b^{(1)\dag}_k {u_{k,L} ^*}
+\hat b^{(2)}_k { u_{k,R}} +\hat  b^{(2)\dag}_k {u_{k,R} ^*}\right),
\end{equation}yielding two alternative vacuum states, the Minkowski vacuum \eqref{M-vacuum} and the Rindler vacuum $|\Omega_R\rang$ defined by
\begin{equation}
\label{R-vacuum}
\hat b^{(1)}_k |\Omega_R\rang=\hat b_k^{(2)}|\Omega_R\rang =0.
\end{equation}These vacuum states are not equivalent as the Rindler modes are not analytic at the origin: because of the sign change in the exponent in \eqref{rindler-modes} at $\bar{u}=\bar{v}=0$, the functions $u_{k,R}$ do not go over smoothly to $ u_{k,L}$ as one passes from $R$ to $L$. In contrast, the Minkowski modes \eqref{pos-modes} and \eqref{neg-modes} are analytic and bounded in the entire lower half of the complex $\bar{u}$ (or $\bar{v}$) planes. This analyticity property remains true of any pure positive frequency function, i.e. any linear superposition of these positive frequency Minkowski modes. Hence the Rindler modes cannot be a linear superposition of pure positive frequency Minkowski modes, but must also contain negative frequencies. In other words, the $\hat b_k^{(1,2)}$ are a linear combination of both $\hat a_k$'s and $\hat a_k^\dag$'s, which means that to the accelerated observer $\hat b_k^{(1,2)}|\Omega_M\rang \neq 0$. The discussion of the actual relation between the $\hat b_k$'s and the $\hat a_k$'s is rather involved, so we shall only state the final result for the expectation value of the number opertor for the Rindler observer:
\begin{equation}
\lang \Omega_M|\hat b_k^{(1,2)\dag}\hat b_k^{(1,2)}|\Omega_M\rang =\left(\me^{2\pi \omega c\alpha}-1\right)^{-1}.
\end{equation}This is the Planck spectrum for radiation at temperature $k_B T=\frac{\hbar \alpha}{2\pi c}$.

\subsubsection{Unruh-DeWitt detectors}

In order to connect the Unruh effect to experimentally measurables facts, the Unruh-DeWitt detectors were introduced (See \cite{alsing_observer-dependent_2012} and references therein for more details). The idea is to calculate the response---or equivalently transitions between states---of this model detector along some trajectory and define accordingly the particle content of the observed state. Formally, a model detector is a quantum system whose states live in a product Hilbert space $\calH_D\otimes \calH_\Phi$ and which is provided with a Hamiltonian operator
\begin{equation}
H_m=H^D_m+H_m^\Phi+H^I_m,
\end{equation}where indices $D,\Phi, m$ and $I$ refer to ``detector", ``field", ``model" and ``interaction" respectively. Consider the simplest scenario of a free scalar field (we take $\hbar=c=1$):
\begin{equation}
\hat \Phi(x,t)=\frac{1}{\sqrt{2\pi}}\int \frac{dk}{\sqrt{2\omega_k}}\left(\hat a_k e^{i(kx-wt)}+\hat a^\dag_k e^{-i(kx-\omega t)}\right),
\end{equation}
with Hamiltonian
\begin{equation}
H^\Phi_m=\int dk \omega_k\hat a^\dag_k\hat  a_k,
\end{equation}where $ \omega_k=\sqrt{k^2+m^2}$ and $\hat a^\dag_k$ and $\hat a_k$ are the associated creation and annihilation operators. The field vacuum state will be denoted by $|\Omega\rang$. The detector part $H^D_m$ in the total Hamiltonian must account for a harmonic oscillator, or in the simplest case for a two-level system: unexcited $|0\rang_D$ (with $H^D_m|0\rang_D=0$) and excited $|E\rang_D$ (with $H^D_m|E\rang_D=E|E\rang_D$). Regardless of the interaction part, the model state $|0\rang_D\otimes |\Omega\rang$ is interpreted, by construction, as ``the detector is in its ground state and the field system is in its vacuum state" \cite{costa_modeling_2009}. The usual Hamiltonian used for Unruh-DeWitt detectors features an interaction part defined by:
\begin{equation}
H^I_m=\sigma \hat \Phi(y(t),t),
\end{equation}where $\sigma$ is a self-adjoint operator acting on $\calH_D$ and containing off-diagonal eflements and $y(t)$ is the detectors trajectory. This Hamiltonian verifies:
\begin{itemize}
\item[(i)] \emph{Particle interpretation:} The detector is a quantum system with discrete energy levels.
\item[(ii)] \emph{Particle absorption and emission:} Transitions between different levels must be possible.
\item[(iii)] \emph{Locality:} The detector interacts locally with the field.
\item[(iv)] \emph{Asymptotic vacuum:} No transition occurs for long enough periods of time when the detector is at rest.
\end{itemize}
One then switches on the interaction between the field and the detector for a period of time $\tau$ and, using Fermi's golden rule, one computes to first order the transition rate per unit of proper time of the detector to a state with one excitation:
\begin{equation}
T=\sum_{n,\psi} \frac{1}{2\tau}\int_{-\infty}^{\infty} dt |A_n|^2,\quad A_n=\lang E|_f\lang\psi|H_m^I|0\rang_D|\Omega\rang,
\end{equation}where $|\psi\rang_f$ is the final state of the field. Since the latter is not observed, one has to average over all possible outcomes to compute the detector response function $R(\omega)$. One finds that the structure of the detector is factored out and that the response function essentially depends on the Wightman function $\lang \Omega|\hat\Phi(x)\hat\Phi(x')|\Omega\rang$, i.e. the detector's response is independent of the particular model \cite{alsing_observer-dependent_2012}. Computations for an inertial detector result in $R(\omega)=0$, and for a uniformly accelerated detector  with acceleration $\alpha$ one finds that $R(\omega)$ is proportional to $\omega (e^{2\pi\omega\alpha}-1)$. This result can be interpreted as follows \cite{alsing_observer-dependent_2012}: one cannot distinguish the equilibrium reached between the accelerated detector and the field $\hat\Phi$ in the vacuum state $|\Omega\rang$ from the case when the detector is at rest and immersed in a bath at (Unurh) temperature $T=\frac{\alpha}{2\pi}$.

Let us come back to the topic of entanglement. We argued that because of the Unruh effect, the vacuum state is a thermal state for uniformly accelerated observers. Since Rindler regions $ L$ and $R$ are causally disconnected, an observer that is uniformly accelerated in one Rindler region has no access to information from the other Rindler region. Therefore, the state of such an observer is mixed because he must trace out over the unaccessible wedge. The degree of mixture depends on the observer's acceleration, namely the degradation is higher for observers with larger proper accelerations/temperature. This suggests the possibility of a thermodynamical relationship between entanglement and the motion of uniformly accelerated observers, a topic that we evoke in the Section \ref{einstein-eq}.

 For the sake of completeness, let us recall that other types of correlations relevant to quantum information theory have also been studies in non-inertial frames\footnote{For curved spacetimes, there are no Killing timelike vector fields, and therefore one cannot define a notion of subsystem. We recall that Killing fields correspond to generators of isometries on a Riemannian manifold. For instance, Minkowski spacetime admits ten Killing fields, in correspondance with the ten generators of the Poincaré group. Nonetheless, several results were obtained for special cases, for instance for black holes where spacetime is approximately flat at the horizon (See \cite{alsing_observer-dependent_2012} and references therein).}. For instance, it was shown that correlations are conserved in case one observer is non-inertial \cite{fuentes-schuller_alice_2005,alsing_entanglement_2006}, and degraded when both observers accelerate \cite{adesso_continuous-variable_2007}. Entanglement in non-inertial frames was also considered for other fields. For instance,  an ambiguity arises for fermionic fields in defining entanglement measures due to the anticommutation properties of field operators \cite{montero_fermionic_2011,bradler_two_2011,montero_comment_2011}. Recently, a density operator formalism on the fermionic Fock space was introduced which can be naturally and unambiguously equipped with a notion of subsystems in the absence of a global tensor product structure \cite{friis_fermionic-mode_2013}. It was shown that entanglement is degraded for such systems and that Bell's inequalities are not violated \cite{friis_residual_2011}, however entanglement remains finite in the infinite acceleration limit \cite{alsing_entanglement_2006,pan_degradation_2008}.

\subsection{Local detection of entanglement}

In the previous discussions we focused on how relative motion affects entanglement. What about entanglement detection and manipulation by local observers at rest in the same reference frame? This scenario is perfectly described by standard nonrelativistic quantum theory, but if one requires a more fundamental description of systems using QFT, problems arise because of the global character of the field modes since their manipulation by local observers would require infinite energy and time.

However, entanglement between modes is not the only kind of entanglement carried by the states of the field. Indeed, we argued in Chapter \ref{first-chapter} that the spatial degrees of freedom of a field system are also entangled, even for the vacuum state. Previous works suggest that this entanglement can be extracted and swapped by using two Unruh-DeWitt detectors at rest which become entangled by interacting with the Minkowski vacuum even if they are spacelike separated \cite{reznik_distillation_2000,reznik_entanglement_2002}. However, these suggestions rely on the ``vacuum dark counts" of Unruh-DeWitt detectors. Indeed, it is known that Unruh-DeWitt detectors undergo temporary transitions to excited states even for short times and the detectors being at rest. The possibility that a detector may ``click" at finite (short) times in the vacuum, but then ``erase" the record later to fulfill condition (iv) of the previous discussion asks for clarification. The natural approach is to extend condition (iv) to infinite times, a condition we call (iv) bis. However, this is problematic because it would contradict condition (iii). Indeed, the Reeh-Schlieder theorem implies that if a detector performs a von Neumann-type measurement corresponding to a projector $\Pi$, then either $\lang \Omega|\Pi|\Omega\rang\neq 0$ or $\Pi$ is nonlocal. A recent result by Costa \emph{et al.} clarified this issue \cite{costa_modeling_2009}. There they exhibit a toy theory with both a \emph{local} Hamiltonian and an \emph{effective} nonlocal two-level detector model that faithfully reproduces the detection rates of the fundamental theory \emph{and} satisfies condition (iv) bis. The configuration ``detector in its ground state + vacuum of the field" can then be considered as a stable bound state of the \emph{underlying} field theory for detectors at rest. Therefore the aforementioned attempts to detect entanglement between the spatial degrees of freedom of the field states are questionned.

 \section{Infinite-mode systems}

We argued that in studying entanglement, the relevant systems for local observers are localized systems, i.e. systems associated with spacelike separated regions of spacetime. This is further justified by the fact that all finite-energy measurements on infinite-mode systems are necessarily performed within some finite region of space and during a finite period of time.  However, the entropy of entanglement between such systems is ill-defined for all finite-energy states in relativistic field theories. In this section, we first review some results on the divergence of entanglement entropy. We then introduce a novel regularization technique at low energy, thus providing an effective low-energy description of entanglement for infinite-mode systems\footnote{This work has been published in \cite{ibnouhsein_renormalized_2014}.}. The derived low-energy formalism is as expected equivalent to standard nonrelativistic quantum theory, hence a controlled transition from the QFT picture of entanglement to the one of nonrelativistic quantum theory. We conclude with a review of the ``thermodynamical" relationship between a high-energy regularization of entanglement entropy and spacetime dynamics.

\subsection{Area law for entanglement entropy}

Consider a finite region of space $A$ at fixed time and its causal complement $\bar A$. Region $A$ has two complementary descriptions: general relativity identifies it with a submanifold of Minkowski spacetime, but as a quantum subsystem, $A$ is described by a Hilbert space $\mathcal{H}(A)$, which is a factor in the tensor product decomposition
\begin{equation}
\mathcal{H}=\mathcal{H}(A)\otimes \mathcal{H}(\bar{A})
\end{equation} of the total Hilbert space of the (exactly continuous) field theory under investigation. This tensor product structure (TPS) is defined by the \emph{choice} of a \emph{localization scheme}, i.e. a specific mapping associating $A$ and $\bar A$ with commuting subalgebras of the algebra of observables of the field system\footnote{Consider a quantum system divided into two parts $P$ (Part) and $R$ (Rest) $\calH=\calH_P\otimes\calH_R$. If we have two sets of observables $\mathfrak{A}_P$ and $\mathfrak{A}_R$, separately defined on subsystems $P$ and $R$ respectively, then we can trivially extend such observables to the entire system as follows: $\mathfrak{A}_P\rightarrow\mathfrak{A}_P\otimes \id_R$ and $\mathfrak{A}_R=\id_P\otimes\mathfrak{A}_R$. The basic idea in \cite{zanardi_virtual_2001,zanardi_quantum_2004} is that the converse is also true: if we isolate two commuting subalgebras $\mathfrak{A}_P$ and $\mathfrak{A}_R$ of $\mathfrak{B}(\calH)$ that generate the entire algebra $\mathfrak{B}(\calH)$, then they induce a unique tensor product decomposition $\mathfrak{B}(\calH)=\mathfrak{A}_P \otimes \mathfrak{A}_R$.}\cite{fleming_reeh-schlieder_1998,halvorson_reeh-schlieder_2001,piazza_quantum_2006,piazza_volumes_2007,piazza_glimmers_2010}. The standard localization scheme associates to space regions the field operators and their conjugates therein defined\footnote{In QFT, the usual local field observables commute at spacelike separated events, therefore local fields define a TPS.}. Suppose that the field is in a state $\rho$. The results of measurements to be performed by a local observer in region $A$ are described by the reduced density matrix obtained by tracing out the degrees of freedom outside $A$
\begin{equation}
 \rho_A=\trace_{\bar{A}} (\rho).
\end{equation}
 The von Neumann entropy associated with region $A$ is then defined as
\begin{equation}
S_A =-\trace (\rho_A \log \rho_A).
\end{equation}
For the vacuum state, this quantity typically scales as the area of the boundary between $A$ and $\bar A$:
\begin{equation}
\label{old-results}
\frac{S_A}{\mathcal{A}}=C(\frac{\lambda}{\mu},m\mu)\mu^{-2},
\end{equation}
where $\mathcal{A}$ is the area of the boundary between $A$ and $\bar{A}$, $m$ the mass of the field, $\lambda$ an infra-red cutoff and $C(x,y)$ some slowly varying function \cite{bombelli_quantum_1986, srednicki_entropy_1993, callan_geometric_1994, calabrese_entanglement_2004,eisert_colloquium:_2010,das_how_2006,das_power-law_2008}. This result can be derived by introducing an ultraviolet (UV) cutoff in the continuous model and by extending the simplified model we detail below to $N$ oscillators.  Most importantly, the area law for entropy extends to all finite-energy states provided power-law correction terms are added \cite{das_how_2006,das_power-law_2008}. Consequently, the entropy of entanglement of finite-energy states diverges in the continuum limit for $m>0$, a result that is deeply connected to the algebraic structure of QFT \cite{hegerfeldt_remark_1974,summers_vacuum_1985,summers_maximal_1987,redhead_more_1995,
summers_bells_1996,clifton_entanglement_2001, summers_yet_2008}. Indeed, we argued in Chapter \ref{first-chapter} that states on type III local algebras of observables are intrinsically entangled, and recent results by Keyl \emph{et al.} relate this type III entanglement to infinite one-copy entanglement \cite{keyl_infinitely_2002, keyl_entanglement_2006}. Consequently, all measures of entanglement are meaningless in most AQFT models. Furthermore, for field systems with type I local algebras of observables, states with infinite entropy of entanglement are trace-norm dense in state space \cite{eisert_quantification_2002}.

The simplified model we consider is that of a free massless scalar bosonic field in one space dimension consisting of only two coupled harmonic oscillators (we take $\hbar = c =1$) \cite{srednicki_entropy_1993}. The space of states of such a system is
\begin{equation}
\calH = \mathcal{H}_1 \otimes \calH_2,
\end{equation}where $\mathcal{H}_1, \calH_2$ are the single-oscillator Fock spaces. The creation and annihilation operators $\{\hat a_j^\dagger,\hat{a}_j\}_{j=1,2}$ acting on the Fock spaces can then be related to the usual position and momentum operators through:
\begin{equation}
\hat x_j=\sqrt{\frac{1}{2m\omega_j}}(\hat a_j^\dag +\hat{a}_j), \hspace{0.5cm} \hat p_j= i\sqrt{\frac{m\omega_j}{2}}(\hat a^\dag_j -\hat a_j), \hspace{0.5cm} j=1,2.
\end{equation}The Hamiltonian is defined as follows:
\begin{equation}
H=\frac{1}{2}\left[\hat p_1^2+\hat p_2^2 + k_0(\hat x_1^2+\hat x_2^2)+ k_1(\hat x_1 - \hat x_2)^2\right].
\end{equation}
The normalized wave function corresponding to the ground state of this system is:
\begin{equation}
\psi_0(x_1,x_2)=\pi^{-\frac{1}{2}}(\omega_+\omega_-)^\frac{1}{4}\exp\left[-(\omega_+ x_+^2+\omega_- x_-^2)/2\right],
\end{equation}where $x_\pm =(x_1 \pm x_2)/\sqrt{2}, \omega_+=\sqrt{k_0}$ and $\omega_-=\sqrt{k_0+2k_1}$. The crucial operation encoding the restriction of the field state to the state observed by a local observer is to form the ground state density matrix, and trace over the first oscillator, resulting in a density matrix for the second oscillator alone
\begin{align}
\begin{aligned}
\rho_2(x_2,x'_2)&=\int_{-\infty}^{+\infty} dx_1 \psi_0(x_1,x_2)\psi_0^\star (x_1,x'_2)\\
&=\pi^\frac{1}{2} (\gamma-\beta)^\frac{1}{2}\exp\left[-\gamma(x_2^2+{x'}_2^{2})/2+\beta x_1 x_2\right],
\end{aligned}
\end{align}where $\beta=\frac{1}{4}(\omega_+  -\omega_-)^2/(\omega_+ + \omega_-)$ and $\gamma-\beta = 2\omega_+\omega_-/(\omega_+ + \omega_-)$. The von Neumann entropy of the second oscillator can be expressed as $S=-\sum p_n\log(p_n)$ where the $p_n$'s are the eigenvalues of $\rho_2(x,x')$:
\begin{equation}
\label{eigenvalue-2-osc}
\int_{-\infty}^{+\infty} dx'\rho_2(x,x')f_n(x') = p_n f_n(x).
\end{equation}
The solution to problem \eqref{eigenvalue-2-osc} is:
\begin{align}
\label{solution-eigen}
\begin{aligned}
p_n&=(1-\xi)\xi^n\\
f_n(x) &= H_n(\nu^\frac{1}{2}x)\exp(-\nu x^2/2),
\end{aligned}
\end{align}where $H_n$ is a Hermite polynomial, $\nu =(\gamma^2-\beta^2)^\frac{1}{2}=(\omega_+\omega_-)^\frac{1}{2}, \xi=\beta/(\gamma+\nu)$ and $n$ runs from zero to infinity. Therefore, the entropy is:
\begin{equation}
S(\xi)=-\log(1-\xi)-\frac{\xi}{1-\xi}\log(\xi),
\end{equation}where $\xi$ is ultimately a function of the ratio $k_1/k_0$ that vanishes when $k_1=0$, highlighting the fact that vacuum entanglement arises due to the field dynamics.

\subsection{Renormalization at low energy}

Results above are in sharp contrast with the finite results one obtains for entanglement measures in low-energy experiments, e.g. quantum information protocols using qubits or trapped ions. Hence a conceptual tension between the QFT description of entanglement for low-energy experiments---which should ultimately be described in terms of entanglement between modes---and a description using nonrelativistic quantum theory. By introducing collective field operators, Zych \emph{et al.} effectively described the Klein-Gordon field system (which has an infinite number of modes) by a two-mode system \cite{zych_entanglement_2010}. The smearing operation preserves the Gaussian character of states, therefore the covariant matrix formalism can be used for an effective low-energy description of entanglement of Gaussian states such as the vacuum state. The results, which are strongly dependent on the interaction Hamiltonian of the detector coupled to the field system and on the choice of a detection profile, show that vacuum entanglement vanishes if the distance between modes is much larger than a few Compton wavelengths of the particle in the theory. This was expected since vacuum entanglement has no detectable effects at low-energy.

We now introduce a novel approach that provides an effective low-energy description of state entanglement for noncritical bosonic infinite-mode systems, with no assumption on the Gaussian character of the field state and independent of any detector model. The idea stems from the observation that entanglement is defined between subsystems of the field system, therefore an alternative tensor decomposition of the total Hilbert space based on a different \emph{choice} of localization scheme can yield finite results for entanglement measures. An obvious requirement is that the new choice of localization scheme be operationally well grounded at low energy. Our results show that this is indeed the case, and the derived effective description shows that all entanglement measures---including entanglement entropy---are finite. Furthermore, the effective description of field states is formally equivalent to the one using nonrelativistic quantum theory. This provides a controlled transition from the QFT picture of entanglement between localized systems, i.e. between spacelike separated regions of spacetime and the nonrelativistic quantum theory one between localized subsystems in the standard sense, e.g. qubits or trapped ions.

\subsubsection{Coarse-graining procedure} For simplicity, we consider a neutral Klein-Gordon field of mass $m$ in one space dimension at fixed time (we put $\hbar = c = 1$). The algebra of local observables for the Klein-Gordon field is generated by the canonical field operators:
\begin{equation}
\label{fields}
\begin{aligned}
\hat \Phi(x)&=\int\frac{dk}{\sqrt{2\pi}} \frac{1}{\sqrt{2\omega_k}}\left(e^{i kx} \hat{a}_k+e^{-i kx}\hat{a}^{\dagger}_k\right),\\
\hat\Pi(x)&=-i\int \frac{dk}{\sqrt{2\pi}}\sqrt{\frac{\omega_k}{2}}\left(e^{i kx}\hat{a}_k-e^{-i kx}\hat{a}^{\dagger}_k\right),
\end{aligned}
\end{equation}
where $\omega_k=\sqrt{k^2+m^2}$ and $\hat{a}^\dagger_k$ creates a field excitation of momentum $k$.
The vacuum is defined by:
\begin{equation}
\label{vacuum}
\hat{a}_k|\Omega\rangle = 0, \quad \forall k.
\end{equation}

Assume that the resolution for distinguishing different points in space is bounded by some minimal length $\epsilon$. The algebra of observables that are accessible under such conditions is generated by the coarse-grained field operators (See Fig.\,\ref{CG}):
\begin{equation}
\label{smeared}
\begin{aligned}
\hat\Phi_{\epsilon}(x) &=  \int dy G_{\epsilon}(x-y) \hat\Phi(y), \\
\hat\Pi_{\epsilon}(x)  &= \int dy G_{\epsilon}(x-y) \hat\Pi(y).
\end{aligned}
\end{equation}
Function $G_{\epsilon}(x)$ describes the detection profile:
\begin{align}
		\label{eq:gauss}
		G_{\epsilon}(x) = \frac{1}{(2\pi \epsilon^2)^{1/4}} e^{- \frac{x^2}{4\epsilon^2}}.
\end{align}This choice of profile is natural if we interpret coarse graining as arising from a random error in the identification of a point in space. More generally, for any profile with a typical length $\epsilon'$, consider intervals of length $\epsilon$ on which the profile is approximately constant. One can then convolute such a profile with a Gaussian of variance $\epsilon$ and consider the limit $\epsilon m\rightarrow\infty$ instead of $\epsilon' m\rightarrow\infty$.

Define the operators:
\begin{equation}
\label{effective}
\hat q_{j,\epsilon} =\hat\Phi_{\epsilon}(j d), \quad\hat p_{j,\epsilon} = \hat\Pi_{\epsilon}(j d),
\end{equation}where $d$ is the distance between neighbouring profiles. If $\epsilon\ll d$, they verify canonical commutation relations:
\begin{equation}
\label{CCR}
 \left[\hat{q}_{j,\epsilon},\,\hat{p}_{k,\epsilon}\right]\sim i\delta_{j\,k}.
 \end{equation}Imposing \eqref{CCR} is equivalent to saying that the operators $\{\hat{q}_{j,\epsilon},\hat{p}_{j,\epsilon}\}_j$ generate commuting subalgebras. As previously mentioned, two commuting subalgebras of observables $\mathfrak{A}$ and $\mathfrak{B}$ that generate the whole algebra of observables induce a TPS on the Hilbert space of states $\mathcal{H}=\mathcal{H}(A)\otimes \mathcal{H}(B)$ such that:
\begin{equation}
 \mathfrak{A} \rightarrow\mathfrak{A} \otimes \id_B,\quad \mathfrak{B} \rightarrow \id_A \otimes \mathfrak{B}.
\end{equation}
  The operators \eqref{effective} generate only a strict subalgebra of the entire algebra of field observables, because under coarse graining some possible observables are inaccessible. The whole algebra can be recovered by completing the set of functions $\{G_\epsilon(jd-y)\}_j$ up to an orthonormal basis in $\mathcal{L}^2(\mathbb{R})$ which, convoluted with the field operators \eqref{smeared}, defines a linear canonical transformation of modes. Thus, the algebra generated by the coarse-grained observables defines a decomposition of the total Hilbert space:
\begin{equation}
 \mathcal{H}={\cal H}_{cg}\otimes {\cal H}_{f},
 \end{equation}where ${\cal H}_{cg}$ are the coarse-grained, hence accessible, and ${\cal H}_{f}$ the fine-grained inaccessible degrees of freedom. The restriction to coarse-grained observables is therefore equivalent to tracing out subsystem ${\cal H}_{f}$, and operators $\{\hat{q}_{j, \epsilon},\hat{p}_{j, \epsilon}\}_j$ define distinct subsystems on $\mathcal{H}_{cg}$, each of which is isomorphic to a one-dimensional harmonic oscillator.

 \begin{figure}[!htbp]
\begin{center}
\includegraphics[scale=0.3]{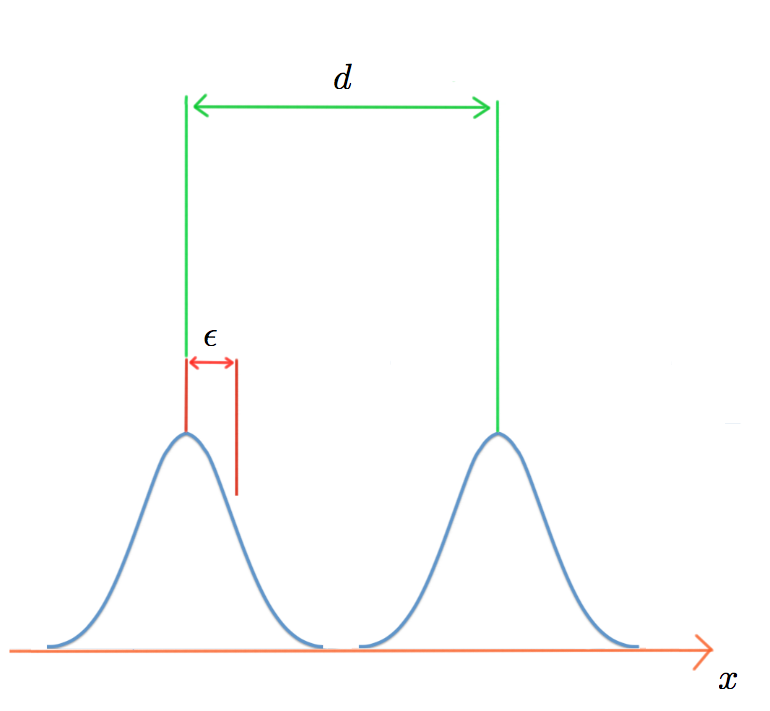}
\end{center}
\caption{The position in space at which a measurement is made can be determined only with limited accuracy, parametrized by $\epsilon$. This source of error is implemented by restricting the observable degrees of freedom to those accessible via measurement of coarse-grained operators. Neighboring profiles define different subsystems only if their separation $d$ verifies $d\gg\epsilon$. Under this condition, entanglement between neigboring profiles is a well-defined notion. We show that for finite-energy states, this entanglement reduces to the one calculated in nonrelativistic quantum theory.}
\label{CG}
\end{figure}

 Thus, we can define on $\mathcal{H}_{cg}$ the coarse-grained ladder operators:
\begin{equation}
\label{collectivecreation}
\begin{aligned}
\hat a_{j,\epsilon}= & \frac{1}{\sqrt{2}}\left(\sqrt{m'} \hat{q}_{j,\epsilon} + \frac{i}{\sqrt{m'}}\hat{p}_{j,\epsilon}\right),\\
 \hat a_{j,\epsilon}^\dag= & \frac{1}{\sqrt{2}}\left(\sqrt{m'} \hat{q}_{j,\epsilon} - \frac{i}{\sqrt{m'}}\hat{p}_{j,\epsilon}\right),
\end{aligned}
\end{equation}which verify $[\hat{a}_{j,\epsilon},\hat{a}_{k,\epsilon}^\dagger] \sim \delta_{jk}$. Parameter $m'$ has the dimension of mass. For a massive Klein-Gordon field, it is natural to take $m'=m$. Indeed, one can alternatively generate the local observables algebra with the ladder operators:
\begin{equation}
\begin{aligned}
\label{localcreation}
\hat{a}(x)&=\frac{1}{\sqrt{2}}\left(\sqrt{m}\hat\Phi(x)+\frac{i}{\sqrt{m}}\hat\Pi(x) \right),\\
 \hat{a}^\dagger(x)&=\frac{1}{\sqrt{2}}\left(\sqrt{m}\hat\Phi(x)-\frac{i}{\sqrt{m}}\hat\Pi(x)\right).
\end{aligned}
\end{equation}Their coarse-grained versions correspond to operators in \eqref{collectivecreation} with $m'=m$.

\subsubsection{The Newton-Wigner localization scheme} We recall that the Newton-Wigner (NW) annihilation and creation operators are respectively defined as the Fourier transforms of the momentum annihilation and creation operators \cite{newton_localized_1949}:
\begin{equation}
\label{NWcreation}
\begin{aligned}
\hat a_{NW}(x) &=  \int \frac{dk}{\sqrt{2 \pi}} e^{i kx}\hat a_k ,\\
 \hat a_{NW}^{\dag}(x) &=  \int \frac{dk}{\sqrt{2 \pi}} e^{-i kx}\hat a^{\dag}_k.
\end{aligned}
\end{equation}The NW operators define a localization scheme and are expressed in terms of the local fields \eqref{fields} as follows:
\begin{equation}
\label{NWfromfield}
\begin{aligned}
\hat a_{NW}(x) =&  \frac{1}{\sqrt{2}}\int dy\left[R(x - y)\hat \Phi(y) +i R^{-1}(x - y)\hat \Pi(y)\right] ,\\
\hat a_{NW}^{\dag}(x) =&  \frac{1}{\sqrt{2}}\int dy\left[R(x - y)\hat \Phi(y) -i R^{-1}(x - y)\hat \Pi(y)\right],
\end{aligned}
\end{equation}
where we have introduced the functions:
\begin{equation}
\label{Rfunction}
R(x) = \int \frac{dk}{2\pi}\sqrt{\omega_k}e^{i kx} , \quad R^{-1}(x)= \int \frac{dk}{2\pi}\frac{1}{\sqrt{\omega_k}}e^{i kx}.
\end{equation}Operators $\{\hat a _{NW} (x)\}$ annihilate the global vacuum, therefore the global vacuum is a product state of local vacua. One can then build a Fock space of the field theory for which states become localized and the entropy of entanglement of finite-energy states, such as thermal states, becomes finite \cite{cacciatori_renormalized_2009}. However, identifying local degrees of freedom with NW operators at a fundamental level is problematic: the Hamiltonian of the field, expressed in terms of the NW operators, is nonlocal\footnote{Note that we are here referring to locality in the relativistic sense.}. We do not address here the question of which localization scheme should be chosen at a fundamental level \cite{fleming_reeh-schlieder_1998,halvorson_reeh-schlieder_2001}. Instead we show that, under coarse graining, the entanglement properties of the NW fields for finite-energy states effectively hold, irrespective of the fundamental choice of local observables.

\subsubsection{Convergence between localization schemes} We now compare the algebra of coarse-grained observables generated by \eqref{collectivecreation} with the algebra generated by the following coarse-grained NW operators:
\begin{equation}
\label{smearedNW}
\begin{aligned}
\hat a_{NW,\epsilon}(x)  =& \int dy G_{\epsilon}(x-y)\hat a_{NW}(y),\\
\hat a_{NW,\epsilon}^\dag(x)  =& \int dy G_{\epsilon}(x-y) \hat a_{NW}^{\dag}(y) .
\end{aligned}
\end{equation}
Computations show that:
\begin{equation}
\label{smearedNWfromlocal}
\begin{split}
 \hat a_{NW,\epsilon}(x)  = \int dy \left[f^+_{\epsilon}(x-y)\hat a(y) + f^-_{\epsilon}(x-y) \hat a^{\dag}(y) \right],
\end{split}
\end{equation}
where:
\begin{equation}
\label{smearedR}
\begin{split}
	 f^{\pm}_{\epsilon}(x)=& \frac{1}{2} \left[\frac{R_{\epsilon}(x)}{\sqrt{m}} \pm \sqrt{m}R^{-1}_{\epsilon}(x) \right], \\
	 		   R_{\epsilon}(x)=& \int dy G_{\epsilon}(x-y)R(y) ,\\
    R^{-1}_{\epsilon}(x)=& \int dy G_{\epsilon}(x-y)R^{-1}(y).
\end{split}
\end{equation}In the limit of poor space resolution, the coarse-grained NW operators become indistinguishable from the coarse-grained local ladder operators since:
\begin{equation}
\label{fintegral}
\begin{aligned}
f^{\pm}_{\epsilon}(x) = & \frac{\sqrt{m}}{2} \int  \frac{dk}{\sqrt{2\pi}} e^{i mkx} G_{\frac{1}{2 m \epsilon}}(k)\\
&\cdot \left[	\left(1+k^2\right)^{1/4}\pm  	\left(1+k^2\right)^{-1/4} \right],
\end{aligned}
\end{equation}and in the limit where the minimal resolvable distances are much larger than the Compton wavelength, $\epsilon m  \gg 1$, the Gaussian $G_{\frac{1}{2 m \epsilon}}(k)$ verifies $G_{\frac{1}{2 m \epsilon}}(k)>0$ for $\left|k\right|\ll 1$ and $G_{\frac{2}{m \epsilon}}(k)\sim 0$ otherwise. Thus, in \eqref{fintegral} we have to integrate only over small values of $k$. We find:
\begin{equation}
\begin{aligned}
f^-_{\epsilon}(x) \sim & \, 0, \\
f^+_{\epsilon}(x) \sim & \, \sqrt{m} \int \frac{dk}{\sqrt{2\pi}} e^{i mkx} G_{\frac{1}{2 m \epsilon}}(k)=G_{\epsilon}(x).
\end{aligned}
\end{equation}
This result, plugged back into \eqref{smearedNWfromlocal}, gives:
\begin{equation}
\label{NWlocalconverge}
\hat a_{NW,\epsilon}(jd) \sim  \int dy G_{\epsilon}(jd-y) \hat a(y) = \hat a_{j,\epsilon}\quad \rm{for}\,\,\epsilon m  \gg 1 .
\end{equation}In the limit $\epsilon m \gg 1$, the coarse-grained NW operators still annihilate the global vacuum, hence the latter is a product state of effective local vacua. Equation \eqref{NWlocalconverge} then shows that the global vacuum is also a product state for the coarse-grained field operators. This implies that, in the limit of poor spatial resolution of detectors, an excitation localized ``around point $j$" is effectively described by applying the creation operator $\hat{a}^\dagger_{j,\epsilon}$ to the global vacuum $|\Omega\rangle$. Therefore, any one-particle state $|\psi\rangle = \int dk f(k) \hat{a}^{\dag}_k |\Omega\rangle$ can be effectively described as a sum $\sum_j \tilde{f}(jd) \hat a^{\dag}_{j,\epsilon} |\Omega\rangle$, where $f$ is a function verifying $\int dk |f(k)|^2 =1$ and $\tilde{f}$ its Fourier transform. As a consequence, such a state, which cannot be interpreted as localized in QFT unless it has infinite energy, can now be properly interpreted as localized, allowing a mapping between the description of a region of space in QFT and an effective description that only includes the nonrelativistic degrees of freedom therein contained.

Thus, the structure of entanglement of any state with a finite number of excitations reduces to the entanglement between localized particles, i.e. to the standard, nonrelativistic, picture of entanglement. In particular, the entropy of entanglement of such states is upper bounded by the number of excitations times a factor describing how many states are available to each excitation. The last point can be shown by considering a finite region of space $A$ and its complement $\bar{A}$ at fixed time. Suppose that the field is in a state with $N$ excitations. $A$ is decomposed into $M$ distinct regions $A_1,...,A_M$, whose points are assumed to be nonresolvable because of the limited spatial resolution of the detectors. An upper bound on the entropy of entanglement between subsystems $A$ and $\bar{A}$ is given by the dimension of the subspace of an $M$-mode system containing any number of particles between 0 and $N$:
\begin{equation}
D^M_N=\sum_{n=0}^N C^M_n = \frac{(M+N)!}{M!N!},
\end{equation}
where:
\begin{equation}
 C^M_n = \binom{M+n-1}{n} = \frac{(M+n-1)!}{n!(M-1)!}
\end{equation}is the dimension of the subspace with exactly $n$ particles. This provides an upper bound on the entropy of entanglement between $A$ and $\bar{A}$ for the $N$-particle state:
\begin{equation}
\label{upper-bound-entropy-two-M-mode}
S_A \leq \log D^M_N.
\end{equation}
If $M\gg N \geq 0$,  $\log D^M_N  \sim N\log M$, expressing the fact that the entropy of entanglement of such states is upper bounded by the number of excitations times a factor describing how many states are available to each excitation. One can encode degrees of freedom other than position by changing the value of $M$. For example, if two polarization states are available to each excitation, one must double the value of $M$.

As an example, consider two mesons or two atoms with integer spin in a singlet state localized ``around points $i$ and $j$". In the QFT picture, entanglement between the region ``around point $i$" containing one meson with the rest of the system is infinite. Under the constraint of a bounded spatial resolution of detectors, the effective description of such a system in QFT is:
\begin{equation}
\begin{aligned}
&\frac{1}{\sqrt{2}}[\hat{a}^\dagger_{i, \epsilon,\uparrow}\hat{a}^\dagger_{j, \epsilon,\downarrow}-\hat{a}^\dagger_{i, \epsilon,\downarrow}\hat{a}^\dagger_{j, \epsilon,\uparrow}]|\Omega\rangle \\
=&\frac{1}{\sqrt{2}}\left(|0\rangle_1\cdots|0\rangle_{i-1}|\uparrow \rangle_i|0\rangle_{i+1}\cdots|0\rangle_{j-1} |\downarrow\rangle_j |0\rangle_{j+1}\cdots\right.\\
& \left.+\right|0\rangle_1\cdots|0\rangle_{i-1}|\downarrow \rangle_i|0\rangle_{i+1}\cdots|0\rangle_{j-1} |\uparrow\rangle_j |0\rangle_{j+1}\cdots),
\end{aligned}
\end{equation}which is formally equivalent to the state:
\begin{equation}
\frac{1}{\sqrt{2}}\left[|\uparrow\rangle_i |\downarrow\rangle_j - |\downarrow\rangle_i |\uparrow\rangle_j\right].
\end{equation}The entropy of entanglement between the region ``around $i$" and the rest of the system is then $S_i=\log(2)$, which is the expected value when modeling this system in nonrelativistic quantum theory.

Since finite-energy states correspond to states with a finite number of excitations, this result provides a controlled transition from the QFT picture of entanglement of finite-energy states to the nonrelativistic quantum theory one. Similar conclusions can be drawn for all noncritical bosonic systems, i.e. systems endowed with a finite length scale such as lattice models or models with local interactions and an energy gap (a natural length scale is then provided by the lattice spacing and the correlation length respectively).

\subsection{Area law at high energy and spacetime dynamics}
\label{einstein-eq}

In the previous section, we showed how coarse-graining the spatial resolution of detectors allows for a finite entanglement entropy between two complementary space regions at fixed time. As a by-product, we showed that finite-energy states become localized, allowing an identification between a region of space and the degrees of freedom therein contained. This provided an effective low-energy regularization of the entropy of entanglement in agreement with the finite results one obtains for entanglement measures in nonrelativistic quantum theory. A natural question one may now ask is: what are the fundamental laws at high energy that can render this entropy finite? A naive cutoff at some short distance would select a preferred reference frame, in violation of Lorentz symmetry. Since Lorentz boost symmetry lends the vacuum its thermal character, this seems a rather unlikely means to regulate the entropy. Remarkably, Jacobson and others showed that the assumption of a finite horizon entanglement entropy implies that the spacetime causal structure is dynamical, and that the metric satisfies Einstein’s equation as a thermodynamic equation of state (See \cite{jacobson_thermodynamics_1995,jacobson_gravitation_2012} and references therein).  We are not considering here what UV physics renders the entropy finite, but rather the connection between a high-energy regularization of the area law for entanglement entropy---by any theory still to be elaborated---and Einstein's equation.

Jacobson uses the thermodynamics of a simple homogeneous system as a starting point for his discussion  \cite{jacobson_thermodynamics_1995}. Considering that entropy $S(E,V)$ is a function of energy and volume, one can use the Clausius relation
\begin{equation}
\label{clausius}
\delta Q=T dS.
\end{equation} to deduce an equation of state. Differentiating the first law of thermodynamics $\delta Q=dE+pdV$ yields the identity $dS=(\partial S/\partial E)\, dE + (\partial S/\partial V)\, dV$. One then obtains a relation for temperature and an equation of state :
\begin{equation}
T^{-1}=\partial S/\partial E,\hspace{0.7cm}
p=T(\partial S/\partial V).
\end{equation} If the entropy function $S$ is known, one can derive the global dynamics of the system from the equation of state. Consider for example weakly interacting molecules at low density. One can easily show that 
\begin{align*}
\begin{aligned}
S &=\ln(\# {\rm ~accessible~states}) \\
& \propto \ln V + f(E)
\end{aligned}
\end{align*}for some function $f(E)$. Thus $\partial S/\partial V\propto V^{-1}$,
therefore $pV\propto T$, which is the equation of state of an ideal gas.

\begin{figure}[!htbp]
\begin{center}
\includegraphics[scale=0.45]{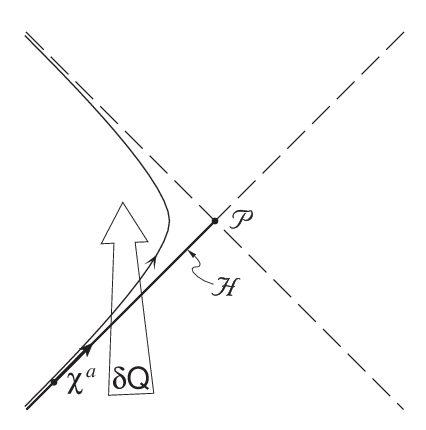}
\end{center}
\caption{If one demands that the Clausius relation $\delta Q=TdS$ hold for all the local Rindler causal horizons through each spacetime point, with $\delta Q$ and $T$ interpreted as the energy flux and the Unruh temperature seen by an accelerated observer just inside the horizon, then the Einstein equation can be viewed as an equation of state (From \cite{jacobson_thermodynamics_1995}).}
\label{local-rindler}
\end{figure}

More generally, thermodynamics defines heat as the energy that flows between degrees of freedom that are not macroscopically \emph{observable}. If we are to analyze spacetime dynamics, Jacobson suggests to define heat as energy that flows across a causal horizon. This suggestion is based on the observation that the overwhelming majority of the information that is hidden resides in correlations between inaccessible UV vacuum fluctuations just inside and outside of the horizon. One then retrieves a notion of hidden information similar to what happens at black hole event horizons when considering the boundary of the past of any spacetime set ${\cal O}$  (${\cal O}$ for ``Observer"). This notion of observer-dependent horizon is a null hypersurface (not necessarily smooth) whose generators are null geodesic segments with origin in the set ${\cal O}$ and emanating backwards in time.

Jacobson formalizes the definition of an observer-dependent horizon by noting that in a small neighborhood of any spacelike 2-surface element ${\cal P}$, one has an approximately flat region of spacetime with the usual Poincaré symmetries. In particular, there exists an approximate Killing field $\chi^a$ generating boosts orthogonal to ${\cal P}$ and vanishing at ${\cal P}$. One can define the ``local Rindler horizon" of any point $p$ in this neighborhood as the past horizon of $p$ with respect to the Killing field $\chi^a$, and ``think of it as defining a system---the part of spacetime beyond the Rindler horizon---that is instantaneously stationary (in `local equilibrium') at $p$" \cite{jacobson_thermodynamics_1995}. Therefore, there are local Rindler horizons in all null directions through any spacetime point. This ``local equilibrium" condition is essential if we are to apply equation \eqref{clausius} (See Fig. \ref{local-rindler}). The heat flow is defined as the energy flowing orthogonally to $\cal P$, and since vacuum fluctuations have a thermal character when seen from the perspective of a uniformly accelerated observer, we can identify the temperature of the system into which the heat is flowing with ``the Unruh temperature associated with an observer hovering just inside the horizon" \cite{jacobson_thermodynamics_1995}. Jacobson arrives at the conclusion that ``we can consider a kind of local gravitational thermodynamics associated with such causal horizons, where the `system' is the degrees of freedom beyond the horizon" \cite{jacobson_thermodynamics_1995}. Thus, the spacetime geometry cannot be inert because ``the light rays generating the horizon must focus so that the area responds to the flux of energy in just the way implied by the Einstein equation (at least at long distances)" \cite{jacobson_thermodynamics_1995}. The constant of proportionality $\eta$ between the entropy and the area determines Newton's constant as $G=c^3/4\hbar\eta,$ which identifies the length $\eta^{-1/2}$ as twice the Planck length $L_p=\sqrt{\hbar G/c^3}.$ The cosmological constant $\Lambda$ remains undetermined.

\bigskip

In summary, we analyzed how relative motion of inertial observers affects entanglement, and showed that the violation of the Bell inequality still holds if local observers fine tune their measurement apparatus. For non-inertial observers, the Unruh effect requires a QFT description of systems, and the thermal character of the vacuum state for uniformly accelerated observers implies a degradation of entanglement, suggesting the possibility of a relationship between entanglement thermodynamics at event horizons and accelerated motion/curvature. We also considered the possibility of local manipulation of entanglement for field systems, and introduced a novel approach to regularize the divergence of entanglement entropy for localized systems. The idea is based on the fact that all finite-energy states produce finite results for entanglement measures for the NW localization scheme. We showed that the NW localization scheme coincides with the standard one at low energy, and therefore entails a an operationally well-defined TPS. Finally, we quickly reviewed the relationship between the intuited thermodynamics of entanglement at event horizons and Einstein's general relativity equation.

\chapter{Beyond entanglement and definite causal structures}
\label{third-chapter}

In Chapter \ref{first-chapter}, the analysis of general non-signalling correlations clearly showed in what sense quantum correlations respect causality. On the contrary, it appeared that entanglement is much more deeply entrenched in relativistic quantum theory because of the constraints causality imposes on the algebraic structure of observables. In Chapter \ref{second-chapter}, we focused on more technical issues concerning entanglement detection and quantification in the relativistic setting. In this chapter, we expand our analysis of the interplay between quantum correlations and the causal structures ordering measurements events. We begin by a quick review of how causal relations can be formulated operationally. The historical example of the relationship between specific exotic causal structures allowed by general relativity, the so-called closed-timelike curves (CTC), and some non-unitary quantum computation models will serve as a guideline. We then present a recent quantum operational framework that provides such an operational formulation of causal correlations while going beyond causally ordered quantum correlations. We place such correlations in a more general probabilistic framework in order to analyze the connections between the local ordering of events and the emergence of an indefinite global order. We also analyze how various informational principles can partially account for the structure of quantum correlations with indefinite causal order. We end with a discussion of an alternative attempt to go beyond fixed causal structures.

\section{Operational approaches to causal relations}

As shown in the first chapter, significant progress has been made in understanding quantum theory and the structure of quantum correlations in an operational context where primitive laboratory procedures, like preparations transformations and measurements (PTM) are basic ingredients. Thus far, the historical challenge quantum nonlocality posed to causality led to approaches primarily focused on spacelike separated experiments, the main goal being the identification of a complete set of physical principles which select the non-signalling quantum correlations out of the strictly larger class of non-signalling ones. In this line of research, spacetime is typically regarded as a given, predefined ``stage" in which the causal relations between events are defined.

Consider for instance the seminal reconstruction proposed by Hardy in \cite{hardy_quantum_2001}. There physical systems are defined by two numbers: the number of degrees of freedom $K$, representing the minimum number of measurements to determine the state of the system, and the dimension $N$, corresponding to the maximum number of states perfectly distinguishable in one measurement of the system. The assumption of a global causal structure is encoded in how systems compose. Indeed, consider a composite system with subsystems $A$ and $B$. Hardy's fourth axiom expresses the operationally defined parameters $K_{AB}$ and $N_{AB}$ of the composite system in terms of the parameters of subsystems $A$ and $B$:
\begin{equation}
N_{AB}=N_A N_B, \hspace{0.3cm} K_{AB}=K_A K_B.
\end{equation}
This definition implies that only a super-observer can calculate $K_{AB}$ and $N_{AB}$, for it requires PTM on each subsystem by the same observer, even if $A$ and $B$ are not localized in the same laboratory. This in turn implies the existence of a global structure ordering PTM events that occur in the frame of the super-observer. To take another example, Rovelli argued informally that quantumness follows from a limit on the amount of ``relevant'' information that can be extracted from a system \cite{rovelli_relational_1996}. If the notion of relevance is to be connected to lattice orthomodularity in the quantum logical framework \cite{grinbijqi}, the ensuing reconstruction of quantum theory will fundamentally depend on the order of binary questions asked to the system. For many systems, it requires the existence of a global causal structure ordering all incoming information.

One of the first attempts to go beyond such frameworks is due to Deutsch's CTC model \cite{deutsch_quantum_1991}. CTCs appeared in van Stockum's 1937 paper which provided a solution to the Einstein field equation that corresponds to an infinitely long cylinder made of rigidly and rapidly rotating dust \cite{vanstockum}. Thorne commented on this model as follows \cite{thorne_closed_1993}:
\begin{quote}
``The dust particles are held out against their own gravity by centrifugal forces, and inertial frames are so strongly dragged by their rotation that the light cones tilt over in the circumferential direction [as shown in Figure \ref{vanstockum}], causing the smallest circle to be a CTC. CTCs pass through every event in the spacetime, even an event on the rotation axis where the light cone is not tilted at all: one can begin there, travel out to the vicinity of the [smallest] circle (necessarily moving forward in $t$ as one travels), then go around the cylinder a number of times traveling backward in $t$ as one goes, and then return to the rotation axis, arriving at the same moment one departed."
\end{quote} 
Such a model was generally dismissed as unphysical because its source is infinitely long \cite{thorne_closed_1993}. A second famous example of a CTC corresponds to the solution to the Einstein equation provided by Gödel in 1949 \cite{godel_example_1949}. Thorne commented that ``physicists have generally dismissed Gödel's solution as unphysical because it requires a nonzero cosmological constant and/or it does not resemble our own universe (whose rotation is small or zero)" \cite{thorne_closed_1993}. 

Many other studies of chronology-violating spacetimes followed, but they all generally treat CTCs as properties of spacetime geometry and use the methods of differential geometry and general relativity. One of the drawbacks of such approaches is their inability to distinguish between merely counterintuitive effects and unphysical ones, especially when technical and conceptual issues about wormholes and sigularities enter the discussion \cite{deutsch_quantum_1991}. For instance, the ``grandfather" or ``knowledge creation" paradoxes receive no clear answers in such frameworks. According to Deutsch, these approaches suffer from a second drawback, much more serious and deep \cite{deutsch_quantum_1991}:
\begin{quote}
 ``[...] Classical spacetime models do not take account of quantum mechanics which, even aside from any effects of quantum gravity, actually dominates both microscopic and macroscopic physics on and near all [CTCs]." 
 \end{quote}
\begin{figure}
\centering
\includegraphics[scale=0.4]{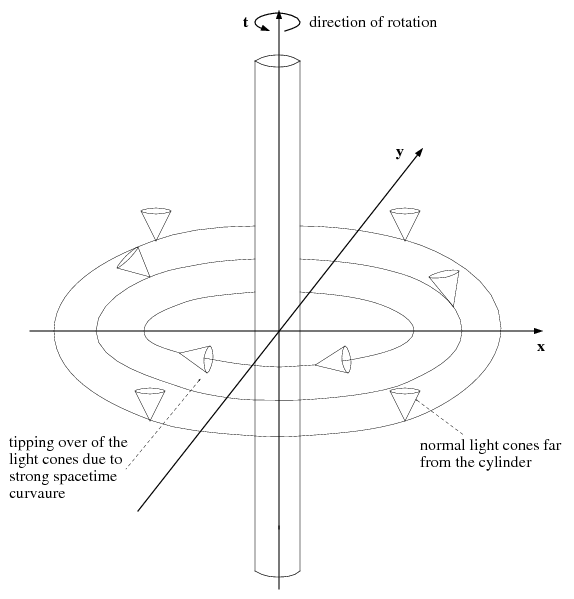}
\caption{Van Stockum's spacetime (From \cite{lobo_closed_2010}).}
\label{vanstockum}
\end{figure}

Adopting the operational point of view of quantum computation, Deustch analyzed how quantum physical systems would behave near CTCs. His model avoids paradoxes such as the ``grandfather" paradox by demanding self-consistent solutions for the time-travelling systems: a time-traveling qubit state $\rho_{CTC}$ must obey the dynamical equation
\begin{equation}
\rho_{CTC}=\trace_{TR}\left[U^\dag(\rho_{CTC}\otimes\rho_{in})U\right],
\end{equation}where the partial trace is over the time-respecting qubit (See Fig. \ref{dctc}). This self-consistency requirement defines multiple solutions for $\rho_{CTC}$, each of which corresponds to a (generally non-linear) map on $\rho_{in}$, which can be worked out from the solution $\rho_{CTC}$. The price to pay is that one needs to extend quantum theory in order to account for non-unitary dyncamics. An alternative CTC model based on the Horowitz-Maldacena ``final state condition" for black hole evaporation \cite{horowitz_black_2004} and on the suggestion of Svetlichny that teleportation and postselection could be used to describe time travel \cite{svetlichny_effective_2011} was shown to be inequivalent to the Deutsch model \cite{lloyd_closed_2011,bub_quantum_2014}. To see the underlying intuition, consider a simple teleportation experiment. Alice and Bob share a Bell state
\begin{equation}
|\phi_+\rang =\frac{1}{\sqrt{2}}\left(|0\rang_A|0\rang_B+|1\rang_A|1\rang_B\right),
\end{equation}while Alice has an additional qubit in state $|\psi\rang_C$. Alice measures qubits $A$ and $C$ in the Bell basis and communicates the result to Bob, who applies an appropriate unitary depending on which of the four outcomes occurred. If the outcome of the Bell measurement corresponds to the original Bell state $|\phi_+\rang$ then Bob does not need to do anything. In the words of Lloyd \emph{et al.} \cite{lloyd_closed_2011}:
\begin{quote}
``In this case, Bob possesses the unknown state even before Alice implements the teleportation. Causality is not violated because Bob cannot foresee Alice's measurement result, which is completely random. But, if we could pick out only the proper result, the resulting `projective' teleportation would allow us to travel along spacelike intervals, to escape from black holes, or to travel in time. We call this mechanics a projective or postselected CTC, or P-CTC."
\end{quote}
This model also introduces a non-linear modification to quantum theory theory through postselection. However, it is well known that non-linear extensions of quantum theory allow for arbitrarily fast communication \cite{gisin_weinbergs_1990} (which leads to the aforementioned ``preparation problem" \cite{cavalcanti_preparation_2012}) or state cloning (which leads to perfect distinguishability of quantum states) \cite{brun_localized_2009,brun_quantum_2013,ahn_quantum_2012}, both of which are problematic. Thus, it appears that if we are to go beyond spacetime as a predefined ``stage" ordering laboratory events, preserving the linearity of the standard quantum framework is a requirement.

\begin{figure}
\centering
\includegraphics[scale=0.4]{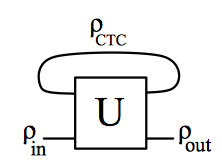}
\caption{Deustsch's CTC model.}
\label{dctc}
\end{figure}

A first step towards such a framework is to look for the simplest unified description of timelike and spacelike separated correlations, and then see what the new framework has to say about causal relations \cite{hardy_probability_2005,brukner_quantum_2014}. Indeed, descriptions of timelike and spacelike separated correlations are very different in the standard (causal) formalism of quantum theory. For example, correlations between results obtained on a pair of spacelike separated systems are described by a joint state on the tensor product of two Hilbert spaces, whereas those obtained from measuring a single system at different times are described by an initial state and a map on a single Hilbert space. Various results such as  the isomorphism between spatial and temporal quantum correlations \cite{brukner_quantum_2004,leifer_quantum_2006,marcovitch_structural_2011,fitzsimons_quantum_2013} , and between the Bell inequality and its temporal analogue the Leggett-Garg inequality\footnote{One can formalize a hypothesis called `macroscopic realism' which can be falsified by a violation of the Leggett-Garg inequality, in analogy with the falsification of local realism by a violation of the Bell inequality.}\cite{leggett_quantum_1985,fritz_quantum_2010,budroni_bounding_2013} indicate that a unified quantum description may be possible for experiments involving distinct systems at one time and those involving a single system at distinct times.

\section{The process matrix framework}

Recently, a framework that fullfils the two mains requirements we formulated was introduced by Oreshkov, Costa and Brukner \cite{oreshkov_quantum_2012}: it provides a unified description of timelike and spacelike separated correlations \emph{and} preserves the linear structure of quantum theory. There it is assumed that operations in local laboratories are described by quantum theory, i.e. by (trace non-increasing) completely positive (CP) maps, but no assumption is made on the existence or not of a global causal structure ordering events occuring at local laboratories. One can use this framework to describe even more general quantum correlations, for which the causal ordering of events and whether they take place between spacelike or timelike regions is not fixed.

\subsection{General framework}
\label{section-3-2-1}

Consider two laboratories whose ``agents" are called Alice and Bob. Assume they are equipped with random bit generators and that Alice and Bob are capable of free choice. At each run of the experiment, each laboratory receives exactly one physical system, performs transformations allowed by quantum theory and subsequently sends the system out. We assume that each laboratory is isolated from the rest of the world, except when it receives or emits the system.

Denote the input and the output Hilbert spaces of Alice by $\mathcal{H}^{A_1}$ and $\mathcal{H}^{A_2}$ and those of Bob by $\mathcal{H}^{B_1}$ and $\mathcal{H}^{B_2}$. The sets of all possible outcomes of a quantum instrument at Alice's, respectively Bob's, laboratory corresponds to the set of completely positive (CP) maps  $\{\mathcal{M}^{A_1 A_2}_i\}_{i=1}^n$, respectively $\{\mathcal{M}^{B_1 B_2}_j\}_{j=1}^n$. It is convenient to represent a CP map as a positive semi-definite matrix via the Choi-Jamio\l kowski (CJ) isomorphism \cite{jamiolkowski_linear_1972,choi_completely_1975}. The CJ matrix $M^{A_1 A_2}_i \in\calL(\mathcal{H}^{A_1}\otimes \mathcal{H}^{A_2})$ corresponding to a linear map $\mathcal{M}^{A_1 A_2}_i: \mathcal{L}(\mathcal{H}^{A_1})\longrightarrow\mathcal{L}(\mathcal{H}^{A_2})$ at Alice's laboratory is defined as
\begin{equation}
M_i^{A_1A_2}= \left[\id\otimes \calM^{A_1A_2}_i(|\phi^+\rang\lang \phi^+|)\right]^T,
\end{equation}where $\{|j\rang\}_{j=1}^{d_{A_1}}$ is an orthonormal basis of $\calH^{A_1}$, $|\phi^+\rang =\sum_{j=1}^{d_{A_1}} |jj\rang \in \calH^{A_1}\otimes \calH^{A_1}$ is a (non-normalized) maximally entangled state and $T$ denotes matrix transposition. Similarly, one can associate to a CP map $\mathcal{M}^{B_1 B_2}_j:\mathcal{L}(\mathcal{H}^{B_1})\longrightarrow\mathcal{L}(\mathcal{H}^{B_2})$
at Bob's laboratory a CJ operator $M^{B_1 B_2}_j$ acting on $\mathcal{H}^{B_1}\otimes \mathcal{H}^{B_2}$ (See Fig. \ref{causal-game}).

Using this correspondence, the non-contextual probability for two measurement outcomes can be expressed as a bilinear function of the corresponding CJ operators:
\begin{equation}
\label{extension-born-rule}
P(\mathcal{M}^{A_1 A_2}_i,\mathcal{M}^{B_1 B_2}_j)=\trace\left[W^{A_1 A_2 B_1 B_2} \left(M^{A_1 A_2}_i\otimes M^{B_1 B_2}_j\right) \right],
\end{equation}
where $W^{A_1 A_2 B_1 B_2} \in \mathcal{L}(\mathcal{H}^{A_1}\otimes\mathcal{H}^{A_2}\otimes
\mathcal{H}^{B_1}\otimes\mathcal{H}^{B_2})$ is fixed for all runs of the experiment. Requiring that such probabilities be non-negative for any choice of CP maps and equal to $1$ for any choice of CP trace-preserving (TP) maps\footnote{These correspond to operators $M^{A_1A_2}>0$ and $M^{B_1B_2}>0$ verifying $\trace_{A_2}M^{A_1A_2}=\id^{A_1}$ and $\trace_{B_2} M^{B_1B_2}=\id^{B_1}$ respectively.} yields a space of valid $W$ operators referred to as \emph{process matrices} (See Fig. \ref{allowed}). A process matrix can be understood as a generalization of a density matrix and equation \eqref{extension-born-rule} can be seen as a generalization of Born's rule, or equivalently an extension of Gleason's theorem to CP maps. In fact, when the output systems $A_2$ and $B_2$ are taken to be one-dimensional, i.e. when each party performs a measurement and then discards the system, expression \eqref{extension-born-rule} reduces to:
\begin{equation}
P(\mathcal{M}^{A_1 A_2}_i,\mathcal{M}^{B_1 B_2}_j)=\trace\left[W^{A_1 B_1 } \left(M^{A_1}_i\otimes M^{B_1 }_j\right) \right],
\end{equation}where now $M_i^{A_1},M_j^{B_1}$ are local POVMs, and the probability 1 condition for CPTP maps becomes $\trace (W^{A_1B_2})=1$, i.e. $W^{A_1B_1}$ is a quantum state.

Note that any process matrix can be interpreted as a CPTP map from the outputs $A_2,B_2$ of the parties to their inputs $A_1,B_1$. In other words, any process can be thought of as a \emph{noisy} CTC (See Fig. \ref{not-allowed}).

\begin{figure}
\begin{center}
\includegraphics[scale=0.2]{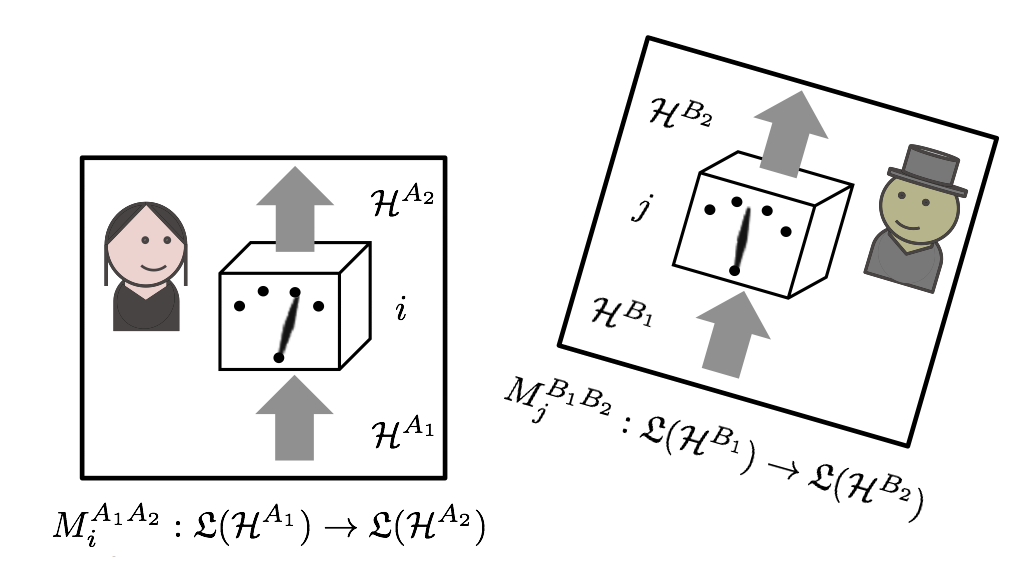}
\end{center}
\caption{Two parties Alice and Bob each receive an input system on which they perform operations allowed by quantum theory before sending it out from the laboratory. While events are ordered within each laboratory, no assumption is made on the existence of a global causal structure ordering events at Alice's and Bob's laboratories (From \cite{oreshkov_quantum_2012}).}
\label{causal-game}
\end{figure}

\begin{figure}
\begin{center}
\includegraphics[scale=0.4]{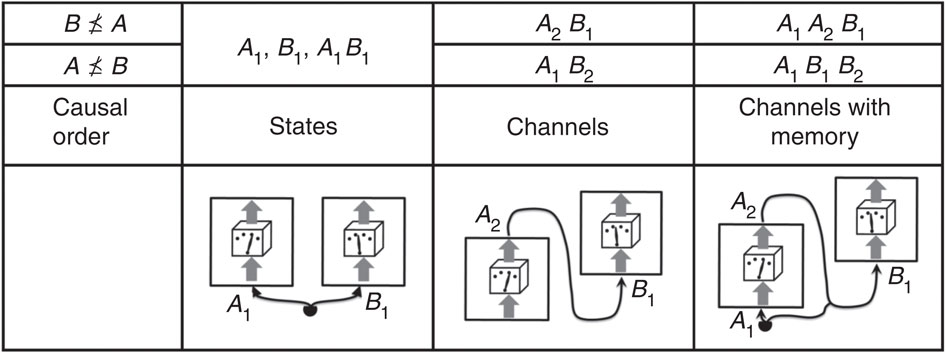}
\end{center}
\caption{A matrix satisfying the positivity condition can be expanded as $W^{A_1A_2B_1B_2}=\sum_{\mu\nu\lambda\gamma}w_{\mu\nu\lambda\gamma} \sigma_\mu^{A_1}\otimes\sigma_\nu^{A_2}\otimes\sigma_\lambda^{b_1}\otimes \sigma_\gamma^{B_2}$, where $w_{\mu\nu\lambda\gamma}\in\mathbb{R}$ and the set of matrices $\{\sigma_\mu^X\}_{\mu=0}^{d_X^2-1}$, with $\sigma_0^X=\id^X$, $\trace(\sigma_\mu^X\sigma_\nu^X)=d_X\delta_{\mu\nu}$ and $\trace(\sigma_j^X)=0$ for $j=1,...,d_X^2$ is the Hilbert-Schmidt basis of $\mathfrak{L}(\calH^X)$. We refer to terms of the form $\sigma_i^{A_1}\otimes \id^{rest}$ $(i\geq 1)$  as of the type $A_1$, terms of the form  $\sigma_i^{A_1}\otimes\sigma_j^{A_2}\otimes \id^{rest}$ ($i,j\geq 1$) as of the type $A_1A_2$ and so on. Matrices that also satisfy the condition on CPTP maps are listed in this table. Each of the terms can allow signalling in at most one direction and can be realized in a situation in which either Bob’s actions are not in the causal past of Alice's $B\npreceq A$ or vice versa $A\npreceq B$. The most general process matrix can contain terms from both rows and may not be decomposable into a mixture of quantum channels from Alice to Bob and from Bob to Alice (From \cite{oreshkov_quantum_2012}).}
\label{allowed}
\end{figure}

\begin{figure}
\begin{center}
\includegraphics[scale=0.4]{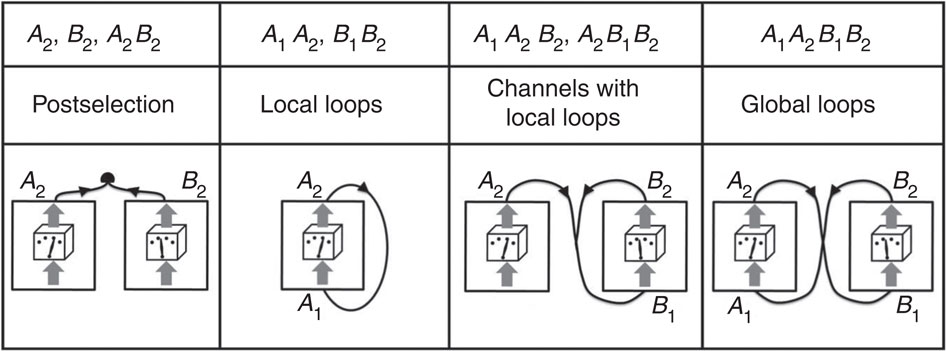}
\end{center}
\caption{These terms are not compatible with local quantum mechanics because they yield non-unit probabilities for some CPTP maps (From \cite{oreshkov_quantum_2012}).}
\label{not-allowed}
\end{figure}

\subsection{Theory-independent tests via a causal inequality}
\label{section-3-2-2}

One can study the causal properties of correlations in a theory-independent way by recording experimental data and defining a `causal inequality' that should be verified by all (mixtures of) causally ordered correlations. In analogy with the Bell inequality, the violation of such a causal inequality would imply a deep revision of our understanding of allowed causal structures.

Suppose that two parties Alice and Bob play a game (in the sense we defined in Chapter \ref{first-chapter}) where each party receives a system. Each of them tosses a coin, whose value is denoted by $a$ for Alice and $b$ for Bob. Bob further tosses a coin to produce a random task bit $b'$ with the following meaning: if $b' = 0$, Bob must communicate $b$ to Alice, and if $b' = 1$, Bob must guess the value of $a$. Both parties always produce a guess, denoted by $x$ for Alice and $y$ for Bob. It is crucial to assume that the bits $a$, $b$, and $b'$ are random.

The goal of Alice and Bob is to maximize the probability of success:
\begin{equation}
P_{success}=\frac{1}{2} \left[p(x=b|b'=0)+p(y=a|b'=1)\right],
\end{equation}
i.e. Alice should guess Bob's toss, or vice versa, depending on the value of $b'$. We will refer to this game as the Oreshkov-Costa-Brukner (OCB) game in the remainder of the thesis. If all events occur in a causal sequence, then
\begin{equation}
\label{causineq}
P_{success}\leq \frac{3}{4}.
\end{equation}
Indeed, since Alice and Bob perform their experiments inside closed laboratories, at most unidirectional signalling between the laboratories is allowed. Therefore, it is true that either Alice cannot signal to Bob or Bob cannot signal to Alice. Consider the latter case. If $b' = 1$, Alice and Bob could in principle achieve up to $P(y = a|b' = 1) = 1$. However, if $b' = 0$, Alice can only make a random guess, hence $P(x = b|b' = 0) = \frac{1}{2}$ and the probability of success in this case satisfies \eqref{causineq}. The same argument shows that the probability of success will not increase when Alice cannot signal to Bob or under any mixing strategy. Violation of inequality \eqref{causineq} would mean that it is impossible to interpret local events as occurring within a global causal structure.

\subsection{Causally non-separable processes}
\label{caus-non-sep}

Surprisingly, processes that violate causal inequality \eqref{causineq} do exist. Consider the following process matrix:
\begin{equation}
\label{W-violation}
W^{A_1 A_2 B_1 B_2}=\frac{1}{4}\left[\id^{A_1 A_2 B_1 B_2}+\frac{1}{\sqrt{2}}\left(\sigma^{A_2}_z\sigma^{B_1}_z+\sigma^{A_1}_z\sigma^{B_1}_x\sigma^{B_2}_z\right)\right],
\end{equation}
where $A_1,A_2,B_1$, and $B_2$ are two-level systems and $\sigma_x, \sigma_y$ and $\sigma_z$ are the usual Pauli matrices.
Consider the following CP maps at Alice's and Bob's laboratories respectively:
\begin{equation}
\label{local-maps}
\begin{aligned}
\xi^{A_1 A_2}(x,a,b') & =\frac{1}{2}\left[\id+(-1)^x\sigma_z\right]^{A_1}\otimes \left[\id+(-1)^a\sigma_z\right]^{A_2},\\
\eta^{B_1 B_2}(y,b,b')&=b'\cdot \eta^{B_1 B_2}_1(y,b,b')  +(b'\oplus 1)\cdot \eta^{B_1 B_2}_2(y,b,b'),
\end{aligned}
\end{equation}
where
\begin{equation}
\eta^{B_1 B_2}_1(y,b,b') =\frac{1}{2}\left[\id+(-1)^y\sigma_z\right]^{B_1}\otimes \id^{B_2}
\end{equation}
 and
\begin{equation}
 \eta^{B_1 B_2}_2(y,b,b')=\frac{1}{2}\left[\id^{B_1 B_2}+(-1)^b\sigma_x^{B_1}\sigma_z^{B_2}\right].
 \end{equation}
Inserting these expressions into \eqref{extension-born-rule} shows that the success probability associated to \eqref{W-violation} and \eqref{local-maps} violates causal inequality \eqref{causineq}:
\begin{equation}
\label{quantum-bound}
P_{success}=\frac{2+\sqrt{2}}{4}>\frac{3}{4}.
\end{equation}
Hence it is impossible to interpret local events as occurring within a global causal structure. This is an example of a causally non-separable process, viz. a process that cannot be written as (a mixture of) causal processes:
\begin{equation}
W\neq \lambda W^{A \npreceq  B}+(1-\lambda )W^{B\npreceq  A},\label{causform}
\end{equation}where $0\leq \lambda \leq 1$, $W^{A\npreceq  B}$ is a process in which Alice cannot signal to Bob and $W^{B\npreceq  A}$ a process in which Bob cannot signal to Alice. ``Cannot signal'' here means either that the channels go in the other direction or that parties share a bipartite state. If a process matrix $W$ can be written in the form (\ref{causform}), it will be called causally separable. The violation of the causal inequality in \eqref{quantum-bound}, which invalidates the assumption of a predefined causal order between events, plays a role similar to the quantum violation of the Bell inequality invalidating the existence of local hidden variables.

One can show that bipartite probability distributions generated by local classical operations are always causally separable \cite{oreshkov_quantum_2012}. However, a process for three parties has recently been found in which perfect signalling correlations among three parties are possible even if local operations are assumed to be classical \cite{baumeler_maximal_2014}. Moreover, Brukner showed that the $\frac{2+\sqrt{2}}{4}$ bound on the violation of \eqref{causineq}---which we call quantum bound---is maximal for qubits and under a restricted set of local operations involving traceless binary observables \cite{brukner_bounding_2014}. As we already argued, processes can be understood as noisy CTCs, so it seems that there is ``just the right amount of noise" that allows to preserve the linear structure of quantum theory \emph{and} still go beyond definite causal structures. Therefore we should aim at a better understanding of the origin of the quantum bound.

\section{Boxes compatible with predefined causal order}

Following the approach we reviewed in Chapter \ref{first-chapter} to the study of entanglement, we introduce a generalized probabilistic framework that accounts for classical, quantum---including standard timelike and spacelike separated quantum correlations---and supra-quantum correlations of the OCB game under specific constraints on the local operations. As we shall see, this framework is not suited for a study of the quantum bound on correlations with indefinite causal order in full generality, rather it lays the groundwork for a better understanding of the connections between the local ordering of events and the role of the control bit in the emergence of an indefinite global order\footnote{This work has been published in \cite{ibnouhsein_information_2014}.}.

\subsection{Local order and causal separability of processes}

Consider as in the OCB game two parties Alice and Bob with inputs $a,b$ and outputs $x,y$ with obvious notations. Bob also possesses a control bit $b'$. Now, suppose we are given a quantum process matrix and a strategy (with local quantum operations) by means of which we realize a specific joint probability distribution $p(x,y|a,b)$ after tracing over the control bit $b'$
\begin{equation}
p(x,y|a,b) = \sum_{\alpha} p(x,y|a,b,b'=\alpha)p(b'=\alpha),
\end{equation}thus yielding a new effective strategy. We show that if the effective strategy can be realized using \emph{fixed} local quantum instruments, i.e. independent of $a$ and $b$, then there exists an equivalent diagonal quantum process by means of which we obtain the same probabilities $p(x,y|a,b)$ for all $a,b,x,y$. Since a diagonal bipartite process is causally separable, $p(x,y|a,b)$ thus cannot violate \emph{any} causal inequality. Hence a property we call ``compatibility with predefined causal order" (CPO) shared by both causally separable and non-separable processes for a certain subset of local operations which includes, as we shall see below, the operations of Section \ref{caus-non-sep} allowing a violation of the causal inequality. The CPO property will serve as a basis for the construction of a generalized probabilistic framework, paralleling the use of the no-signalling principle in the construction of PR-boxes.

It is crucial for CPO to be true that the effective local operations can be taken diagonal in a \emph{fixed} local basis so that there exists a \emph{single} diagonal process matrix that yields the joint probabilities \emph{for all $a,b,x, y$}. Obviously, if $x$ and $y$ are produced before $a$ and $b$, then the quantum instruments whose outcomes yield $x$ and $y$ cannot depend on $a$ and $b$, and hence can be considered as fixed. Operations used in Section \ref{caus-non-sep} to violate the causal inequality verify
\begin{equation}
 b'  \preceq    y \preceq b \quad \mbox{ and } \quad x  \preceq a,
 \end{equation}
and therefore can be considered as independent of $a$ and $b$ for each fixed value of $b'$, thus yielding fixed effective strategies.

Using this mapping of the initial strategy to an effective fixed strategy, we can now prove the equivalence of the initial process to a diagonal one:
\begin{itemize}
\item[$\circ$] By assumption, for each value of $b'$, the most general strategy for Bob is to apply a fixed quantum instrument denoted by $I_1(b')$ on the input system, whose outcome yields $y$, and to subject the output system of that instrument to a subsequent CPTP map dependent on the value of $b$ denoted by $I_2(b',b)$.
\item[$\circ$] For each value of $b'$, the first quantum instrument $I_1(b')$ can be implemented by a unitary $U_1(b')$ on the input system plus an ancilla, followed by a projective measurement $P(b')$ on part of the resulting joint system \cite{oreshkov_quantum_2012}. The CPTP map $I_2(b',b)$ can be implemented by a unitary $U_2(b')$ applied on the output of $I_1(b')$, an ancilla, and a qubit prepared in the state $|b\rangle$ (we feed $b$ in the form of a quantum state $|b\rangle$, where different vectors $|b\rangle$ are orthogonal).
\item[$\circ$] The projective measurement $P(b')$ and the preparation of $|b\rangle$ fully define Bob's operation: other transformations as well as the ancillas can be seen as occuring outside Bob's laboratory by attaching them to the original process \emph{before the input}, which yields a new equivalent process with a new process matrix\footnote{Here lies the aforementioned connection between an effective \emph{fixed} strategy and the existence of a \emph{single} effective process: if the first local unitary before the projective measurement depends on $a$ or $b$, then for each particular value of $a$ or $b$ we can pull it out of the laboratory before the input system, but this does not yield one \emph{single} quantum process from which $p(x,y|a,b)$ is obtained with diagonal operations \emph{for all} $a,b,x$, and $y$.}. If the original matrix were valid, then whatever Bob may choose to do on his redefined input and output systems could have occurred anyway and would have yielded valid probabilities, hence the redefined process matrix is also valid. As a result, we obtain that the correlations for each value of $b'$ are equivalent to the correlations obtained by diagonal measurement and repreparation operations, i.e. classical local operations.
\end{itemize}
Here we focused on operations in Bob's laboratory, but similar arguments hold for operations in Alice's laboratory (which are independent from $b'$).

In conclusion, we exhibited a property shared by both causally separable and non-separable processes under specific constraints on the local operations of the OCB game, namely $x \preceq a$ and $b' \preceq y \preceq b$. These constraints are verified by the operations used in \cite{oreshkov_quantum_2012} to violate the causal inequality, but it should be emphasized that they are not necessary for such a violation to occur. Therefore, the general framework one can build using CPO is probably not suited for the study of the origin of the quantum bound on correlations with indefinite causal order. It rather aims at a better understanding of the connections between the local ordering of events at Alice's and Bob's laboratories and the role of the control bit in the emergence of an indefinite global order.

\subsection{Generalized probabilistic framework}

The previous discussion of the CPO property is based on several notions defined withing the process matrix framework, such as a diagonal process or causal separability. One can provide an alternative definition using only the input bits, the output bits and the notion of causal order as primitives.

Consider two parties Alice and Bob sharing a box with inputs $a,b$ and outputs $x,y$ with obvious notations. Bob also possesses a control bit $b'$, and we assume that $b' \preceq y \preceq b$ and $x \preceq a$. A formal definition of CPO is the following:
\begin{align}
\begin{aligned}
& p(x,y|a,b)\\
=& p(x,y,A \npreceq B|a,b) + p(x,y,B \npreceq A|a,b) - p(x,y,A \npreceq \nsucceq B|a,b)\\
= & p(A \npreceq B|a,b) p(x,y|a,b,A \npreceq B)
+p(B \npreceq A|a,b) p(x,y|a,b,B \npreceq A)\\
-&p(A \npreceq \nsucceq B|a,b) p(x,y|a,b,A \npreceq \nsucceq B)\\
\end{aligned}
\end{align}
In the OCB game, if parties are causally ordered then their order is implicitly assumed to be independent from the tossed bits, therefore we restrict our attention to boxes that verify $p(B \npreceq A|a,b) = p(B \npreceq A)$, $p(A\npreceq B|a,b) = p(A\npreceq B)$ and $p(A \npreceq \nsucceq B|a,b)=p(A \npreceq \nsucceq B)$. Furthermore, the definitions of $A\npreceq B$, $B\npreceq A$ and $A \npreceq \nsucceq B$ imply that:
\begin{align}
\begin{aligned}
p(x,y|a,b)=&p(A \npreceq B) \sum_{b'} p(b'|a,b,A \npreceq B) p(y|b,b',A\npreceq B)p(x|a,b,b',A\npreceq B)\\
+&p(B \npreceq A)  \sum_{b'} p(b'|a,b,B \npreceq A)p(x|a,b',B \npreceq A)p(y|a,b,b',B \npreceq A)\\
-&p(A \npreceq \nsucceq B)  \sum_{b'} p(b'|a,b,A \npreceq \nsucceq B)p(x|a,b',A \npreceq \nsucceq B)p(y|b,b',A \npreceq \nsucceq B),
\end{aligned}
\end{align}where we introduced the control bit $b'$ to have the full dependencies between input and output bits. The condition that $b'$ can be freely chosen implies that:
\begin{equation}
p(b'|a,b,B \npreceq A)=p(b'|a,b,A\npreceq B)=p(b'|a,b,A \npreceq \nsucceq B)=p(b'),
\end{equation}therefore the CPO condition for the OCB game can be written as:
\begin{align}
\begin{aligned}
p(x,y|a,b)=&p(A \npreceq B) \sum_{b'} p(b') p(y|b,b',A\npreceq B)p(x|a,b,b',A\npreceq B)\\
+&p(B \npreceq A)  \sum_{b'} p(b')p(x|a,b',B \npreceq A)p(y|a,b,b',B \npreceq A)\\
-&p(A \npreceq \nsucceq B)  \sum_{b'} p(b')p(x|a,b',A \npreceq \nsucceq B)p(y|b,b',A \npreceq \nsucceq B).
\end{aligned}
\end{align}
Thus, we see that the ``quantumness" of correlations with indefinite causal order in the context of the OCB game does not lie in the correlations between inputs $a,b$ and outputs $x,y$ only: the free character of the control bit $b'$, i.e. of what Bob chooses to do plays a crucial role, at least when additional constraints are placed on local operations, as done implicitly in the OCB game of Section \ref{caus-non-sep} violating the causal inequality.

\section{Entropic characterizations of causal structures}

The quantum bound on correlations with indefinite causal order is lower than what is algebraically possible. Thus, it can be seen as a figure of merit characterizing the types of causal correlations that are allowed within the process matrix framework. Since the latter was derived by relaxing the assumption of a global causal structure, one may ask whether the quantum bound can be derived by formulating the constraints on the signalling possibilities imposed by a fixed structure using entropic quantities. The intuition behind such an approach is the derivation of the Tsirelson bound using an entropic characterization of the signalling possibilities between Alice and Bob in the LOCC paradigm (See the protocol we detailed in Chapter \ref{first-chapter}).

\subsection{Constrained signalling and mutual information}

Using mutual information for an entropic characterization of causal structures is natural because it is a measure of dependence between inputs and outputs of Alice's and Bob's laboratories. Measuring dependence describes how much information two random variables share with each other, i.e. the amount of uncertainty about one random variable given knowledge of the other random variable. Rényi proposed the following conditions that any measure of dependence should satisfy \cite{renyi_new_1959,renyi_measures_1959}:
\begin{itemize}
\item[(i)] It is defined for any pair of random variables.
\item[(ii)] It is symmetric.
\item[(iii)] Its value lies between 0 and 1.
\item[(iv)] It equals 0 iff the random variables are independent.
\item[(v)] It equals 1 if there is a strict dependence between the random variables.
\item[(vi)] It is invariant under marginal one-to-one transformations of the random variables.
\item[(vii)] If the random variables $X,Y$ are Gaussian distributed, it equals the absolute value of their correlation coefficient
\begin{equation}
\rho(X,Y)=\frac{E\left[(X-\mu_x)(Y-\mu_Y)\right]}{\sigma_X \sigma_Y},
\end{equation}
where $\mu_X$ and $\mu_Y$ are the means of $X$ and $Y$ respectively, and $\sigma_X$ and $\sigma_Y$ are their standard deviations.
\end{itemize}
In a fixed causal structure, either Alice cannot signal to Bob or Bob cannot signal to Alice\footnote{We recall that our discussion is still constrained by the conditions of the causal game, most importantly Alice and Bob perform their experiments inside closed laboratories, therefore at most unidirectional signalling between the laboratories is allowed.}, a constraint that can be formulated using mutual information as follows:
\begin{equation}
\label{one-con}
I(x:b|b'=0)+I(y:a|b'=1)\leq 1.
\end{equation}However, this condition is not sufficient for limiting correlations to the ones allowed by the process matrix framework because
\begin{align}
I(x:b|b'=0)&=\frac{1+E_1}{2}\log_2(1+E_1)+\frac{1-E_1}{2}\log_2(1-E_1)\\
I(y:a|b'=1)&=\frac{1+E_2}{2}\log_2(1+E_2)+\frac{1-E_2}{2}\log_2(1-E_2),
\end{align}and one can show that there are supra-quantum correlations with $E_1^2+E_2^2>1$ that verify \eqref{one-con}. Consequently, stronger constraints are needed.

\begin{Pro}
\label{prop-box}
Consider two rounds of the OCB game $(E_1^{(1)},E_2^{(2)},x_1,y_1,a_1,b_1,b'_1)$ and $(E_1^{(2)},E_2^{(2)},x_2,y_2,a_2,b_2,b'_2)$ where:
\begin{equation}
\label{assumption}
p(b_i\oplus x_i|b'_i=0)=\frac{1+E_1^{(i)}}{2},\quad p(a_i \oplus y_i|b'_i=1)= \frac{1+E_2^{(i)}}{2},\quad i=1,2.
\end{equation}
The following two conditions are equivalent:
\begin{description}
  \item[(i)] \begin{equation}
\label{hypo1}
\begin{aligned}
I(x_1:b_1|b'_1=0) &\geq I(x_1\oplus x_2:b_1\oplus b_2|b'_1=0,b'_2=0)\\
&+I(x_1\oplus y_2:b_1\oplus a_2|b'_1=0, b'_2=1),
\end{aligned}
\end{equation}
\begin{equation}
\label{hypo2}
\begin{aligned}
I(y_1:a_1|b'_1=1) &\geq I(y_1\oplus x_2:a_1\oplus b_2|b'_1=1,b'_2=0)\\
&+I(y_1\oplus y_2:a_1\oplus a_2|b'_1=1, b'_2=1),
\end{aligned}
\end{equation}
\item[(ii)]
\begin{equation}\label{hypo3}
(E_1^{(2)})^2+(E_2^{(2)})^2 \leq 1
\end{equation}
\end{description}
\end{Pro}
\begin{proof}
Suppose that $(E_1^{(2)})^2+(E_2^{(2)})^2 \leq 1$ holds. Define the variables
\begin{align}
\begin{aligned}
 X& =b_1|[b'_1=0], \\
 Y&=x_1|[b'_1=0],\\
  Z&=x_1\oplus x_2 \oplus b_2|[b'_1=0,b'_2=0],
 \end{aligned}
 \end{align} where the entire expression on the left-hand side of the bar is conditioned by the expression between brackets on the right-hand side. Because the two rounds are assumed to be independent, one can see that
\begin{equation}
X\rightarrow Y \rightarrow Z
\end{equation}
 is a Markov chain with transition parameters $p_1=\frac{1+E_1^{(1)}}{2}$ and $p_2=\frac{1+E_1^{(2)}}{2}$, therefore a strong form of the data processing inequality applies \cite{Erkip98theefficiency}:
\begin{equation}
\label{mrs}
I(X:Z)\leq \rho^*(Y:Z)^2 I(X:Y),
\end{equation}
where $\rho^*(Y:Z)$ defines the Hirschfeld-Gebelein-Rényi maximal correlation of variables $Y$ and $Z$ \cite{hirschfeld_connection_1935,Gebelein,renyi_new_1959,renyi_measures_1959}. Since $Y,Z$ are Bernoulli variables, we have $\rho^*(Y:Z)=2p_2-1=E_1^{(2)}$, therefore
\begin{equation}
 I(x_1\oplus x_2:b_1\oplus b_2|b'_1=0,b'_2=0) \leq (E_1^{(2)})^2 I(x_1:b_1|b'_1=0).
 \end{equation}
Similarly, one can show that:
\begin{equation}
 I(x_1\oplus y_2:b_1\oplus a_2|b'_1=0, b'_2=1)\leq (E_2^{(2)})^2  I(x_1:b_1|b'_1=0).
 \end{equation}
Therefore imposing $(E_1^{(2)})^2+(E_2^{(2)})^2 \leq 1$  implies \eqref{hypo1}. One can similarly show that $(E_1^{(2)})^2+(E_2^{(2)})^2 \leq 1$ also implies \eqref{hypo2}.

To prove the converse, we recall that since are $Y,Z$ Bernoulli variables, we have \cite{anantharam_maximal_2013}:
\begin{equation}
\rho^*(Y:Z)^2=\underset{X\rightarrow Y \rightarrow Z}{\sup} \frac{I(X:Z)}{I(X:Y)}.
\end{equation}Using \eqref{mrs}, one can show that \eqref{hypo1} and \eqref{hypo2} imply \eqref{hypo3}.
\end{proof}

If and only if a causal order is fixed, equations \eqref{hypo1} and \eqref{hypo2} take the form of the usual data processing inequality. In general, however, these equations involve sums of variables from two possible causal orders for a single round, while the data processing inequality requires that information be discarded in a fixed direction. Consequently, this alternative approach leads to two original conditions but their significance is blurred by their complexity. Keeping with our initial intuition, we now reformulate the causal game as a distributed random access code, and retrieve the quantum bound on correlations with indefinite causal order from an informational principle somewhat analogous to the principle of information causality.

\subsection{Causal games as random access codes}

In this section, we reformulate the causal game as a distributed random access code (RAC). This is motivated by the RAC formulation of the game for which the information causality principle was introduced \cite{pawlowski_information_2009}. Such an approach might open the path for a formulation of an analogue of the information causality principle in the context of causal games.

\subsubsection{Reformulation of the causal game}

Consider two runs of the experiment described in the causal game in Section \ref{section-3-2-1}, with bits $\{x_1,a_1,y_1,b_1\}$ and $\{x_2,a_2,y_2,b_2\}$ respectively. The random task bit $b'$ now corresponds to a pair of bits $b'_1b'_2$ denoting the four possible combinations of tasks for two runs of the experiment: $b'=0_1 0_2$ means that in both runs Alice must guess Bob's bit, $b'=0_1 1_2$ means that Alice must guess Bob's bit in the first run and Bob must guess Alice's bit in the second run, and so forth. It is straightforward to generalize this notation for $n$ runs.

Assume that different runs of the experiment use the same box as a resource:
\begin{equation}
p(b_i\oplus x_i|b'_i=0)= p(b_j\oplus x_j|b'_j=0), \quad \forall i,j.
\end{equation}Again, we write:
\begin{equation}
\label{assumption}
p(b_i\oplus x_i|b'_i=0)=\frac{1+E_1}{2},\quad p(a_i \oplus y_i|b'_i=1)= \frac{1+E_2}{2},\quad \forall i.
\end{equation}

Now consider $n$ runs of the experiment and define:
\begin{equation}
\label{optimize-n}
\begin{aligned}
P_n&=\frac{1}{2^n}\left[p(b_1\oplus x_1 \oplus..\oplus b_n \oplus x_n=0|b'=0_10_2..0_{n})\right.\\
+&\left. p(b_1\oplus x_1 \oplus..\oplus b_{n-1} \oplus x_{n-1} \oplus a_n \oplus y_n=0|\right.\\
& \left. \hspace{5cm} b'=0_1..0_{n-1}1_{n})\right.\\
&\left. +...+p(a_1\oplus y_1 \oplus..\oplus a_n \oplus x_n=0|b'=1_11_2..1_{n})\right].\\
\end{aligned}
\end{equation}For each term inside the brackets, the condition that the sum over the guesses for $n$ runs means that either both Alice and Bob make an even number of mistakes or both make an odd number of mistakes. We now compute the expression of a term $p_{n-k,k}$ inside the brackets for which the number of 1's in $b'$ is $k$.

The probability of an even number of wrong guesses by Alice is:
\begin{equation}
Q^{(k)}_{even}(\mbox{Alice})= \sum\limits_{j=1}^{\lfloor \frac{n-k}{2}\rfloor}\binom{n-k}{2j} \left(\frac{1-E_1}{2}\right)^{2j}\left(\frac{1+E_1}{2}\right)^{n-k-2j}=\frac{1+E_1^{n-k}}{2}
\end{equation}
Similarly, the probability for an odd number of wrong guesses by Alice is:
\begin{equation}
Q^{(k)}_{odd}(\mbox{Alice})= \sum\limits_{j=1}^{\lfloor\frac{n- k-1}{2}\rfloor}\binom{n-k}{2j+1} \left(\frac{1-E_1}{2}\right)^{2j+1}\left(\frac{1+E_1}{2}\right)^{n-k-2j-1}=\frac{1-E_1^{n-k}}{2}.
\end{equation}
The probability of an even number of wrong guesses by Bob is:
\begin{equation}
Q^{(k)}_{even}(\mbox{Bob})= \sum\limits_{j=1}^{\lfloor \frac{k}{2}\rfloor}\binom{k}{2j} \left(\frac{1-E_2}{2}\right)^{2j}\left(\frac{1+E_2}{2}\right)^{k-2j}=\frac{1+E_2^k}{2}
\end{equation}Similarly, the probability for an odd number of wrong guesses by Bob is:
\begin{equation}
Q^{(k)}_{odd}(\mbox{Bob})= \sum\limits_{j=1}^{\lfloor\frac{ k-1}{2}\rfloor}\binom{k}{2j+1} \left(\frac{1-E_2}{2}\right)^{2j+1}\left(\frac{1+E_2}{2}\right)^{k-2j-1}=\frac{1-E_2^k}{2},
\end{equation}
The final expression for a term inside the brackets where the number of 1's in $b'$ is $k$ is:
\begin{align}
\begin{aligned}
p_{n-k,k}&=Q^{(n-k)}_{even}(\mbox{Alice})\cdot Q^{(k)}_{even}(\mbox{Bob})+Q^{(n-k)}_{odd}(\mbox{Alice})\cdot Q^{(k)}_{odd}(\mbox{Bob})\\
&= \frac{1}{2}\left[1+E_1^{n-k}E_2^k\right],
\end{aligned}
\end{align}and
\begin{equation}
P_n=\frac{1}{2^n}\sum_{k=0}^{2^n-1}\binom{n}{k} p_{n-k,k}.
\end{equation}

\subsubsection{Bounds on the protocol efficiency}

We now treat the two bits in $b'$ as binary notation of an integer and identify $b^\prime$ with this integer. For example, when $n=2$, $b'=01$ corresponds to 1 and $b^\prime = 10$ to 2.
For a given decimal $b^\prime=i$, we group the rounds by specifying an expression to be set to 0, which we denote by $g_i \oplus t_i = 0$, where $g_i$ is the sum of output bits (`guesses') and $t_i$ the sum of input bits (`tosses'). To continue the example $n=2$, for $b'=1$ we set $x_1\oplus b_1 \oplus y_2\oplus a_2=0$ with the bit of guesses $g_1=x_1\oplus y_2$ and the bit of tosses $t_1=b_1\oplus a_2$. For $b'=2$ the corresponding expression is $y_1\oplus a_1 \oplus x_2\oplus b_2=0$ with the bit of guesses $g_2=y_1\oplus x_2$ and the bit of tosses $t_2=a_1\oplus b_2$.

\begin{Lem}
\label{lem1}
The following inequality holds:
\begin{equation}
\sum_{i=0}^{2^n-1} h(P(g_i \oplus t_i = 0|b'=i)) \geq 2^n - I(n),\label{littlelemma}
\end{equation}where $I(n)=\sum_{i=0}^{2^n-1} I(g_i:t_i|b'=i)$ is a measure of efficiency of the $n$ runs protocol, $I(X:Y)$ denotes mutual information between random variables $X$ and $Y$, and $h$ is the binary entropy.
\end{Lem}
\begin{proof}
We have:
\begin{align}
\begin{aligned}
I(n)&=\sum_{i=0}^{2^n-1} I(g_i:t_i|b'=i) \\
&= \sum_{i=0}^{2^n-1} H(g_i|b'=i) + H(t_i|b'=i)- H(g_i,t_i|b'=i),
\end{aligned}
\end{align}where $H$ is the Shannon entropy. Moreover:
\begin{align}
 H(g_i|b'=i) - H(g_i,t_i|b'=i) = - H(t_i|g_i,b'=i),
 \end{align}
 and
\begin{align}
\begin{aligned}
H(t_i|g_i,b'=i) &= H(t_i\oplus g_i|g_i,b'=i) \\
& \leq H(g_i\oplus t_i|b'=i) \\
&= h(P(g_i\oplus t_i=0|b'=i)).
\end{aligned}
\end{align}
It follows that $I(g_i:t_i|b'=i)\geq H(t_i|b'=i) - h(P(g_i\oplus t_i=0|b'=i))$, hence \eqref{littlelemma}.\end{proof}

\begin{Thm}The following inequality holds:
\begin{equation}
\label{cond1}
\frac{(E_1^2+E_2^2)^{n}}{2\ln(2)} \leq I(n)  \leq (E_1^2+E_2^2)^n.
\end{equation}
\end{Thm}
\begin{proof}
Using the lemma we have:
\begin{align}
\begin{aligned}
I(n) &\geq  \sum_{i=0}^{2^n-1} \left[1 - h(P(g_i \oplus t_i = 0|b'=i)\right]\\
& =\sum_{k=0}^{n} \binom{n}{k}\left[1- h(p_{n-k,k})\right]\\
&=\frac{1}{2\ln 2 }\sum_{k=0}^{n} \binom{n}{k}(E_1^2)^{n-k}(E_2^2)^k\\
&=\frac{1}{2\ln 2}(E_1^2+E_2^2)^n
\end{aligned}
\end{align}where we used $h(\frac{1}{2}(1+y))\leq 1-\frac{y^2}{2 \ln 2}$.
We also have:
\begin{equation}
\label{form}
\begin{aligned}
I(g_i:t_i|b'=i)&=\frac{1+E_1^{n-k}E_2^k}{2}\log_2 (1+E_1^{n-k}E_2^k) \\
&+ \frac{1-E_1^{n-k}E_2^k}{2}\log_2 (1-E_1^{n-k}E_2^k)\\
&\leq (E_1^2)^{n-k}(E_2^2)^k,
\end{aligned}
\end{equation}therefore
\begin{equation}
I(n) \leq \sum_{k=0}^n \binom{n}{k} (E_1^2)^{n-k}(E_2^2)^k = (E_1^2+E_2^2)^n.
\end{equation}
\end{proof}
Any causally separable process verifies:
\begin{equation}
\label{caus-ineq}
I(n)\leq 1, \quad \forall n,
\end{equation}and 1 is the only nonzero bound on $I(n)$. To see this, consider a fixed causal structure and a given value $b' = i$. Then all $g_k \oplus t_k, k \neq i$, are equal to 0 with probability $\frac{1}{2}$, therefore $I(g_i:t_i|b'=i)\leq 1$ and $I(g_k:t_k|b'=k)=0$ for $k\neq i$, leading to $I(n)\leq 1$. The mutual information expression $I(X:Y|Z)$ is convex in $p(y|x,z)$, where $x,y$ and $z$ are values that the random variables $X,Y$ and $Z$ can respectively take, therefore no mixture of strategies with fixed causal structures can increase the value of $I(n)$. Consequently, inequality \eqref{caus-ineq} is valid for all causally separable processes, i.e. we have found a class of causal games for which causally separable processes perform with bounded efficiency.

This bounded efficiency can now be taken as a constraint on the correlations between Alice's and Bob's laboratory. It alone suffices to derive the limit on quantum correlations with indefinite causal order: a limit on protocol efficiency for any number of runs is equivalent to the bound 1 on $E_1^2+E_2^2$, or equivalently to the bound $\frac{1}{\sqrt{2}}$ on $E$ where $E=E_1=E_2$ if all probabilities \eqref{assumption} are equal. Therefore, we derived the quantum bound on correlations with indefinite causal order by relaxing the constraints a fixed causal structure imposes on the signalling possibilities into an entropic form.

This result is somewhat analogous to the principle of information causality where, given a set of ``classical" resources (shared non-signalling correlations and one-way signalling) and a class of games (increasing size of Alice's data set), one can derive the quantum bound on correlations by keeping the same entropic figure of merit quantifiying the performance of the parties in winning such games for ``classical" and ``quantum" resources. Note that this similarity is only intuitive and not at all rigorous, because, in the context of no-signalling games, one can show that the principle of information causality is distinct from the `no-supersignalling' principle which encodes the idea that protocol efficiency must not increase \cite{wakakuwa_chain_2012}. Moreover, it should be emphasized that the derivation of the Tsirelson bound using information causality is based on natural properties of mutual information, while we derived the bound on the protocol efficiency---as measured by an expression using mutual information between inputs and outputs conditionned by the control bit---\emph{only for causally ordered correlations}, and related this bound---which is the only possible finite bound---to the quantum bound on the violation of the causal inequality. The main obstacle to a direct transposition of the proof of \cite{pawlowski_information_2009} to our game is the dependence between the guesses expressions `$g_i$' (and similarly for tosses expressions `$t_i$').

\subsection{The Hirschfeld-Gebelein-Rényi maximal correlation}

In the previous sections, we focused on mutual information as a measure of dependence to formulate entropic characterizations of the constraints imposed by (a mixture of) fixed causal structures on the signalling possibilities. We have shown the formal connection between two such constraints and the quantum bound on correlations with indefinite causal order.

Shifting the focus from mutual information to another measure of dependence, one can easily check that conditions \eqref{hypo1} and \eqref{hypo2} (or alternatively the quantum bound) are equivalent to imposing:
\begin{equation}
\label{causal-ineq}
\rho^*(Y:Z)^2+\rho^*(Y:Z')^2\leq 1,
\end{equation}where we kept the notation from the corresponding proof and defined
\begin{equation}
 Z'=x_1\oplus y_2\oplus a_2|[b'_1=0,b'_2]=1.
\end{equation}
 More generally, the quantum bound is equivalent to the following constraint:
\begin{equation}
\label{causal-ineq2}
\rho^*(x|b'=0:b|b'=0)^2+\rho^*(y|b'=1:a|b'=1)^2\leq 1,
\end{equation}while causally separable processes are characterized by:
\begin{equation}
\label{causal-ineq3}
\rho^*(x|b'=0:b|b'=0)+\rho^*(y|b'=1:a|b'=1)\leq 1.
\end{equation}
Since the Hirschfeld-Gebelein-Rényi (HGR) maximal correlation is also a measure of dependence, equation \eqref{causal-ineq3} has the same clear informational interpretation in terms of allowed signalling directions between parties within (a mixture of) fixed causal structures as equation \eqref{one-con}. The square of the HGR maximal correlation of Bernoulli variables, which appears in \eqref{causal-ineq2}, also has an information-theoretic interpretation: it quantifies the initial efficiency of communication between parties \cite{Erkip98theefficiency}. Indeed, taking $Y=x|[b'=0]$ and $Z=b|[b'=0]$ we obtain:
\begin{equation}
\label{connect-HGR-MI}
\rho^*(Y:Z)^2=\Delta^\prime(0),
\end{equation}where $\Delta^\prime$ is the derivative of
\begin{equation}
\Delta(R)= \sup\limits_{\substack{X \rightarrow Y \rightarrow Z\\I(X:Y)\leq R}} I(X:Z).
\end{equation}
Thus condition \eqref{causal-ineq2} means that the dependence between parties can exceed one bit as long as total initial efficiency of communication is does not exceed one bit.

In summary, equality \eqref{connect-HGR-MI} connects the HGR maximal correlation and the increase in mutual information. It is based on inequality \eqref{mrs} and is central to an information-theoretic interpretation of condition \eqref{causal-ineq2}. To prove \eqref{cond1} or \eqref{mrs}, one uses the standard properties of symmetry, non-negativity, chain rule and data processing of mutual information. Therefore, the bound on quantum correlations with indefinite causal order is equivalent to imposing these standard properties on mutual information between inputs and outputs of parties with an additional consistency condition for classical systems, so that mutual information between independent systems equal 0, along with one of conditions \eqref{caus-ineq} or \eqref{causal-ineq2}.

\section{Beyond causal quantum computation}

Indefinite causal order was also considered in the context of quantum circuits theory, in which information is described as a quantum state that evolves in time under a sequence of quantum gates \cite{chiribella_theoretical_2009}. Since a computation can be understood abstractly as a transformation of an input into an output---which need not be quantum states---one can go beyond processing of quantum states \cite{chiribella_theoretical_2009,chiribella_quantum_2013}. This type of computations on black boxes is known as higher-order quantum computation, the simplest example being quantum supermaps, viz. deterministic transformations with inputs and outputs corresponding to quantum operations \cite{chiribella_normal_2010}.

The importance of higher-order computation was established because it accounts for situations which cannot be simulated using quantum states. Indeed, a physically allowed task called `quantum switch', where a pair of input black boxes $A$ and $B$ are connected in two different orders $B \rightarrow A$ vs. $A\rightarrow B$ conditionally on the value of an input bit, cannot be realized with causally ordered circuits unless causality is violated. Such a task can be realized in the laboratory using quantum circuits where the geometry of the connections can be entangled with the state of a control qubit, and it can outperform the standard causally ordered quantum computers in specific tasks, such as discriminating between two non-signalling channels \cite{chiribella_perfect_2012}.

We now consider two examples to better grasp the new computational features of the quantum switch. The first example is a simplified model of the discrimination between two non-signalling channels task \cite{brukner_quantum_2014}. One must distinguish whether a pair of boxes---which represent two unitaries $\hat A$ and $\hat B$---commute or anticommute, i.e. whether $\hat A \hat B =\pm \hat B \hat A$. Solving this task using only causally ordered circuits requires at least one of the unitaries to be applied twice, while a simple algorithm exploiting superpositions of causal circuits can solve the problem with only one use of each of the two boxes: one coherently applies the two unitaries on the initial state $|\psi\rang$ of the computer in two possible orders, depending on the state of a control qubit. If the control qubit is prepared in the superposition $\frac{1}{\sqrt{2}}(|0\rang +|1\rang)$, the output of the algorithm is:
\begin{equation}
\frac{1}{\sqrt{2}}(\hat A\hat B|\psi\rang|0\rang)+\hat B\hat A|\psi\rang|1\rang)=\frac{1}{\sqrt{2}}\hat A\hat B |\psi\rang(|0\rang\pm|1\rang),
\end{equation}
where the phase of the control qubit state is $+1$ if the two unitaries commute and $–1$ if they anticommute. Measuring the control qubit in the basis $\frac{1}{\sqrt{2}}(|0\rang \pm |1\rang)$ finally solves the task. Thus, the number of queries to the oracle is reduced from 2 to 1 with respect to a causal computation (See Fig. \ref{non-causal-circuit}).
\begin{figure}
\begin{center}
\includegraphics[scale=0.27]{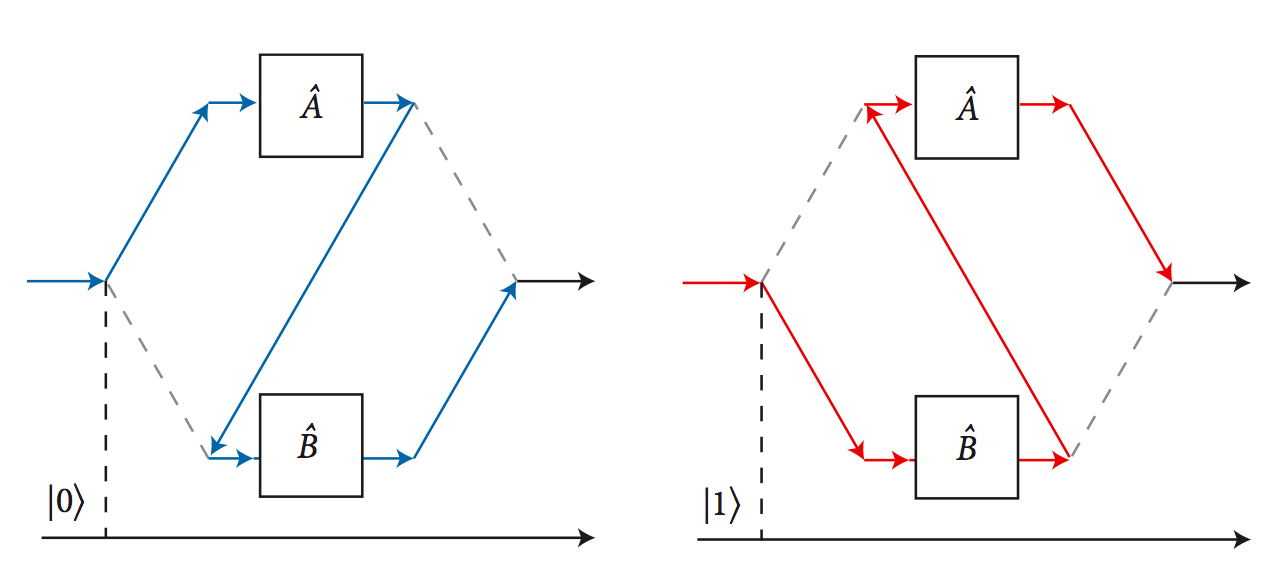}
\end{center}
\caption{The ordering of physical boxes $\hat A$ and $\hat B$ depends on the quantum state of a control bit. Until the state of the control bit is projected, the circuit is in a superposition of causal orders $A\preceq B$ (left) and $B \preceq A$ (right) (From \cite{chiribella_quantum_2013,brukner_quantum_2014}).}
\label{non-causal-circuit}
\end{figure}

The second task is a related computational problem of size $n$ \cite{araujo_computational_2014}. Let $\{U_k\}_{k=1}^{n}$ be a set of unitary matrices of dimension $d\geq n!$ and define
\begin{equation}
 \Pi_x=U_{\sigma_x(n)}\cdots U_{\sigma_x(1)},
\end{equation}
 where $\{\sigma_x\}_{x\in[\![1;n! ]\!]}$ is a set of permutations.
We say that the set of unitaries $\{U_k\}_{k=1}^{n}$ verifies property $P_x$ iff
$\Pi_x=e^{i\frac{2\pi}{n!}xy}\Pi_1$. The computational problem is defined as follows: given a set $\{U_k\}_{k=1}^{n}$ of unitary matrices of dimension $d\geq n!$, decide which of the properties $P_x$ is satisfied given the promise that at least one of these $n!$ properties is satisfied. Superpositions of quantum circuits solve this problem by using $O(n)$ black-box queries, whereas the best-known quantum algorithm with fixed order between the gates requires $O(n^2)$ queries. Therefore, even though the reduction is only polynomial, extending the quantum circuit model can provide a computational advantage.

In general, one can show that any superposition of quantum circuits can be simulated (with, at most, polynomial overhead) by a standard causal quantum circuit \cite{brukner_quantum_2014}. This implies that the quantum switch cannot be used to violate the causal inequality of Section \ref{section-3-2-2}.

\bigskip

In conclusion, we analyzed how processes and networks where the geometry of the wires between the gates are entangled with the state of a control qubit predict the existence of correlations with no fixed causal order. These frameworks can be thought of as toy models for the quantum gravity situation where a massive object is put in a spatial superposition, producing a gravitational field (and hence a metric) in a superposition of states. Thus understanding the nature of the correlations these frameworks reveal is crucial to topics in quantum foundations, quantum information and quantum gravity. As a first step in this program, we placed a subset of the quantum correlations arising from causally non-separable processes in the larger landscape of CPO-boxes (including correlations that maximally violate the causal inequality) which allowed a better grasp of the connections between the local ordering of events and the role of the control bit in the emergence of an indefinite global order. We also formulated various constraints on mutual information and the HGR maximal correlation that allow to retrieve the classical and quantum bounds on correlations with indefinite causal order.

\chapter*{Conclusions and Outlook}
\addcontentsline{toc}{chapter}{Conclusions and Outlook}
\markboth{Conclusions and Outlook}{}

The topics addressed in this dissertation involve the notions of quantum correlations, causal structures and information.

In Chapter \ref{first-chapter}, we reviewed some basic concepts of quantum theory. We first clarified the meaning of the word ``nonlocality" and its relation to entanglement by using the LOCC and LOSR paradigms, and then analyzed the deep implications the notion of nonlocality has on our understanding of ``reality". The PR-box framework placed quantum correlations between spacelike separated systems in the more general landscape of non-signalling correlations, thus clarifying the fact that nonlocal quantum correlations respect causality. We also reviewed the partial reconstruction of bipartite quantum correlations based on information causality. Whether this principle can account for the whole set of bipartite quantum correlations is still an open question. At any rate, the impossibility to describe all quantum correlations using bipartite principles calls for a new genuinely multipartite approach. Nonetheless, this and similar partial reconstructions revealed that many features that were thought to be purely quantum are in fact quite generic if one selects a suitable generalized probabilistic framework. We also reviewed the notion of entanglement in the AQFT framework, an approach that highlighted how much more entrenched entanglement is in QFT, namely \emph{precisely} when relativistic constraints are fully integrated into quantum theory.

In Chapter \ref{second-chapter}, we analyzed in what sense the standard causal structure imposed by relativity affects entanglement detection and quantification. We reviewed some recent results on the observer-dependent character of entanglement and some difficulties that arise when describing physical systems at the most fundamental level using QFT. We argued that the usual identification between regions of space and quantum observables produces divergences in the calculation of entanglement entropy. We detailed how one can build a Fock space for the field system where finite-energy states are localized and have a finite entanglement entropy by identifying regions of space through the NW position operators instead of the covariant fields. We then showed that under coarse graining, the localization and entanglement properties of the NW operators for finite-energy states effectively hold, irrespective of the choice of local observables at the fundamental level. This result suggests that different localization schemes should be adopted for regimes of different energy, and one can imagine, by arbitrarily increasing the resolution, to access regimes where the vacuum entanglement becomes more and more accessible while the divergent entropy associated with it is confined in the fine-grained degrees of freedom.

The above coarse-graining procedure encodes the idea of a loss of information at the effective ``everyday" level. This idea can be generalized to that of an  ``observer effective horizon", and the work by Jacobson we reviewed relating the ``equilibrium thermodynamics" of entanglement entropy at event horizons and Einstein's equation is an example of application of this ``observer effective horizon" approach, which highlights the fact that it might be fruitful for deeper topics involving quantum gravity \cite{lashkari_gravitational_2014}. A possible extension of our work is to formulate the same problem in the AQFT framework we presented in Chapter \ref{first-chapter}. Indeed, there we showed how by taking ``slightly larger" local algebras of observables, one can get rid of most entanglement effects entailed by the type III property. Thus, one could hope by using this algebraic approach for a more rigorous and general account of how coarse-graining confines the divergence of entanglement entropy to unobservable degrees of freedom.

In Chapter \ref{third-chapter}, we reviewed an operational framework that unifies the descriptions of spacelike and timelike separated correlations, a necessary step if we are to explore quantum correlations beyond definite causal order. This framework was based on the notion of process which generalizes that of quantum state to general situations where a number of local operations are performed with no prior assumption about how such operations are embedded in a global spacetime. The local operations correspond to the most general quantum operations that can be performed in a localized spacetime volume, viz. (trace non-increasing) completely positive maps. We reviewed how some processes can give rise to correlations that violate a causal inequality verified by all causally ordered classical and quantum correlations. These results aside, the general properties of processes are still largely unexplored, and in the same fashion that von Neumann entropy generalizes the Shannon entropy to quantum states, we still have to look for sensible notions that would generalize the notions of entropy, channel capacity, purification---the latter would first need a notion of measure of causal non-separability---etc. to processes.

Following the PR-box approach to non-signalling correlations, we introduced a generalized probabilistic framework based on a property we called ``compatibility with predefined causal order" (CPO) with the aim to highlight the importance of the control bit for the emergence of an indefinite global order, at least when certain constraints are placed on local operations as done implicitly in the original proof of the violation of the causal inequality. Work in progress seems to indicate that the violation of causal inequalities is possible without using a control bit \cite{private} if more general local operations are allowed. Therefore, the CPO condition we introduced might not be the truly fundamental analogue of the no-signalling principle. However, all attempts to this day to find a physically clear principle shared by causally separable and non-separables processes and using only measurement results $x,y$ and inputs $a,b$ with no restrictions on local operations have failed.

We also analyzed how various information-theoretic constraints can characterize the quantum bound on correlations with indefinite causal order. Unlike the information causality derivation for non-signalling correlations, we only \emph{derived} the bound on the protocol efficiency---as measured by an expression using mutual information between inputs and outputs conditionned by the control bit---for causally ordered correlations, and related this bound---which is the only possible nonzero bound---to the quantum bound on the violation of the causal inequality. Thus, a possible extension of our work would be to look for a general derivation of the bound on the protocol efficiency using only ``natural" properties of mutual information, or some other measure of dependence. This more general approach failed in the alternative game we introduced because the guesses expressions `$g_i$' were not independent (and similarly for tosses expressions `$t_i$'). This is not a fundamental limitation for finding an intuitive principle that can bound quantum correlations with indefinite causal order because, keeping with the logic of our initial analogy, one can show that information causality holds in the context of non-signalling correlations even when the bits in Alice's set are not independent \cite{pawlowski_hyperbits:_2012}. However, the derivation heavily relies on computations using the quantum formalism and not on simple manipulations of information-theoretic quantities.

Other open questions in this topic concern the experimental implementation of a violation of the causal inequality. We mentionned that the quantum switch cannot violate it, but one can show that protocols exist that use the quantum switch and correspond to causally non-separable processes \cite{private}. Since a physical implementation of quantum switches can be achieved by using an interferometric setup that implements a quantum control of the order between gates, at least a subset of causally non-separable processes can be implemented in the near future and offer new advantages for quantum computation. Consequently, it is reasonable to hope that future work will provide intuitions on possible implementations of causally non-separable processes that violate the causal inequality, a result that would have very deep implications on our understanding of causal structures.

\begin{appendices}
\renewcommand\chaptername{Appendix}

\chapter{C{$^*$}-algebras and von Neumann algebras}
\label{algebra-theory}

In this appendix, we review the connections between C$^*$-algebras and von Neumann algebras with topology and probability theory respectively. We adopt the structuralist point of view in order to emphasize the natural character of such connections.

All algebras are assumed to be unitary. We denote by $\sigma(A)$ the spectrum (which can be the empty set) of an element $A$ of an algebra $\mathfrak{A}$.

\section{Generalities}

\begin{Def}
A Banach algebra $(\mathfrak{A},+,\circ,||.||)$ over a field $\mathbb{K}$ is a normed associative $\mathbb{K}$-algebra such that the underlying normed vector space is a Banach space, i.e. a complete normed vector space.
\end{Def}

If $\mathfrak{A}$ and $\mathfrak{B}$ are two Banach algebras, we call $\sigma(\mathfrak{A},\mathfrak{B})$-topology the weakest topology on $\mathfrak{A}$ such that all the elements of $\mathfrak{B}$ are continuous.

\begin{Def}
\label{topologies}
Let $(\mathfrak{A},||.||)$ be a Banach space and $\mathfrak{A}^*$ its topological dual.
\begin{itemize}
\item[(i)] The norm topology is the natural topology generated by the open sets defined with norm $||.||$.
\item[(ii)] The weak Banach topology on $\mathfrak{A}$ is the $\sigma(\mathfrak{A}, \mathfrak{A}^*)$-topology.
\item[(iii)] If $\mathfrak{A}$ admits a predual $\mathfrak{A}_{*}$, the ultraweak topology, or weak $*$-topology on $\mathfrak{A}$ is the $\sigma(\mathfrak{A},\mathfrak{A}_{*})$-topology.
\end{itemize}
\end{Def}

\section{{C$^*$}-algebras}

We provide in this section some definitions and key results of the theory of C$^*$-algebras.

\subsection{General definitions}

\begin{Def}
A \emph{concrete} C$^*$-algebra is a sub-$*$-algebra of $\mathfrak{B}(\calH)$, where $\calH$ is a Hilbert space,  closed under the norm topology.
\end{Def}

\begin{Def}
An \emph{abstract} C$^*$-algebra $\mathfrak{A}$ is a complex algebra endowed with a norm $||.||$ and an antilinear involution $*:\mathfrak{A}\rightarrow \mathfrak{A}$ such that:
\begin{itemize}
\item[(i)] $(\mathfrak{A},||.||)$ is a Banach space.
\item[(ii)] $\forall A,B\in\mathfrak{A}, ||AB||\leq ||A||\cdot ||B||$.
\item[(iii)]  $\forall A \in \mathfrak{A},||A^*A||=||A||^2$ (C$^*$-condition).
\end{itemize}
The previous conditions imply that $||A||=||A^*||, \quad \forall A \in \mathfrak{A}$.
\end{Def}

\begin{Def}
A linear form $\omega:\mathfrak{A}\rightarrow\mathbb{C}$ on $\mathfrak{A}$ is called positive iff $\omega(A^*A)\geq 0, \forall A \in \mathfrak{A}$, normalized iff $\omega(\id)=1$. A positive normalized linear form is called a state, and the set of states is denoted by $\mathcal{S}(\mathfrak{A})$.
\end{Def}

\subsection{Commutative C{$^*$}-algebras}

\begin{Def}
A multiplicative linear form $m$ on a commutative Banach algebra $\mathfrak{A}$, or character on $\mathfrak{A}$, is a homomorphism from $\mathfrak{A}$ to $\mathbb{C}$, i.e. a map $\mathfrak{A}\rightarrow\mathbb{C}$ preserving algebraic relations:
\begin{equation}
m(AB)=m(A)m(B), \hspace{0.2cm} m(A+B)=m(A)+m(B), \hspace{0.1cm} \quad \forall A,B \in \mathfrak{A}.
\end{equation}
We denote by $\Sigma(\mathfrak{A})$ the set of characters of $\mathfrak{A}$. Nonzero elements of $\Sigma(\mathfrak{A})$ are normalized.
\end{Def}

\begin{Pro}
Let $\mathfrak{A}$ be a commutative Banach algebra. For all $A\in\mathfrak{A}$, $\lambda \in \sigma(A)$ iff there exists a character $m\in\Sigma(A)$ such that $m(A)=\lambda$, the correspondance being one-to-one. This justifies the denomination Gelfand spectrum for $\Sigma(\mathfrak{A})$.
\end{Pro}

\begin{Thm}[Commutative Gelfand-Naimark]
A unitary commutative abstract C$^*$-algebra $\mathfrak{A}$ is isometrically isomorphic to the concrete C$^*$-algebra of continuous functions on the Gelfand spectrum $\Sigma(\mathfrak{A})$ of $\mathfrak{A}$ which, endowed with the $\sigma(\mathfrak{A}^*,\mathfrak{A})$-topology, is a compact Hausdorff topological space.
\end{Thm}

We consider the following categories and functors:
\begin{description}
\item[$\Calg$] is the dual category to the category whose objects are unitary commutative C$^*$-algebras and morphisms are unit preserving continuous $*$-homomorphisms.
\item[$\Top$]\hspace{0.05cm}  is the category of compact Hausdorff topological spaces and continuous maps
\item[$\euC$] \hspace{0.3cm} is the functor sending a compact Hausdorff topological space onto the algebra of continuous functions defined on it and transforming a morphism $f:X\rightarrow Y$ into the morphism $\{\mathfrak{C}(Y) \ni g \mapsto g \circ f \in \mathfrak{C}(X)\}$.
\item[$\euS$]\hspace{0.3cm} is the functor sending a commutative C$^*$-algebra $\mathfrak{A}$ onto the set of characters $\Sigma(\mathfrak{A})$ equipped with the  $\sigma(\mathfrak{A}^*,\mathfrak{A})$-topology, and a $*$-homomophism $\pi$ onto $\{\Sigma(\mathfrak{A}) \ni m \mapsto m\circ \pi\}$.
\end{description}

\begin{Thm}[Gelfand's duality]
\begin{equation}
\Calg\mathrel{\substack{\euC\\\rightleftarrows\\ \euS}}\Top
\end{equation}
is an equivalence of categories.
\end{Thm}
This correspondance between algebraic and topological concepts, and the possibility of extending it to non-commutative C$^*$-algebras justifies the denomination \emph{non-commutative topology} associated to the study of general C$^*$-algebras.

\subsection{Gelfand-Naimark-Segal (GNS) construction}

\begin{Def}
A $*$-homomorphism between two unitary $*$-algebras $\mathfrak{A}$ and $\mathfrak{B}$ is a map $\pi:\mathfrak{A}\rightarrow\mathfrak{B}$ preserving algebraic relations and involution:
\begin{itemize}
\item[(i)] $\pi(\lambda A+\mu B)=\lambda \pi(A) + \mu \pi(B), \quad \forall A,B \in \mathfrak{A}, \forall \lambda,\mu \in \mathbb{C}.$
\item[(ii)] $\pi(A^*)=(\pi(A))^*,\quad \forall A\in\mathfrak{A}.$
\item[(iii)] $\pi(AB)=\pi(A)\pi(B), \quad \pi(\id_\mathfrak{A})=\id_\mathfrak{B}, \quad \forall A\in\mathfrak{A}, \forall B\in\mathfrak{B}.$
\end{itemize}If $\pi$ is bijective, it is called $*$-isomorphism, and if additionally  $\mathfrak{A}=\mathfrak{B}$ it is called $*$-automorphism.
\end{Def}

\begin{Def}
A representation $\pi$ of a C$^*$-algebra $\mathfrak{A}$ into a Hilbert space $\mathcal{H}$ is a $*$-homomorphism from $\mathfrak{A}$ into the C$^*$-algebra $\mathfrak{B}(\mathcal{H})$ of bounded linear operators on $\calH$. A representation is said to be:
\begin{itemize}
\item[(i)] Faithful iff $\ker(\pi)=\{0\}.$
\item[(ii)] Irreducible iff $\{0\}$ et $\mathcal{H}$ are the only invariant closed subspaces of $\pi(\mathfrak{A}).$
\end{itemize}
\end{Def}

\begin{Def}
Let $\pi:\mathfrak{A}\rightarrow\mathcal{H}$ be a representation. A vector $\Psi \in \calH$ is called cyclic iff $\pi(\mathfrak{A})\Psi$ is dense in $\calH$.
\end{Def}

\begin{Thm}[Gelfand-Naimark-Segal (GNS) construction]
\label{NC-GN}
Let $\mathfrak{A}$ be a unitary C$^*$-algebra and $\omega \in \mathcal{S}(\mathfrak{A})$ a state. Then there exists a Hilbert space $\mathcal{H_\omega}$ and a representation $\pi_\omega:\mathfrak{A}\rightarrow\mathcal{B(H_\omega)}$ such that:
\begin{itemize}
\item[(i)] $\mathcal{H_\omega}$ contains a cyclic vector $\Psi_{\omega}.$
\item[(ii)] $\omega(A)=\langle \Psi_{\omega},\pi_\omega(A)\Psi_{\omega}\rangle,\quad \forall A \in \mathfrak{A}.$
\item[(iii)] Any representation $\pi$ in a Hilbert space $\mathcal{H_\pi}$ with cyclic vector $\Psi$ such that:
\begin{equation}
\label{scal-prod-GNS}
\omega(A)=\langle\Psi,\pi(A)\Psi\rangle, \quad \forall A \in \mathfrak{A},
\end{equation}
is unitarily equivalent to representation $\pi_\omega$, i.e. there exists an isometry $U:\mathcal{H_\pi}\rightarrow\mathcal{H_\omega}$ such that:
\begin{equation}
U\Psi=\Psi_\omega,
\end{equation}
\begin{equation}
U\pi(A)U^{-1}=\pi_\omega(A),\quad  \forall A \in \mathfrak{A}.
\end{equation}
\end{itemize}
This is called the GNS representation of $\mathfrak{A}$ defined by state $\omega$.
\end{Thm}

\begin{Def}
The universal representation $\pi_u$ of a C$^*$-algebra $\mathfrak{A}$ is defined as the direct sum of all GNS representations $\{\pi_{\omega}\}_{\omega \in \mathcal{S}(\mathfrak{A})}$. The corresponding Hilbert space is:
\begin{equation}
\mathcal{H}_u=\underset{\omega \in \mathcal{S}(\mathfrak{A})}{\oplus}\mathcal{H}_{\omega}.
\end{equation}
\end{Def}

Thus, the GNS construction shows that a representation of a C$^*$-algebra by operators acting on a Hilbert space is always possible, given that states exist\footnote{This can be shown by expressing polynomials on the spectrum of normal elements as evaluations of states on polynomials of normal elements. The spectral theorem implies that the spectrum of normal elements is not empty, hence states exist.}. We start to see a sketch of the usual formalism of quantum theory, however representations so far are not faithful. The following theorem states that such faithful representations exist.

\begin{Thm}[Gelfand-Naimark]
\label{GNTNC}
Any C$^*$-algebra $\mathfrak{A}$ is isometrically isomorphic to a self-adjoint sub-$*$-algebra of $\mathfrak{B}(\cal{H})$, where $\calH$ is a Hilbert space, closed for the strong topology.
\end{Thm}

It is also possible to build faithful representations via GNS from the so-called faithful states, but their existence is not assured:

\begin{Def}
A state $\omega$ on $\mathfrak{A}$ is said to be faithful iff $\omega(A^*A)>0,\forall A \in \mathfrak{A}, A \ne 0$.
\end{Def}

\section{Von Neumann algebras}

\subsection{General definitions}

\begin{Def}
A concrete von Neumann algebra is a unitary sub-$*$-algebra of $\mathfrak{B}(\calH)$, where $\calH$ is a Hilbert space, closed under the weak Banach topology.
\end{Def}

Since all topologies defined in \ref{topologies} are locally convex, they can also be defined using a family of seminorms if we are dealing with the algebra $\mathfrak{B}(\calH)$ of bounded operators on some Hilbert space $\mathcal{H}$. Therefore, one can define the ultraweak topology on any sub-$*$-algebra of $\mathfrak{B}(\calH)$ using a seminorm and with no reference to the existence of a predual. Concrete von Neumann algebras can then also be defined as unitary sub-$*$-algebras of $\mathfrak{B}(\calH)$ closed under the ultraweak topology. As we shall see, consistency with the former definition of ultraweak topology is preserved by showing that all concrete von Neumann algebras admit a predual, which motivates the following definition:

\begin{Def}
An abstract von Neumann algebra , or W$^*$-algebra, is a C$^*$-algebra that admits a predual.
\end{Def}

\begin{Def}
Let $\mathfrak{M}_1$ and $\mathfrak{M}_2$ be two W$^*$-algebras and $\Phi:\mathfrak{M}_1\rightarrow \mathfrak{M}_2$ a $*$-homomorphism. $\Phi$ is a W$^*$-homomorphism iff it is a continuous map between $\mathfrak{M}_1$ and $\mathfrak{M}_2$ equipped with their ultraweak topologies.
\end{Def}

\begin{Def}
Let $\mathfrak{A}$ be a C$^*$-algebra. The universal enveloping von Neumann algebra of $\mathfrak{A}$ is the closure under the weak Banach topology of the universal representation $\pi_u$ of $\mathfrak{A}$.
\end{Def}

\subsection{Commutative von Neumann algebras}

\begin{Def}
A compact Hausdorff topological space $X$ is called hyperstonian iff there exists a Banach algebra $\mathfrak{B}$ such that the algebra $\mathfrak{C}_\mathbb{C}(X)$ of continuous functions on $X$ is isometric to the dual $\mathfrak{B}^*$ of $\mathfrak{B}$.
\end{Def}

\begin{Def}
Let $(X,M)$ be a measurable space, i.e. a set $X$ equipped with a distinguished $\sigma$-algebra $M$ of subsets called measurables subsets of $X$. A subset $N \subset M$ is a $\sigma$-ideal iff:
\begin{itemize}
\item[(i)] $ \emptyset \in N.$
\item[(ii)] For all $A \in N$ et $B\in M$, $B \subset A \Rightarrow B\in N.$
\item[(iii)] $\{A_n\}_{n\in\mathbb{N}}  \subset N \Rightarrow \bigcup_{n\in\mathbb{N}} A_n \in N.$
\end{itemize}
Once $(X,M)$ is equipped with a measure $\mu$, we call $(X,M,\mu)$ a measure space.
\end{Def}

\begin{Def}Define:
\begin{itemize}
\item[(i)] A measured space is a triplet $(X,M,N)$ where $X$ is a set, $M$ is the $\sigma$-algebra of measurables subsets of $X$ and $N\subset L$ a $\sigma$-ideal of sets with measure zero.
\item[(ii)] A measured space is called localized iff the boolean algebra $M/N$ of equivalence classes of measurable sets is complete.
\end{itemize}
\end{Def}

This distinction between the concepts of measurable, measured and measure space is crucial to understanding the connexions between the theory of commutative von Neumann algebras and measure theory. As we shall see, commutative von Neumann algebras are connected to measured spaces rather than measure spaces. The choice of a state on a von Neumann algebra will correspond to the choice of a measure on the corresponding measured space, and the localization condition points to a connexion to integration theory rather than to general measure spaces. In the words of Segal \cite{segal_equivalences_1951}:

\begin{quote}
``The class of measure spaces with these properties (we call such spaces `localizable') constitutes in some ways a more natural generalization of the $\sigma$-finite measure spaces, than the class of arbitrary measure spaces. In particular, for a measure space to be localizable is equivalent to the validity for the space of the conclusion of the Radon-Nikodym theorem, or alternatively to the conclusion of the Riesz representation theorem for continuous linear functionals on the Banach space of integrable functions. Every measure space is metrically equivalent (by which we mean there is a measure-preserving isomorphism between the $\sigma$-finite measure rings, roughly speaking this means the spaces are equivalent as far as integration over them is concerned) to a localizable space, and this latter space is essentially unique."
\end{quote}

The $\sigma$-finite condition Segal is refering to is essential to build a rich integration theory. Indeed important concepts/results such as the notion of product measure or the Fubini theorems make crucial use of the notion of $\sigma$-finite measures.

We now consider the following categories:
\begin{description}
\item[$\vN$] is the dual category of the category of unitary commutative concrete von Neumann algebras and unit preserving W$^*$-homomorphisms.
\item[$\Hys$]\hspace{0.01cm} is the category of hyperstonian spaces and open continuous functions.
\item[$\Lms$] \hspace{0.005cm} is the category of localized measured spaces whose morphisms correspond to maps $(X,M,N)\mapsto (Y,P,Q)$ such that the preimage of every element of $P$ is a union of an element of $M$ and a subset of an element of $N$ and the preimage of every element of $Q$ is a subset of an element of $N$.
\end{description}

The following theorem is the restriction of Gelfand's duality to von Neumann algebras and their morphisms:
\begin{Thm}
Categories $\vN, \Hys$ and $\Lms$ are equivalent.
\end{Thm}

A few comments on this result are useful. First, the equivalence between the algebra of hyperstonian maps and unitary commutative concrete von Neumann algebras shows that in the commutative case, we indeed have an equivalence between closure under the weak Banach (or ultraweak) topology and existence of a predual, i.e. between concrete and abstract von Neumann algebras. Therefore, we will refer to such algebras simply as commutative von Neumann algebras. Second, the equivalence between commutative von Neumann algebras and measurable localized spaces, and the possibility of extending such a correspondence to the non-commutative case justifies the denomination \emph{non-commutative probability/measure theory} given to the study of non-commutative von Neumann algebras. Note that the the category of hyperstonian spaces and their morphisms is not a sub-category of the category of compact Hausdorff topological spaces and their morphisms, a hint to the fact that though von Neumann algebras are C$^*$-algebras they define a distinct theory.

\subsection{Representations}
 The following theorem is a version of the Gelfand-Naimark theorem for W$^*$-algebras:

\begin{Thm}[Sakai's theorem]
Let $\mathfrak{M}$ be a W$^*$-algebra. Then there exists a faithful W$^*$-representation $(\pi,\calH)$ of $\mathfrak{M}$, i.e. $\mathfrak{M}$ is $*$-isomorphic to a self-adjoint sub-$*$-algebra of $\mathfrak{B}(\calH)$, where $\calH$ is a Hilbert space, closed under the ultraweak topology.
\end{Thm}

We will now identify abstract von Neumann algebras with concrete ones, and refer to them simply as von Neumann algebras.

\begin{Def}
Let $\mathfrak{M}=\mathfrak{B}(\calH)$ be a von Neumann algebra where $\calH$ a Hilbert space, and $\omega$ a state on $\mathfrak{M}$. $\omega$ is said to be normal iff there is a positive trace-class operator $D_\omega$ on $\calH$ verifying $\trace(D_\omega)=1$ and such that $\omega(A)=\trace(D_\omega  A)$ for all $A\in\mathfrak{B}(\calH)=\mathfrak{M}$.
\end{Def}

Remarkably enough, von Neumann algebras can be characterized both topologically and algebraically:

\begin{Def}
The commutant of an arbitrary subset $\mathfrak{A} \subset \mathfrak{B}(\calH)$ is a subset $\mathfrak{A}' \subset \mathfrak{B}(\calH)$ such that:
\begin{equation}
B \in \mathfrak{A}' \Leftrightarrow \forall A \in \mathfrak{A}, [B,A]=0.
\end{equation}
\end{Def}

\begin{Thm}[Von Neumann's bicommutant theorem]
Let $\mathfrak{M}$ be a sub-$*$-algebra of $\mathfrak{B}(\calH)$. The following conditions are equivalent:
\begin{itemize}
\item[(i)] $\mathfrak{M}$ is closed for the weak Banach (or ultraweak) topology.
\item[(ii)] $\mathfrak{M}=\mathfrak{M}''$.
\end{itemize}
\end{Thm}

\begin{Cor}
\label{cor-bicommutant}
The vector space spanned by the set $\mathfrak{P(M)}$ of projections of a von Neumann algebra is norm dense in $\mathfrak{M}$
\end{Cor}

The equivalence between topological and algebraic constraints suggests a strong conceptual rigidity, which naturally leads to the question of classification of von Neumann algebras.

\subsection{Classification of factors}

\begin{Def}
A von Neumann algebra $\mathfrak{M}$ is called a factor iff $\mathfrak{M}\cap \mathfrak{M}'=\mathbb{C} \id$, i.e. the center $\mathfrak{Z(M)}$ of $\mathfrak{M}$ is reduced to multiples of identity.
\end{Def}

As we shall see, a theorem shows that one only needs to classify factors.

\begin{Def}
A measure space $(X,M,\mu)$ is said to be complete iff $S\subset N\in M$ and $\mu(N)=0\Rightarrow S\in M$.
\end{Def}

\begin{Def}
Let $(X,M,\mu)$ be a complete measure space and $\{\calH_x\}_{x\in X}$ a set of Hilbert spaces indexed by $X$. The direct integral of the Hilbert spaces $\{\calH_x\}_{x\in X}$ with respect to (w.r.t.) $\mu$ is a Hilbert space $\calH$ such that there exists a map $x\mapsto v(x)\in \calH_x$ verifying:
\begin{itemize}
\item[(i)]For all elements $u,v \in \calH$, the map $x\mapsto \lang u(x),v(x)\rang$ is measurable and integrable w.r.t. $\mu$, and:
\begin{equation}
\lang u,v\rang=\int_{x\in X} \lang u(x),v(x)\rang d\mu(x).
\end{equation}
\item[(ii)]For all $w_x\in\calH_x$ with $x\in X$ and $v\in \calH$, the map $x\mapsto \lang w_x,v(x)\rang$ is measurable and integrable with respect to $\mu$, and there exists $w\in\calH$ such that $w(x)=w_x$ for almost all $x$ in $X$.
\end{itemize}
We then denote:
\begin{equation}
\calH = \int_{x\in X} \calH_x d\mu(x) \hspace{0.5cm} \mbox{and} \hspace{0.5cm} v=\int_{x\in X} v(x)d\mu(x).
\end{equation}
\end{Def}

\begin{Def}
Let $\calH$ be a direct integral of $\{\calH_x\}_{x\in X}$ where $(X,M,\mu)$ is a complete measure space. An operator $T$ on $\calH$ is called decomposable (w.r.t. the integral decomposition associated to $\calH$) iff:
\begin{itemize}
\item[(i)] There exists a map $x\mapsto T(x)$ defined on $X$ such that $T(x)\in\mathfrak{B}(\calH_x)$ for all $x \in X$.
\item[(ii)] For all $v\in \calH$, $T(x)v(x)=(Tv)(x)$ for almost all $x \in X$.
\end{itemize}
If all operators are decomposable, we denote $\mathfrak{B}(\calH) = \int_{x\in X} \mathfrak{B}(\calH_x) d\mu(x)$.
\end{Def}

\begin{Thm}
Let $\mathfrak{M}$ be a von Neumann algebra. There exists a complete measure space $(X,M,\mu)$ with $\mu$ a $\sigma$-finite measure, $\{\calH_x\}_{x\in X}$ a set of Hilbert spaces and $\{\mathfrak{M}\}_{x\in X}$ a corresponding set of von Neumann algebras such that:
\begin{itemize}
\item[(i)] $\mathfrak{M}_x$ is a factor for all $x\in X$.
\item[(ii)] $\mathfrak{M}=\int_{x\in X} \mathfrak{M}_x d\mu(x)$.
\item[(iii)] $\mathfrak{Z(M)}=\calL^{\infty}(X,M,\mu)$.
\end{itemize}
This decomposition is essentially unique.
\end{Thm}

The set of projections $\mathfrak{P(M)}=\{p\in\mathfrak{M}: p^2=p^*=p\}$ of a von Neumann algebra $\mathfrak{M}$ corresponds to the non-commutative analogue of indicator functions, and plays a central rôle in the classification of factors:

\begin{Pro}
The set $\mathfrak{P(M)}$ of projections of a von Neumann algebra $\mathfrak{M}$ is a complete orthomodular lattice. Furthermore, this lattice generates $\mathfrak{M}$ in the sense that:
\begin{equation}
\mathfrak{M}=\mathfrak{P(M)}''.
\end{equation}
\end{Pro}

An important notion is that of equivalence between two projections:

\begin{Def}
Let $\mathfrak{M}$ be a von Neumann algebra on a Hilbert space $\calH$, and $p,q$ two projections of $\mathfrak{M}$. We denote:
\begin{itemize}
\item[(i)] $p\preceq q$ iff there exists an partial isometry $u\in\mathfrak{M}$ verifying $u^*u=p$ and $q-u^*u \geq 0$.
\item[(ii)] $p \sim q$ iff there exists a partial isometry $u\in\mathfrak{M}$ such that $p=u^*u$ and $q=uu^*$.
\end{itemize}
\end{Def}

One can then show that $\sim$ defines an equivalence relation on $\mathfrak{P(M)}$.

\begin{Thm}[Comparison theorem]
Let $p,q$ be two projections of a von Neumann algebra $\mathfrak{M}$. There exists a projection $z\in\mathfrak{Z(M)}$ such that $pz \preceq qz$ and $q(\id-z) \preceq p(\id-z)$.
\end{Thm}

\begin{Cor}
Let $\mathfrak{M}$ be a factor. $(\mathfrak{P(M)},\preceq)$ is an ordered set. Moreover, two factors $\mathfrak{M}_1$ and $\mathfrak{M}_2$ are isomorphic iff $(\mathfrak{P(M)}_1,\preceq)$ and $(\mathfrak{P(M)}_2,\preceq)$ are isomorphic as ordered spaces.
\end{Cor}

One can then show that classifying factors is equivalent to classifying the corresponding ordered spaces:

\begin{Def}
A projection is called:
\begin{itemize}
\item[(i)] Finite iff it does not admit equivalent subprojections, i.e.:
\begin{equation}
r\leq p \hspace{0.3cm}\mbox{and}\hspace{0.3cm} r\sim p \hspace{0.3cm} \Rightarrow \hspace{0.3cm} r=p.
\end{equation}
\item[(ii)] Infinite iff it is not finite.
\item[(iii)] Purely infinite iff for any $r$ finite and verifying $r\leq p$, then $r=0$.
\item[(iv)] Semi-finite iff it is not finite and is equal to the supremum of an increasing family of finite projections.
\item[(v)] Minimal iff $p\neq 0$ and $r\leq p\Rightarrow r=0$
\end{itemize}
A von Neumann algebra $\mathfrak{M}$ is called finite (respectively infinite, purely infinite, semi-finite) iff $\id \in \mathfrak{M}$ is a finite (resp. infinite, purely infinite, semi-finite) projection.
\end{Def}

\begin{Def}
Let $\mathfrak{M}$ be a factor. $\mathfrak{M}$ is called a:
\begin{itemize}
\item[(i)] Type I$_n$ factor iff $(\mathfrak{P(M)},\preceq)$ is isomorphic to $(\{0,1,...,n-1\},\leq)$.
\item[(ii)] Type I$_\infty$ factor iff $(\mathfrak{P(M)},\preceq)$ is isomorphic to $(\mathbb{N},\leq)$.
\item[(iii)] Type II$_1$ factor iff it is finite and $(\mathfrak{P(M)},\preceq)$ is isomorphic to $([0;1],\leq)$.
\item[(iv)] Type II$_\infty$ factor iff it is semi-finite and $(\mathfrak{P(M)},\preceq)$ is isomorphic to $([0;\infty[,\leq)$.
\item[(v)] Type III factor iff it is purely infinite and $(\mathfrak{P(M)},\preceq)$ is isomorphic to $(\{0;\infty\},\leq)$.
\end{itemize}
\end{Def}

\begin{Thm}
Any factor is either a type I$_n$, I$_\infty$, II$_1$, II$_\infty$ or type III factor. Moreover, there exists at least one factor of each type.
\end{Thm}

\section{Spectral theorem}

\begin{Def}
Let $\calH$ be a Hilbert space. A family $\{P_\lambda\}_{\lambda\in\mathbb{R}}$ of self-adjoint operators on $\calH$ is called a resolution of the identity of $\calH$ iff:
\begin{itemize}
\item[(i)] $P_\lambda P_\mu=P_{\min(\lambda,\mu)}$.
\item[(ii)] $P_\lambda=0$ if $\lambda$ is ``small enough", and $P_\lambda=\id$ if $\lambda$ is ``big enough".
\item[(iii)]  $\lim_{\mu\rightarrow\lambda^+}P_\mu(x)=P_\lambda(x), \forall x\in\calH$.
\end{itemize}
\end{Def}

If $\calH$ is a Hilbert space and $\{P_\lambda\}_{\lambda\in\mathbb{R}}$ is a resolution of the identity of $\calH$, then for all $x\in\calH$, the map:
\begin{align}
\begin{aligned}
\mathbb{R}&\rightarrow\mathbb{R}\\
\lambda &\mapsto \langle P_\lambda (x),x\rangle
\end{aligned}
\end{align}is vanishing around $-\infty$, equal to $||x||^2$ around $+\infty$, right-continuous and increasing.

\begin{Pro}
If $F:\mathbb{R}\rightarrow\mathbb{R}$ is a bounded, increasing, right-continuous function, equal to 0 on $]-\infty;m[$, and constant on $]M;+\infty[$, where $\{m,M\} \in\mathbb{R}$, then there exists a unique finite positive borelian measure $\mu$ on $\mathbb{R}$ with support in $[m;M]$ such that $\mu(]-\infty;\lambda])=F(\lambda),\forall \lambda\in\mathbb{R}$. Such a measure is called the Stieljes measure of $F$ and is denoted by $dF$.
\end{Pro}

For all $x\in\calH$, $\calH$ being a Hilbert space, the Stieljes measure associated to the map $\lambda\mapsto\langle P_\lambda(x),x \rangle$ is denoted by $d\langle P_\lambda(x),x\rangle$.

\begin{Pro}
Let $\calH$ be a Hilbert space, $\{P_\lambda\}_{\lambda\in\mathbb{R}}$ a resolution of the identity of $\calH$ and $f\in\mathfrak{C}(\mathbb{R},\mathbb{C})$. There exists a unique continuous operator $U\in\calL(H)$ such that:
\begin{equation}
\langle U(x),x \rangle=\int_{\lambda\in\mathbb{R}}f(\lambda) d\langle P_\lambda(x),x\rangle, \forall x\in\calH.
\end{equation}
This continuous operator is denoted by $U=\int_{\lambda\in\mathbb{R}}f(\lambda)dP_\lambda$, where $dP_\lambda$, defined by $\lang dP_\lambda(x),x\rang = d \lang P_\lambda(x),x\rang, \forall x \in \calH$, is called the projection-valued measure associated to the map $\lambda\mapsto\langle P_\lambda(x),x \rangle$.
\end{Pro}

The following theorem shows that all normal operators, i.e. operators commuting with their adjoint, are of the above form:

\begin{Thm}[Spectral theorem]
Let $\calH$ be a Hilbert space and $A \in \cal{L(H)}$ a normal operator. There exists a unique resolution of the identity $\{P_\lambda\}_{\lambda\in\mathbb{R}}$, called the spectral resolution of $A$, such that for all $f\in \mathfrak{C}(\sigma(A),\mathbb{C})$, one can define:
\begin{equation}
f(A)=\int_{\lambda\in\sigma(A)}f(\lambda)dP_\lambda.
\end{equation}
\end{Thm}

The restriction from normal operators to self-adjoint ones imposes that $f$ is real-valued. This theorem defines an isomorphism between the C$^*$-algebra (von Neumann algebra) generated by a normal element $A \in \mathfrak{L}(\calH)$ and continuous (measurable) functions on $\sigma(A)$. However, while the spectral resolution associated to the decomposition of $A$ necessarily belongs\footnote{This statement is connected to corollary \ref{cor-bicommutant}.} to the commutative von Neumann subalgebra generated by $A$, it may not belong to the commutative C$^*$-algebra generated by $A$.

\end{appendices}

\addcontentsline{toc}{chapter}{Bibliography}
\printbibliography

\end{document}